\documentclass[a4paper,reqno]{amsart}
\usepackage[utf8]{inputenc}

\usepackage[dvipsnames]{xcolor}
\usepackage[utf8]{inputenc}
\usepackage[english]{babel}
\usepackage{amsthm}
\usepackage{mathtools}
\usepackage{hyperref}
\usepackage{import}
\usepackage{amssymb}
\usepackage{cancel}
\definecolor{light-gray}{gray}{0.95}
\usepackage[color=light-gray]{todonotes}
\usepackage{dirtytalk}
\usepackage{blindtext}
\usepackage{tikz-cd}
\usepackage{float}
\restylefloat{table}
\usepackage{empheq}
\usepackage{indentfirst}
\usepackage{mathrsfs}

\allowdisplaybreaks

\newcommand{\MS}[1]{{\color{MidnightBlue}{#1}}}

\usepackage[
backend=biber,
style=alphabetic,
maxbibnames=99,
maxalphanames=5,
]{biblatex}

\newtheorem{theorem}{Theorem}[subsection]

\newtheorem{proposition}[theorem]{Proposition}
\newtheorem{proposition/definition}[theorem]{Proposition/Definition}
\newtheorem{lemma}[theorem]{Lemma}
\theoremstyle{definition}
\newtheorem{remark}[theorem]{Remark}
\newtheorem{definition}[theorem]{Definition}
\newtheorem{example}[theorem]{Example}

\renewcommand{\d}{\mathrm{d}} 
\newcommand{\dt}{\mathrm{d} t} 
\newcommand{\F}{\mathcal{F}}  
\renewcommand{\S}{\mathcal{S}}
\newcommand{\gh}{\text{gh}}
\newcommand{\ld}{\#}
\newcommand{\cfd}{\text{cfd}_M}
\newcommand{\bvbfv}{\mathfrak{BV\text{-}BFV}}

\newcommand{\tf}{\theta_1}        
\newcommand{\vpf}{\varpi_1}       
\newcommand{\Lf}{L_{1}}             
\newcommand{\Qf}{{Q_1}}             
\newcommand{\clSf}{{S_1}}           
\newcommand{\betaf}{\beta_1}        
\newcommand{\ff}{f_1}               
\newcommand{\Rf}{{R_1}}             
\newcommand{\Df}{D_1}             

\newcommand{\ts}{\theta_{2}}    
\newcommand{\vps}{\varpi_{2}}   
\newcommand{\Ls}{L_{2}}         
\newcommand{\Qs}{{Q_2}}         
\newcommand{\clSs}{{S_2}}               

\newcommand{\qf}{q}                 
\newcommand{\dqf}{\dot q}
\newcommand{\ddqf}{\ddot q}
\newcommand{\pf}{p}                 
\newcommand{\dpf}{\dot p}
\newcommand{\gf}{g}                 
\newcommand{\dgf}{\dot g}
\newcommand{\hf}{h}                 
\newcommand{\FfCM}{\F^\mathrm{lax}_{1CM}} 

\newcommand{\qs}{\tilde{q}}
\newcommand{\dqs}{\dot{\tilde{q}}}
\newcommand{\ddqs}{\ddot{\tilde{q}}}
\newcommand{\gs}{\tilde{g}}

\newcommand{\hs}{\tilde{h}}
\newcommand{\FsCM}{\F^\mathrm{lax}_{2CM}}

\newcommand{\Tr}{\mathrm{Tr}}
\newcommand{\es}{\varepsilon_s}

\newcommand{\Bf}{B}
\newcommand{\Af}{A}
\newcommand{\dAf}{\d_{\Af}}
\newcommand{\FAf}{F_{\Af}}
\newcommand{\cf}{c}
\newcommand{\FfYM}{\F^\mathrm{lax}_{1YM}}
\newcommand{\As}{\tilde{A}}
\newcommand{\dAs}{\d_{\As}}
\newcommand{\FAs}{\tilde{F}_{\As}}
\newcommand{\cs}{\tilde{c}}
\newcommand{\FsYM}{\F^\mathrm{lax}_{2YM}}

\newcommand{\qj}{\tilde{q}}
\newcommand{\dqj}{\dot{\tilde{q}}}
\newcommand{\ddqj}{\ddot{\tilde{q}}}
\newcommand{\xij}{\tilde{\xi}}
\newcommand{\dxij}{\dot{\tilde{\xi}}}

\newcommand{\Fj}{\F^\mathrm{lax}_{J}}
\newcommand{\tj}{\theta_{J}}
   
\newcommand{\Lj}{L_{J}}
\newcommand{\Qj}{{Q_{J}}}

\newcommand{\Tj}{T}

\newcommand{\Fgr}{\F^\mathrm{lax}_{GR}}
\newcommand{\tgr}{\theta_{GR}}
   
\newcommand{\Lgr}{L_{GR}}
\newcommand{\Qgr}{{Q_{GR}}}
\newcommand{\Dgr}{{D_{GR}}}
\newcommand{\Dgrk}{{D^k_{GR}}}
\newcommand{\Rgr}{{R_{GR}}}

\newcommand{\betagr}{\beta_{GR}}
\newcommand{\fgr}{f_{GR}} 

\title[BV equivalence with boundary]{BV equivalence with boundary}
\thanks{A. S. C. acknowledges partial support of SNF Grant No. 200020\_192080. This research was (partly) supported by the NCCR SwissMAP, funded by the Swiss National Science Foundation.}

\author{F. M. Castela Sim\~{a}o}
\address{Queen Mary University of London, School of Mathematical Sciences, Mile End Rd,London E1 4NS, United Kingdom}
\email{f.castelasimao@qmul.ac.uk}

\author{A. S. Cattaneo}
\address{Institut f\"ur Mathematik, Universit\"at Z\"urich, Winterthurerstrasse 190, 8057 Z\"urich, Switzerland}
\email{cattaneo@math.uzh.ch}

\author{M. Schiavina}
\address{Institute for Theoretical Physics and Department of Mathematics, ETH Zurich, R\"amistrasse 101, 8092, Z\"urich, Switzerland}
\curraddr{Dipartimento di Matematica, Università di Pavia, Via Ferrata 5, 27100 Pavia, Italy}
\email{michele.schiavina@unipv.it}

\begin{document}

\begin{abstract}
An extension of the notion of classical equivalence of equivalence in the Batalin--(Fradkin)--Vilkovisky (BV) and (BFV) framework for local Lagrangian field theory on manifolds possibly with boundary is discussed. Equivalence is phrased in both a strict and a lax sense, distinguished by the compatibility between the BV data for a field theory and its boundary BFV data, necessary for quantisation.
In this context, the first- and second-order formulations of non-Abelian Yang--Mills and of classical mechanics on curved backgrounds, all of which admit a strict BV-BFV description, are shown to be pairwise equivalent as strict BV-BFV theories. This in particular implies that their BV-complexes are quasi-isomorphic.  Furthermore, Jacobi theory and one-dimensional gravity coupled with scalar matter are compared as classically-equivalent reparametrisation-invariant versions of classical mechanics, but such that only the latter admits a strict BV-BFV formulation.  They are shown to be equivalent as lax BV-BFV theories and to have isomorphic BV cohomologies. This shows that strict BV-BFV equivalence is a strictly finer notion of equivalence of theories.
\end{abstract}

\maketitle

\tableofcontents
\section{Introduction}
\label{sec:Intro}
The notion of equivalence of field theories is one that can be found throughout physics. Such a concept is relevant and useful for various reasons. At a classical level, the equations of motions of a given theory might be easier to handle than other ``equivalent'' ones, even if they ultimately yield the same moduli space of solutions. Such reformulations often result in different and enlightening new interpretations of a given problem. Moreover, one theory might be better suited for quantisation than another, but the question of whether two classically equivalent theories result in the same quantum theory is in general still open. With this work, we attempt to take another step towards the answer.

The classical physical content of a given field theory is encoded in the set $\mathcal{EL}$ of solutions of the Euler--Lagrange equations. In the case where the theory in question also enjoys a local symmetry --- encoded by a tangent distribution $D$ --- we are interested in the moduli space of inequivalent solutions $\mathcal{EL}/D$. \emph{Classical observables} are then defined to be suitable functions on $\mathcal{EL}/D$. {Such a quotient is typically singular: defining a sensible space of functions over it becomes challenging and it is often more convenient to find a replacement; a problem best addressed within the Batalin--Vilkovisky (BV) formalism.}

The BV formalism was first introduced in \cite{batalin1983feynman,batalin1983quantisation,batalin1984gauge} as an extension of the BRST formalism \cite{becchi1976renormalisation,Tyutin2008}, named after Becchi, Rouet, Stora and Tyutin, used to quantise Lagrangian gauge theories in a way that preserves covariance. Around the same time, the Batalin-Fradkin-Vilkovisky (BFV) formalism was introduced, which deals with constrained Hamiltonian systems \cite{batalin1977relativistic,batalin1983generalized}. It was later noticed by various authors \cite{mcmullan1984constraints,henneaux1985hamiltonian,browning1987batalin,dubois1987systemes,fisch1990homological,henneaux1990lectures,mccloud1994jet,stasheff1997homological,stasheff1998secret} that the aforementioned formalisms enjoy a rich cohomological structure. For example, a BV theory associates a chain complex to a spacetime manifold, the \emph{BV complex}, which aims at a \emph{resolution} of the desired space of functions over the quotient $\mathcal{EL}/D$. In the case of the BFV formalism, one introduces the \emph{BFV complex} \cite{stasheff1997homological,schatz2009bfv,schatz2009coisotropic} as a resolution of the space of functions over the \emph{reduced phase space} of a given constrained Hamiltonian system.

One can then address the question of equivalence of theories in the BV setting. Following the discussion above, a natural way of comparing two classical theories is through their BV cohomologies, also called classical observables, as done for example in \cite{barnich1995local}. However, a BV theory comes equipped with several pieces of data other than the underlying dg-algebra structure (for example a symplectic structure and a Hamiltonian function) that one might want an equivalence relation to preserve. Finding the appropriate notion of \say{BV-equivalence} is thus a non-trivial open question.  In \cite{cattaneo2018bv,canepa2019fully}, a \emph{stronger} notion of BV-equivalence is implemented, which requires all data to be preserved by a symplectomorphism. A nontrivial example of such equivalence is found between 3d gravity and (nondegenerate) $BF$ theory. In \cite{canepa2021general} various alternative weaker notions of BV equivalence have been presented, which apply to higher dimensional formulations of General Relativity.

The BV and BFV approaches were linked by Cattaneo, Mnev and Resheti\-khin in \cite{cattaneo2014classical}, where the authors showed that a BV theory on the bulk induces a compatible BFV theory on the boundary, provided that some regularity conditions are met. The presence of a boundary will typically spoil the symmetry invariance of the BV data, encoded in the BV cohomology, but this failure will be controlled by the BFV data associated to the boundary. From this perspective, the regularity conditions can be seen as a compatibility condition between the BV complex on the bulk and the BFV complex on the boundary.

Derived geometry \cite{pantev2013shifted} extends the above setting to algebraic geometry, even though currently only in the restricted setting of AKSZ theories. The induced boundary theory is in this case an example of derived intersection, \cite{calaque2015lagrangian,calaque2017shifted,}. As derived geometry mainly addresses classical problems, the nondegeneracy of the symplectic form is only required up to homotopy, which yields problems in the direction of quantisation.

On the other hand the BV-BFV approach \cite{cattaneo2014classical,cattaneo2018perturbative} is especially successful since it allows for a quantisation procedure that is compatible with cutting and gluing. This has already been shown to work in various examples such as $BF$ theory \cite{cattaneo2018perturbative,cattaneo2020cellular}, split Chern--Simons theory \cite{cattaneo2017split}, 2D Yang-Mills theory \cite{iraso2019two} and AKSZ sigma models \cite{cattaneo2019globalization}.

This approach was first tested\footnote{Another approach to General Relativity by means of the BV formalism (without boundary) can be found in \cite{rejzner2011batalin,fredenhagen2012local}.} on General Relativity in \cite{schiavina2015bv}. For diffeomorphism invariant theories, the compatibility between bulk and boundary data becomes a non-trivial matter, and there are various cases where the regularity conditions necessary for the BV-BFV description fail to be met. Most notable are the examples of Palatini--Cartan gravity in (3+1) dimensions \cite{cattaneo2017bv}, Plebanski theory \cite{schiavina2015bv}, the Nambu--Goto string \cite{martinoli2020bvanalysis} and the Jacobi action for reparametrisation invariant classical mechanics \cite{cattaneo2017time}. On the other hand, the respectively classically equivalent Einstein--Hilbert formulation of gravity \cite{cattaneo2016bv} in (3+1) dimensions and the Polyakov action \cite{martinoli2020bvanalysis} fulfill the BV-BFV axioms. The question of how one can go around these problems and construct a sensible BV-BFV theory for Palatini--Cartan gravity was addressed in \cite{canepa2019boundary,canepa2021general}.

As not all field theories are suitable for a BV-BFV description, the \emph{lax approach} to the BV-BFV formalism was proposed in \cite{mnev2020towards}, which gathers the data prior to the step where the regularity conditions become relevant. This setting already allows us to construct the \emph{BV-BFV complex} \cite{mnev2020towards}, which is the adaptation of the BV complex to the case with boundary. Likewise, classical observables are contained in its cohomology. As such, the lax BV-BFV formalism offers a sensible way of comparing two field theories on manifolds with boundary, even if one does not have a \emph{strict} BV-BFV theory. 

In this paper, we provide an explicit method to lift classical equivalence to a (potential) BV equivalence, also in the presence of boundaries. This naturally introduces the notion of {lax equivalence of BV theories on manifolds with boundary, which is in principle finer than BV equivalence.} Our method is applied to the simple cases of classical mechanics on a curved background as well as to (non-abelian) Yang--Mills theory, where we explicitly show that the first- and second-order formalisms are lax BV-BFV equivalent (and hence BV-quasi-isomorphic).

We then turn our attention to the main objective of this paper: the analysis of the classically-equivalent Jacobi theory and one-dimensional gravity coupled to matter (1D GR). These two models can be regarded as the one-dimensional counterparts of the Nambu--Goto and Polyakov string models respectively, and they both represent a reparametrisation-invariant version of classical mechanics. In \cite{cattaneo2017time} it was shown that while 1D GR satisfies the regularity conditions of the BV-BFV formalism, Jacobi theory produces a singular theory on the boundary, and a similar result was proven for their 2d string-theoretic analogues \cite{martinoli2020bvanalysis} which raises the question of the origin of this boundary discrepancy.

By comparing the BV and BV-BFV cohomologies of Jacobi theory and 1D GR, we find that, even though the two theories on manifolds with possibly non-empty boundary are lax equivalent, and hence their associated BV (and lax BV-BFV) complexes are quasi-isomorphic, the chain maps that connect the two theories do not preserve the regularity condition required by the strictification procedure (Theorem \ref{theorem:chispoilsboundary}). 

In other words, quasi-isomorphisms of lax BV-BFV complexes do not preserve strict BV-BFV theories, which then should be taken as a genuine subclass of BV theories: even in the best case scenario of two theories that are classically equivalent with quasi-isomorphic lax BV-BFV complexes, an obstruction to their strict BV-BFV compatibility distinguishes the two. { Indeed, consider two lax equivalent theories (Definition \ref{def:laxequivalence} --- see e.g.\ the case described in Theorem \ref{theorem:laxequivalenceJacobi1DGR}) such that one of the two models fails to be compatible with the strict BV-BFV axioms (cf. Remark \ref{rem:strictification}). In this case only one of the two admits a quantisation in the BV-BFV setting. Even if they both could ultimately admit a sensible quantisation, our result suggests anyway that they might have different quantisations in the presence of boundaries.

Another way of viewing our result is the following. Suppose we are given a lax BV-BFV theory that is not strict (Definition \ref{def:LaxBVBFV} and Remark \ref{rem:strictification}). Can we find a quasi-isomorphic lax BV-BFV theory that is strictifiable? If so, we may think of the second theory as a good replacement for the first, suitable for quantisation with boundary.}

We should stress that the enrichment of the BV complex by the de Rham complex of the source manifold (in the sense of local forms) has been the object of past research (see among all \cite{barnich1995local}). {The lax BV-BFV complex we consider coincides with their Batalin--Vilkovisky--de Rham complex; however, our notion of lax equivalence is different (Definition \ref{def:laxequivalence}), as it requires the existence of chain maps that are quasi-inverse to one another and compatibile with the whole lax BV-BFV structure.}

Crucially, our approach diverges from other investigations of local field theory that only look at pre-symplectic data. The strictification step is precisely the pre-symplectic reduction of such data, and where the obstruction lies. We are not aware of a viable quantisation procedure for pre-symplectic structures.

This paper is structured as follows: 
Section \ref{sec:FieldtheoriesandEquivalence} is dedicated to a review of local Lagrangian field theory (Section \ref{sec:LagrangianFT}), which is followed by the BV formalism (Section \ref{sec:BVformalism}) and the BV-BFV and lax BV-BFV formalisms (Section \ref{sec:FieldTheoriesonManifoldsWithHigherStrata}). We will showcase several notions of equivalence in classical field theory, starting from Lagrangian field theory in Section \ref{sec:LagrangianFT}, while the discussion of equivalence in the BV and lax BV-BFV cases can be found in Sections \ref{sec:EquivalenceBV} and \ref{sec:EquivalenceLax} respectively. Later in Section \ref{sec:Examples} we discuss our general procedure to prove lax equivalence between two theories (Section \ref{sec:Strategy}) and three examples of such equivalence, namely
\begin{itemize}
    \item first- and second-order formulations of classical mechanics on a curved background (Section \ref{sec:CM}); 
    \item first- and second-order formulations of (non-abelian) Yang--Mills theory (Section \ref{sec:YM});
    \item one dimensional gravity coupled to matter and Jacobi theory (Section \ref{sec:1Dreparametrisationinvarianttheories}).
\end{itemize}

\textbf{Results and outlook:} We present our notion of BV equivalence (Definition \ref{def:BVequivalence}) for theories over closed manifolds and lax equivalence (Definition \ref{def:laxequivalence}) in the case of manifolds with higher strata, and show that the latter implies the former for the respective bulk (codimension-$0$ stratum) BV theories (Theorem \ref{theorem:laxEqimpliesBVEq}).

We then show lax equivalence for the aforementioned examples, in the sense that their lax BV-BFV data can be interchanged in a way that preserves their cohomological structure. In particular, we show that the respective BV-BFV complexes are quasi-isomorphic $$H^\bullet(\bvbfv^\bullet_1) \simeq H^\bullet(\bvbfv^\bullet_2).$$

Most notably, this means that the boundary discrepancy present in the BV-BFV formulations of Jacobi theory and 1D GR found in \cite{cattaneo2017time} does not have a cohomological origin, and is rather to be interpreted as an obstruction in pre-quantisation.

We expect the procedure to be applicable to other relevant examples of BV-BFV obstructions such as the Nambu-Goto and Polyakov actions \cite{martinoli2020bvanalysis} and, for a more challenging one, Einstein-Hilbert and Palatini-Cartan gravity in (3+1) dimensions, { whose extendibility as BV-BFV theories have been shown to differ in} \cite{cattaneo2016bv,cattaneo2017bv}. 

This obstruction, which bars certain theories from being quantisable in the BV formalism with boundary without additional requirements on the fields, suggests that, { even assuming that some quantum theory exists for both models, they might differ}. Alternatively, it might suggest that among various classically- and BV-equivalent models, there is a preferred choice for models which are BV-BFV compatible. Either way, these results call for additional investigations in this direction.

\section*{Acknowledgements}
We would like to thank G. Barnich, M. Grigoriev and M. Henneaux for instructive discussions on the topic of equivalence of field theories in the BV formalism, relevant to this work.

\section{Field theories and equivalence}
\label{sec:FieldtheoriesandEquivalence}

We start by presenting the field theoretical structures and objects used throughout this work, following \cite{deligne1999quantum, anderson1989variational}. Subsequently, we review the BV formalism for closed manifolds\footnote{For a discussion of the BV formalism in the setting of non-compact manifolds see \cite{rejzner2011batalin,fredenhagen2012batalin}. For the extension of the BV-BFV framework to manifolds with asymptotic boundary see \cite{RejznerSchiavina2020}.}  \cite{batalin1983feynman,batalin1983quantisation,batalin1984gauge} $-$ see also \cite{henneaux1990lectures,gomis1995antibracket,mnev2017lectures} $-$ and the BV-BFV formalism, its generalisation for manifolds with boundaries and corners \cite{cattaneo2014classical,cattaneo2020introduction}. As some theories we consider are not compatible with the BV-BFV axioms, we revise the lax BV-BFV formalism \cite{mnev2020towards}, which not only lets us study these cases, but presents a better stage for our discussions in the presence of boundaries and corners.

Moreover, this section is used to develop our notion of equivalence of field theories at every step of the way, first showcasing how we want to compare two classical field theories in Definition \ref{def:classicalequivalence} and adapting these considerations to the BV and lax BV-BFV formalisms in Definitions \ref{def:BVequivalence} and \ref{def:laxequivalence} respectively. 

\subsection{Lagrangian field theories}
\label{sec:LagrangianFT}

Let $M$ be a manifold of any dimension. In order to build a classical field theory on $M$ we need a space of fields ${\mathcal{E}}$, a local functional $S$ called the \emph{action functional} and \emph{local observables}. In most cases, we can achieve such a construction by considering a (possibly graded) fibre bundle $E\rightarrow M$ over $M$ and by defining the space of fields as its space of smooth sections ${\mathcal{E}} \coloneqq \Gamma(M,E)$ with coordinates $\varphi^i$. \emph{Local} objects can then be regarded as a subcomplex of the de Rham bicomplex $\Omega^{\bullet,\bullet}({\mathcal{E}}\times M)$, where \say{local} essentially means that these objects only depend on the first $k$ derivatives of the fields $\varphi^i$ (or the $k$-th jet). Let us make this notion precise:
\begin{definition}[{(Integrated)} local forms \cite{anderson1989variational}]
Let $E\rightarrow M$ be a (possibly graded) fibre bundle over $M$, ${\mathcal{E}} = \Gamma(M,E)$ its space of smooth sections, $J^k (E)$ the $k$-th jet bundle and $\{j^k \colon {\mathcal{E}} \times M \rightarrow J^k(E)\}$ the evaluation maps. We consider $j^\infty$ as the inverse limit of these maps and construct the \emph{infinite jet bundle} $J^\infty(E)$ as the inverse limit of the sequence
\begin{align*}
    E = J^0(E) \leftarrow J^1(E) \leftarrow \dots \leftarrow J^k(E) \leftarrow \dots
\end{align*}
The bicomplex\footnote{{Note that, strictly speaking, this works when $j^\infty$ is surjective. When $E$ has connected fibres, this is true if and only if $E$ admits a global section \cite[Proposition 3.1.14]{BlohmannLFT}.}} of \emph{local forms} on ${\mathcal{E}}\times M$ is defined as
\begin{align*}
    (\Omega^{\bullet,\bullet}_{\text{loc}}({\mathcal{E}}\times M),\delta,\d) \coloneqq (j^\infty)^*\left(\Omega^{\bullet,\bullet}(J^\infty(E)),\d_V,\d_H\right),
\end{align*}
where $\mathrm{d}_H$ and $\mathrm{d}_V$ are the {horizontal and vertical differentials on the variational complex for} $J^\infty(E)$, respectively. Let $\alpha \in \Omega^{\bullet,\bullet}(J^\infty(E))$. The differentials $\delta, \d$ are defined through
\begin{subequations}\begin{align}
    &\d (j^\infty)^* \alpha = (j^\infty)^* \d_H \alpha, \label{e:horizontaldifferential}\\
    &\delta (j^\infty)^* \alpha = (j^\infty)^* \d_V \alpha. \label{e:verticaldifferential}
\end{align}\end{subequations}

Elements of $\Omega^{0,\bullet}_{\text{loc}}({\mathcal{E}}\times M)$ will be called \emph{local functionals} on ${\mathcal{E}}\times M$.

{Whenever the manifold $M$ is compact, one can define the complex of \emph{integrated local $k$ forms $\Omega^\bullet_{\int}(\mathcal{E})$,} as the image of $\int_M \colon \Omega^{k,\mathrm{top}}_{\mathrm{loc}}(\mathcal{E}\times M) \to \Omega^k_{\int}(\mathcal{E})$ with the (variational) differential\footnote{{Explicitly, this is $\delta \int_M (j^\infty)^*\alpha := \int_M \delta (j^\infty)^*\alpha = \int_M (j^\infty)^*d_V\alpha$.}} $\delta$.}
\end{definition}

\begin{remark}[On various notions of local forms] \label{rem:LocalFormsBulksvsBoundary}
{Notice that, in some field-theory literature (see e.g.\ \cite{deligne1999quantum}), the term ``local form'' is often used to denote integrals over the manifold $M$ of elements of $\Omega^{\bullet, \mathrm{top}}_{\mathrm{loc}}(\mathcal{E}\times M)$, which instead we call \emph{integrated local forms}. 

When $M$ is not compact, integration comes with \emph{caveats}. One can either consider compactly supported sections or adopt the point of view of \cite{fredenhagen2012batalin}, where the Lagrangian density is tested against a compactly supported function. Alternatively, one can forgo integration and consider the following quotient
\begin{align*}
    \Omega^\bullet_\mathrm{loc}({\mathcal{E}})
    \coloneqq \Omega^{\bullet,\mathrm{top}}_\text{loc}({\mathcal{E}}\times M) / \d\Omega^{\bullet,\mathrm{top}-1}_\text{loc}({\mathcal{E}}\times M).
\end{align*}
In \cite[Page 21]{anderson1989variational}, the elements of $\Omega^\bullet_{\mathrm{loc}}(\mathcal{E})$ are called \emph{variational forms} when endowed with the induced vertical differential\footnote{\label{fnt:inducedvar} {It is possible to induce a differential coming from the variational bicomplex, by means of the \emph{interior Euler operator} (see e.g.\ \cite{anderson1989variational}). We will not be concerned with the details of this construction.}} $\delta_V$.

Clearly, if $M$ is closed, $(\Omega^\bullet_{\int}(\mathcal{E}),\delta)$ is isomorphic (as a complex) to $(\Omega^\bullet_{\mathrm{loc}}(\mathcal{E}),\delta_V)$. Indeed, let $f,g \in \Omega^{\bullet,\mathrm{top}}_\text{loc}({\mathcal{E}}\times M)$ and define $\mathscr{F} \coloneqq \int_M f$, $\mathscr{G} \coloneqq \int_M g$, their respective integrals over $M$. Then $\mathscr{F} =\mathscr{G}$ iff the difference $f - g$ is $\d$-exact
\begin{align*}
    \mathscr{F} - \mathscr{G}
    = \int_M (f-g)
    = \int_M \d(\dots)
    = 0,
\end{align*}
where we used that $M$ is closed in the last step. Hence, $\Omega_{\mathrm{loc}}^\bullet(\mathcal{E})$ can be taken as a replacement of integrated local forms in the noncompact case (assuming there still is no boundary).

If $M$ has a non-empty boundary $\partial M \neq \emptyset$, these considerations no longer hold: boundary terms become relevant. If, on the local densities side we can work with $\Omega^{\bullet,\bullet}_{\mathrm{loc}}({\mathcal{E}}\times M)$, we will see in Section \ref{sec:FieldTheoriesonManifoldsWithHigherStrata} what the consequences of integrating over boundaries bring about in field theory. }
\end{remark}

In addition to the previous construction, we will extensively use the following type of vector field:

\begin{definition} An \emph{evolutionary vector field} \cite{anderson1989variational} $X\in \mathfrak{X}_{\text{evo}}({\mathcal{E}})$ on ${\mathcal{E}}$ is a vector field on $J^\infty(E)$ which is vertical with respect to the projection $J^\infty(E) \rightarrow M$, such that
\begin{align*}
    [\mathcal{L}_X,\d] = 0,
\end{align*}
where $\mathcal{L}_X = [\iota_X,\delta]$ is the variational Lie derivative on local forms on ${\mathcal{E}}\times M$.
\end{definition}

We are now ready to define the notion of a classical field theory. {We will assume for simplicity that $M$ is compact, possibly with boundary}:

\begin{definition}
A \emph{classical field theory} on $M$ is a pair $({\mathcal{E}},S)$, consisting of a space of fields ${\mathcal{E}} = \Gamma(M,E)$\footnote{More generally, the space of fields is an affine space modeled on a space of sections; e.g., a space of connections. Even more generally, e.g., in the case of sigma models, one expands fields around a background field. It is the space of these perturbations that is a space of sections.} and an action functional $S \in \Omega^{0}_{\int}({\mathcal{E}})$.
\end{definition}

{Since $S=\int L$ for some local form $L$, applying the variational differential on $\mathcal{E}$ to $S$ is the same as applying $\delta$, defined in Equation \ref{e:verticaldifferential}, to $L$ and integrating.} {This yields} two terms: 
\begin{align*}
    \delta S = \text{EL} + \text{BT}.
\end{align*}
{The term $\text{EL}$ is an integrated local 1-form}\footnote{{The integrand of $\text{EL}$ is the pullback along $j^\infty$ of a form of source type in the variational bicomplex, see \cite[Definition 3.5]{anderson1989variational}.}} on ${\mathcal{E}}$, whose vanishing locus defines the Euler-Lagrange equations $\text{EL} = 0$. The space where these are satisfied is called the \emph{critical locus}, the zero locus $\mathcal{EL}:=\mathrm{Loc}_0(\mathrm{EL}) \subset F$, and its elements are called \emph{classical solutions}.  
The term {$\text{BT}$} is a boundary term (i.e.\ an integral over $\partial M$, {when not empty}), which will be crucial for the construction of field theories on manifolds with boundary (cf.\ Section \ref{s:constructingBVBFV}).

A further important aspect of field theories is the notion of \emph{(gauge) symmetries}, which are transformations that leave the action functional $S$ and the critical locus $\mathcal{EL}$ invariant. Infinitesimally, they can be described as follows:

\begin{definition}
An \emph{infinitesimal local symmetry} of a classical field theory $({\mathcal{E}},S)$ is given by a distribution $D \subseteq T\mathcal{E}$, such that\footnote{\label{fnt:commentednote}Notice that we want $D$ to be a (generically proper) subspace of all vector fields that annihilate the action functional. We want it to be maximal, in the sense that all symmetries are considered except \emph{trivial} ones, i.e.\ those that vanish on $\mathcal{EL}$. As such it is not automatically a subalgebra. See \cite[Section 1.3]{henneaux1990lectures}. {Observe that, although not necessary, one might want to restrict $D$ to only (genuinely) local symmetries, meaning that we do not consider constant Lie group/Lie algebra actions at this stage.}}
\begin{align*}
    \mathcal{L}_X S = 0 \qquad \forall X \in \Gamma({\mathcal{E}},D).
\end{align*}
Furthermore, we require $D$ to be involutive on the critical locus $\mathcal{EL}$, i.e.\ if $X,Y \in \Gamma({\mathcal{E}},D)$, then $[X,Y]\big \vert_\mathcal{EL} \in \Gamma({\mathcal{E}},D \vert_\mathcal{EL})$.
\end{definition}

Whenever local symmetries are present, the space of interest is not $\mathcal{EL}$ but rather the space of inequivalent configurations $\mathcal{EL}/D$\footnote{By abuse of notation, we denote by $D$ also the restriction of $D$ to $\mathcal{EL}$.}, i.e.\ the space of orbits of $D$ on the critical locus $\mathcal{EL}$. \emph{Classical observables} are then suitable functions over $\mathcal{EL}/D$, whose space we denote by $C^\infty(\mathcal{EL}/D)$. Note that, as a quotient, $\mathcal{EL}/D$ is often singular and defining $C^\infty(\mathcal{EL}/D)$ is a non-trivial task. One way of handling this is to build a \emph{resolution} of $C^\infty(\mathcal{EL}/D)$, by means of the Koszul--Tate--Chevalley--Eilenberg complex, also known as the BV complex (see Definition \ref{def:BVcomplex}).

We are interested in analyzing to what extent two field theories are equivalent. Starting the discussion of equivalence in the setting of classical Lagrangian field theory, we consider the
\begin{definition}
\label{def:classicalequivalence}
Let $({\mathcal{E}}_{i}, S_i)$, $i\in\{1,2\}$, be two classical field theories with symmetry distributions $D_i$.  We say that $({\mathcal{E}}_{i}, S_i)$ are \emph{classically equivalent} if 
\begin{align*}
    &\mathcal{EL}_1 \simeq \mathcal{EL}_2,\\
    &D_1\vert_{\mathcal{EL}_1} \simeq D_2\vert_{\mathcal{EL}_2}.
\end{align*}
\end{definition}

\begin{remark}
If two theories are classically equivalent, we have $\mathcal{EL}_1/D_1 \simeq \mathcal{EL}_2/D_2$. If we have a model for the respective spaces of classical observables, they are isomorphic:
\begin{align*}
    C^\infty(\mathcal{EL}_1/D_1) \simeq C^\infty(\mathcal{EL}_2/D_2).
\end{align*}
This notion will be central in our discussion, and we will provide a refinement of it within the BV formalism, with and without boundary.
\end{remark}

\begin{remark}
In certain cases we can find $C_1\subset {\mathcal{E}}_1$, defined as the set of solutions of \emph{some} of the equations of motion $\text{EL}_1 = 0$.
Then, if we can find an isomorphism $\phi_{cl}: C_1 \to {\mathcal{E}}_2$ such that 
\begin{align*}
    &S_1\vert_{C_1} = \phi_{cl}^*S_2&
    &\text{and}& 
    &D_1\vert_{C_1} \stackrel{\phi_{cl}}{\simeq} D_2,
\end{align*}
the theories are classically equivalent. This is a simple example of the situation in which two theories are classically equivalent because they differ only by \emph{auxiliary fields} (see e.g. \cite{barnich1995local}).
\end{remark}


\subsection{Batalin-Vilkovisky formalism}
\label{sec:BVformalism}
The BV formalism is a cohomological approach to field theory, that allows one to characterise the space of inequivalent field configurations by means of the cohomology of an appropriate cochain complex. It turns out that it also provides a natural notion of equivalence of field theories, that also takes into account ``observables'' of the theory.

In this setting, a classical field theory is described through the following data: 
\begin{definition}
A \emph{BV theory} is the assignment of a quadruple $\mathfrak{F} = (\F,\omega,\S,Q)$ to a closed manifold $M$ where 
\begin{itemize}
    \item ${\F = \Gamma(M,F)}$ is the space of smooth sections of a $\mathbb{Z}$-graded bundle\footnote{For simplicity, in this note we assume that the Grassmann parity of a variable is equal to the parity of its $\mathbb{Z}$-degree. This is okay as long as we only consider theories without fermionic physical fields.} ${F\rightarrow M}$ (the \emph{BV space of fields}),
    \item $\omega \in \Omega^2_\mathrm{{\int}}(\F)$ is {an integrated, local, symplectic form} on $\F$ of degree $-1$ (the \emph{BV form}),
    \item $\S \in \Omega^0_\mathrm{{\int}}(\F)$ is an {integrated, local, functional} on $\F$ of degree $0$ (the \emph{BV action functional}),
    \item $Q \in \mathfrak{X}_\mathrm{evo}(\F)$ is a cohomological, evolutionary, vector field of degree 1, i.e.\ $[Q,Q] = 2Q^2 = 0$, and $[\mathcal{L}_Q,\d] = 0$,
\end{itemize}
such that 
\begin{align}
    \label{eq:BVequation}
    \iota_Q \omega = \delta \S.
\end{align}
The internal degree of $\F$ is called the \emph{ghost number} and will be denoted by $\text{gh}(\cdot)$.
\end{definition}

\begin{remark}
In principle, we only need to consider either $\S$ or $Q$, as they are related to one another through Equation \eqref{eq:BVequation}, apart from the ambiguity of an additive constant in $\S$. We will nonetheless regard them as separate data for later convenience, {as} we will see that introducing a boundary spoils Equation \eqref{eq:BVequation}.
\end{remark}

As $Q$ is cohomological, its Lie derivative $\mathcal{L}_{Q}$ is a differential on $\Omega^{\bullet}_{\text{loc}}(\mathcal{F})$, since $\mathrm{gh}(\mathcal{L}_Q) = 1$ and $2 \mathcal{L}_Q^2 = [\mathcal{L}_Q,\mathcal{L}_Q] = \mathcal{L}_{[Q,Q]} = 0$. In this context, $\mathcal{L}_Q$-cocycles are interpreted as (gauge-)invariant local forms.

\begin{remark}
\label{rem:LQomegaCME}
It is easy to gather that both $\omega$ and $\S$ are $\mathcal{L}_Q$-cocycles by applying $\delta$ and $\iota_Q$ to Equation \eqref{eq:BVequation} respectively. We have
\begin{subequations}
\label{eq:QpreservesBVformBVaction}
\begin{align}
    &\mathcal{L}_Q \omega = 0,\\
    \label{eq:CME}
    &\mathcal{L}_Q \S = (\S,\S)= 0.
\end{align}
\end{subequations}
where $(\cdot,\cdot)$ is the Poisson bracket induced by $\omega$. Equation \eqref{eq:CME} is known as the \emph{Classical Master Equation} \cite{batalin1984gauge,schwarz1993geometry}, and encodes the property that $\S$ is gauge invariant. In particular, Equations \eqref{eq:QpreservesBVformBVaction} mean that we have the freedom to perform the transformations $\omega \mapsto \omega + \mathcal{L}_Q(\dots)$ and $\S \mapsto \S + \mathcal{L}_Q(\dots)$, as long they preserve Equation \eqref{eq:BVequation}.
\end{remark}

\begin{definition}
\label{def:BVcomplex}
We define the \emph{BV complex} of a given BV theory $\mathfrak{F}$ as the {space of integrated} local forms on $\F$ endowed with the differential $\mathcal{L}_Q$
\begin{align*}
    \mathfrak{BV}^{\bullet} \coloneqq \left(\Omega^\bullet_{{\int}}(\F),\mathcal{L}_Q \right),
\end{align*}
where the grading on $\mathfrak{BV}^{\bullet}$ is given by the ghost number. Its cohomology will be denoted by $H^\bullet(\mathfrak{BV}^\bullet)$ and called the \emph{BV cohomology}.
\end{definition}
While the BV complex $\mathfrak{BV}^\bullet$ consists of inhomogeneous local forms (inhomogeneous also in ghost number), its 0-form part\footnote{In the literature the terminology ``BV complex" is used to denote $\mathfrak{BV}^\bullet_0$. We use the same name for $\mathfrak{BV}^\bullet$ as it is the natural extension in the present setting.} $\mathfrak{BV}^{\bullet}_0 \subset \mathfrak{BV}^\bullet$ is of interest as it is a \emph{resolution} of {$D$-invariant functionals on $\mathcal{EL}$ or,} when the quotient is nonsingular, {functionals on $\mathcal{EL}/D$} in the sense that the BV cohomology is given by\footnote{Counterexamples to this scenario have been observed \cite{getzlerSpinning,getzlerSpinningCurved}. In local field theory, the request that the BV complex be a proper resolution of the moduli space of the theory is generally too strong. Hence, we do not insist on the vanishing of negative cohomology.} \cite{henneaux1990lectures,stasheff1998secret,fredenhagen2012batalin}
\begin{align*}
    &H^{-i}(\mathfrak{BV}^{\bullet}_0) = 0& &\text{for } i>0,\nonumber\\
    &H^0(\mathfrak{BV}^{\bullet}_0) \simeq
    {C^\infty(\mathcal{EL}/D)}.& &
\end{align*}

\begin{example}[Lie algebra case \cite{batalin1984gauge}, see also \cite{henneaux1990lectures,mnev2017lectures}]
\label{ex:BRSTcase}
In this paper we will only consider examples which enjoy symmetries that come from a Lie-algebra action. Let $({\mathcal{E}},S)$ be a classical field theory over a closed manifold $M$ with a symmetry on ${\mathcal{E}}$ given by the action of a Lie algebra $(\mathfrak{g},[\cdot,\cdot])$. We can build a BV theory as follows: choose the space of fields to be
\begin{align*}
    \F = T^*[-1]({\mathcal{E}} \times \Omega^0(M,\mathfrak{g})[1]),
\end{align*}
with local coordinates $\Phi^i = (\varphi^j,\xi^a)$ on the base ${\mathcal{E}} \times \Omega^0(M,\mathfrak{g})[1]$ and $\Phi^\dagger_i = (\varphi^\dagger_j,\xi^\dagger_a)$ on the fibers. Usually one calls $\varphi^j$ the \emph{fields}, $\xi^a$ the \emph{ghosts}\footnote{In the case of Yang--Mills theory, the ghost field will be denoted as $c$.} and $\Phi^\dagger_i$ the \emph{antifields}. Note that the ghost numbers are related by $\mathrm{gh}(\Phi^i) + \mathrm{gh}(\Phi^\dagger_i) = - 1$ due to the -1 shift on the fibers. We take the BV form to be the canonical symplectic form on $\F$
\begin{align*}
    \omega = \int_M \langle\delta \Phi^\dagger, \delta \Phi\rangle,
\end{align*}
where $\langle\cdot,\cdot\rangle$ is a bilinear map with values in $\Omega^{\bullet,\mathrm{top}}_\mathrm{loc}(M)$. In the case of a Lie algebra action, the cohomological vector field $Q$ decomposes into the Chevalley-Eilenberg differential $\gamma$ and the Koszul-Tate differential $\delta_{KT}$
\begin{align*}
    Q = \gamma + \delta_{KT}.
\end{align*}
The action of $\gamma$ is defined on the fields and ghosts as
\begin{align*}
    &\gamma \varphi^j 
    = \xi^a v^j_a,&
    &\gamma \xi^a
    = \frac{1}{2} [\xi,\xi]^a,
\end{align*}
where $v^i_a$ are the fundamental vector fields of $\mathfrak{g}$ on $F$. In turn, $\delta_{KT}$ acts as
\begin{align}
    &\delta_{KT} \varphi^i = 0,& 
    &\delta_{KT} \xi^a = 0,\nonumber\\
    &\delta_{KT} \varphi^\dagger_i = \frac{\delta S}{\delta \varphi^i},& 
    &\delta_{KT} \xi^\dagger_a = v^i_a \varphi^\dagger_i.
    \label{eq:KoszulTate}
\end{align}
The BV action functional can then be constructed as an extension of the classical action functional
\begin{align*}
    \S[\Phi,\Phi^\dagger] = S[\varphi] + \int_M \langle\Phi^\dagger, Q \Phi\rangle
\end{align*}
and $Q(\cdot) = (\S,\cdot)$ can be used to compute the full form of $Q\Phi^\dagger_i$. The data $(\F, \omega, \S, Q)$ form a BV theory.
\end{example}

\subsection{Equivalence in the BV setting} 
\label{sec:EquivalenceBV}
We now have all the necessary tools to develop a notion of equivalence in the BV formalism. We are interested in comparing the BV data and cohomology $H^\bullet(\mathfrak{BV}^\bullet_i)$ of two BV theories $\mathfrak{F}_i$, $i \in\{1,2\}$. We recall that a \emph{quasi-isomorphism} is a chain map between chain complexes which induces an isomorphism in cohomology. In this spirit we define:

\begin{definition}
\label{def:BVequivalence}
Two BV theories $\mathfrak{F}_1$ and $\mathfrak{F}_2$ are \emph{BV-equivalent} if there is a {(degree-preserving) map} $\phi\colon \F_2 \rightarrow \F_1$ {that} induces a quasi-isomorphism $\phi^*\colon \mathfrak{BV}^\bullet_1 \rightarrow \mathfrak{BV}^\bullet_2$ of BV complexes, such that $\phi^*$ preserves the cohomological classes of the BV form and BV action functional as
\begin{align}
    \label{eq:BVequivalencetransformation}
    &\phi^* [\omega_1] = [\omega_2],&
    &\phi^* [\S_1] = [\S_2].
\end{align}
{A BV equivalence is called \emph{strong} iff $\phi$ is a symplectomorphism that preserves the BV action functionals.}
\end{definition}

\begin{remark}\label{rem:compchainmaps}
If $\mathfrak{F}_1$, $\mathfrak{F}_2$ are BV-equivalent, we can find a morphism $\psi:  \F_1 \rightarrow \F_2$, such that its pullback map $\psi^*$ is the quasi-inverse of $\phi^*$. In particular, the composition maps
\begin{align*}
    &\chi^* = \psi^* \circ \phi^*: \mathfrak{BV}^\bullet_1 \rightarrow \mathfrak{BV}^\bullet_1,&
    &\lambda^* = \phi^* \circ \psi^*: \mathfrak{BV}^\bullet_2 \rightarrow \mathfrak{BV}^\bullet_2,
\end{align*}
are the identity in the respective BV cohomologies $H^\bullet(\mathfrak{BV}^\bullet_1)$, $H^\bullet(\mathfrak{BV}^\bullet_2)$. This is equivalent to the existence of two maps $h_\chi:\mathfrak{BV}^\bullet_{1} \rightarrow \mathfrak{BV}^\bullet_{1}$, $h_\lambda:\mathfrak{BV}^\bullet_{2} \rightarrow \mathfrak{BV}^\bullet_{2}$ of ghost number $-1$ such that \cite{weibel1995introduction}
\begin{align*}
    &\chi^* - \mathrm{id}_{1} = \mathcal{L}_{Q_1} h_\chi + h_\chi \mathcal{L}_{Q_1},&
    &\lambda^* - \mathrm{id}_2 = \mathcal{L}_{Q_2} h_\lambda + h_\lambda \mathcal{L}_{Q_2}.
\end{align*}

Furthermore, note that applying $\psi^*$ to Equation \eqref{eq:BVequivalencetransformation} yields
\begin{align}
    &\psi^* [\omega_2] = [\omega_1],&
    &\psi^* [\S_2] = [\S_1].
    \label{eq:BVequivalencetransformationpsi}
\end{align}
\end{remark}

Let us now explore some direct implications of Definition \ref{def:BVequivalence}, in particular that the transformation of $\omega_i$ and $\S_i$ are not independent:
\begin{proposition}\label{eq:simplifiedhamiltoniancondition}
Rewrite Equations \eqref{eq:BVequivalencetransformation} and \eqref{eq:BVequivalencetransformationpsi} as 
\begin{align*}
    &\phi^* \omega_1 = \omega_2 + \mathcal{L}_{Q_2}\rho_2,&
    &\psi^* \omega_2 = \omega_1 + \mathcal{L}_{Q_1}\rho_1,\nonumber\\
    &\phi^* \S_1 = \S_2 + \mathcal{L}_{Q_2}\sigma_2,&
    &\psi^* \S_2 = \S_1 + \mathcal{L}_{Q_1}\sigma_1,
\end{align*}
with $\rho_i \in \Omega^2_{{\int}}(\F_i)$, $\sigma_i \in \Omega^0_{{\int}}(\F_i)$. Then
\begin{align}
        \label{eq:conditiontopreserveHamiltoniancondition}
        \mathcal{L}_{Q_i} (\iota_{Q_i} \rho_i + \delta \sigma_i) = 0.
\end{align}
Moreover, Equation \eqref{eq:conditiontopreserveHamiltoniancondition} is satisfied if
\begin{align}\label{eq:conditiontopreserveHamiltonianconditionSimp}
    &\rho_i = - \delta \mu_i,&
    &\sigma_i = \iota_{Q_i} \mu_i
\end{align}
with $\mu_i \in \Omega^1_{{\int}}(\F_i)$.
\end{proposition}

\begin{proof}
Applying $\phi^*$ to $\iota_{Q_1} \omega_1 = \delta \S_1$ yields
\begin{align*}
        &&\iota_{Q_2} \omega_2 + \iota_{Q_2}\mathcal{L}_{Q_2}\rho_2 &= \delta \S_2 + \delta \mathcal{L}_{Q_2}\sigma_2&&\\
        &\Rightarrow& \mathcal{L}_{Q_2}(\iota_{Q_2}\rho_2 + \delta \sigma_2) &= 0, &&
\end{align*}
and analogously $\mathcal{L}_{Q_1}(\iota_{Q_1}\rho_1 + \delta \sigma_1) = 0$.

The simplified condition \eqref{eq:conditiontopreserveHamiltonianconditionSimp} implies Equation \eqref{eq:conditiontopreserveHamiltoniancondition} since
\begin{align*}
        \mathcal{L}_{Q_{i}}(\iota_{Q_{i}}\rho_{i} + \delta \sigma_{i})
        = \mathcal{L}_{Q_{i}}(-\iota_{Q_{i}}\delta \mu_i + \delta \iota_{Q_{i}} \mu_i)
        = -\mathcal{L}^2_{Q_{i}}\mu_i = 0,
\end{align*}
where we used $\mathcal{L}_{Q_i} = [\iota_{Q_i},\delta]$.
\end{proof}

\begin{remark}
\label{rem:Pairwiselaxequivalent}
Let $\mathfrak{F}_1$ and $\mathfrak{F}_2$ be BV-equivalent theories as per Definition \ref{def:BVequivalence}, and let $\chi^*: \mathfrak{BV}^\bullet_1 \rightarrow \mathfrak{BV}^\bullet_1$ and $\lambda^*: \mathfrak{BV}^\bullet_2 \rightarrow \mathfrak{BV}^\bullet_2$ be the chain maps defined in Remark \ref{rem:compchainmaps}. Then, the theories $\mathfrak{F}_1$ and $\chi^*\mathfrak{F}_1 \coloneqq (\F_1, \chi^*\omega_1, \chi^*\S_1, Q_1)$ are clearly BV-equivalent, and so are $\mathfrak{F}_2$ and $\lambda^*\mathfrak{F}_2$. 
\end{remark}

\begin{remark}
{In the literature there exists another notion of equivalence of BV theories, based on what is usually called \emph{elimination of (generalised) auxiliary fields} or \emph{reduction of contractible pairs} (see e.g.\ \cite{Henneaux90Auxiliary,barnich1995local} and \cite{BarnichGrigoriev11} for a review). } {When two} theories differ by auxiliary fields content, {they} have the same BV cohomology. {In Section \ref{sec:Cp} we show how the presence of auxiliary fields leads to the process of elimination of cohomologically contractible pairs by explicitly constructing chain maps that are quasi inverse to one another and homotopic to the identity. Hence, theories that differ by auxiliary fields/contractible pairs are BV equivalent in the sense of Definition \ref{def:BVequivalence}.}
\end{remark}


\subsection{Field theories on manifolds with higher strata}
\label{sec:FieldTheoriesonManifoldsWithHigherStrata}

The BV formalism can be extended to the case where the underlying manifold $M$ has a non-empty boundary $\partial M \neq \emptyset$, as presented in \cite{cattaneo2014classical}. This construction relies on the BFV formalism introduced in \cite{batalin1983generalized}, see also \cite{stasheff1997homological,schatz2009bfv,schatz2009coisotropic}. 

\begin{definition}
An \emph{exact BFV theory} over a manifold $\Sigma$ is a quadruple $\mathfrak{F}^\partial = (\F^\partial,\omega^\partial,\S^\partial,Q^\partial)$ where
\begin{itemize}
    \item ${\F^\partial = \Gamma(\Sigma,F^\partial)}$ is the space of smooth sections of a $\mathbb{Z}$-graded fibre bundle ${F^\partial\rightarrow M^\partial}$,
    \item $\omega^\partial = \delta \alpha^\partial \in \Omega^2_{{\int}}(\F^\partial)$ is an exact, integrated, local, symplectic form on $\F^\partial$ of degree $0$,
    \item $\S^\partial \in \Omega^0_{{\int}}(\F^\partial)$ is a degree 1, integrated, local, functional on $\F^\partial$,
    \item $Q^\partial \in \mathfrak{X}_\mathrm{evo}(\F^\partial)$ is a degree 1, cohomological, evolutionary, vector field, i.e.\ $[Q^\partial,Q^\partial] = 0$, and $[\mathcal{L}_{Q^\partial},\d] = 0$
\end{itemize}
such that $Q^\partial$ is the Hamiltonian vector field of $\S^\partial$
\begin{align*}
    \iota_{Q^\partial} \omega^\partial = \delta \S^\partial.
\end{align*}
We call $\omega^\partial,S^\partial$ the \emph{boundary form} and \emph{boundary action functional} respectively.
\end{definition}
\begin{definition}
A \emph{BV-BFV theory} over a manifold $M$ with boundary $\partial M$ is given by the data 
\begin{align*}
    (\F, \omega, \S, Q, \F^\partial, \omega^\partial,  \S^\partial, Q^\partial, \pi)
\end{align*}
where $(\F^\partial, \omega^\partial, \S^\partial,Q^\partial)$ is an exact BFV theory over $\Sigma=\partial M$ and $\pi: \F \rightarrow \F^\partial$ is a surjective submersion such that
\begin{align}
    \label{eq:BVBFVequation}
    \iota_Q \omega = \delta \S + \pi^* \alpha^\partial
\end{align}
and $Q \circ \pi^* = \pi^* \circ Q^\partial$.
\end{definition}

\begin{remark}
\label{rem:mCME}
Equation \eqref{eq:BVBFVequation} implies that in general $\omega$ and $\S$ are no longer $\mathcal{L}_Q$-cocycles in the presence of a boundary. Instead we have \cite{cattaneo2014classical}
\begin{subequations}
\label{eq:FailureOmegaSLQcocycles}
\begin{align}
    &\mathcal{L}_{Q} \omega = \pi^* \omega^\partial,\\
    &\mathcal{L}_{Q} \S = \pi^*(2 \S^\partial - \iota_{Q^\partial}\alpha^\partial).
    \label{eq:FailureSLQcocycles}
\end{align}
\end{subequations}
Note that the failure of the structural BV forms to be $\mathcal{L}_Q$-cocycles is controlled by (boundary) BFV forms. In particular, Equation \eqref{eq:FailureSLQcocycles} means that $\S$ fails to be gauge invariant, and the right hand side can be related to Noether's generalised charges \cite{RejznerSchiavina2020}. Furthermore, the CME no longer holds. Instead we have the \emph{modified Classical Master Equation} \cite{cattaneo2014classical} 
\begin{align*}
    \frac12 \iota_Q \iota_Q \omega = \pi^* \S^\partial.
\end{align*}
\end{remark}

\subsection{Inducing boundary BFV from bulk BV data}
\label{s:constructingBVBFV}
It is important to emphasize how one can try to construct a boundary theory $\mathfrak{F}^\partial$ from a BV theory $\mathfrak{F}$, since there might be obstructions. The problem we want to address is that of inducing an exact BFV theory on the boundary $\partial M$, starting from the BV data assigned to the bulk manifold $M$. 

Define $\check \alpha $ as:
\begin{align}
    \label{eq:alphacheck}
    \check \alpha \coloneqq  \iota_Q \omega  - \delta \S.
\end{align}
By restricting the fields of $\F$  (and their normal jets) to the boundary $\partial M$ we can define the \emph{space of pre-boundary fields} $\check\F^\partial$ and endow it with a \emph{pre-boundary} 2-form $\check \omega = \delta \check \alpha$. Usually $\check \omega$ turns out to be degenerate. In order to define a symplectic space of boundary fields, one then has to perform symplectic reduction, see, e.g., \cite{da2008lectures}. Let $\ker \check \omega = \{X \in \mathfrak{X}(\mathcal{\check F^\partial}) \left| \iota_X \check \omega = 0 \right.\}$ and set
\begin{align}
    \F^\partial \coloneqq \check\F^\partial/\ker \check \omega.
    \label{eq:quotientboundary}
\end{align}
Since we are taking a quotient, nothing guarantees that $\F^\partial$ is smooth, but we want to assume that this is the case. However, a necessary condition for smoothness is that $\ker \check \omega$ has locally constant dimension, i.e.\ it is a subbundle of $T\check \F^\partial$. As we will see, this condition is not always satisfied, namely that there is a unique symplectic form $\omega^\partial$ such that $\pi^*\omega^\partial = \check\omega$, and a unique cohomological vector field $Q^\partial$ such that $Q \circ \pi^* = \pi^* \circ Q^\partial$. We assume (although this may not be true in general) that there is a 1-form $\alpha^\partial$ such that $\pi^*\alpha^\partial=\check\alpha$. Note that, in this case, $\alpha^\partial$ is unique and $\omega^\partial=\delta\alpha^\partial$. See \cite{cattaneo2020introduction} for details. However, for $\mathcal{F}^\partial$ smooth, we have the surjective submersion $\pi:\F \rightarrow \F^\partial$.

Consider now 
\begin{definition}
The \emph{graded Euler vector field} $E \in \mathfrak{X}_{\text{evo}}(\F)$ is defined as the degree 0 vector field which acts on local forms of homogeneous ghost number as
\begin{align*}
    \mathcal{L}_E F = \gh(F) F.
\end{align*}
Similarly we have $E^\partial = \pi_* E \in \mathfrak{X}_{\text{evo}}(\F^\partial)$ on the boundary.
\end{definition}

The cohomological vector field $Q^\partial$ is actually Hamiltonian and the corresponding boundary action functional can be computed as \cite{roytenberg2007aksz}
\begin{align}
    \S^\partial = \iota_{E^\partial}\iota_{Q^\partial} \omega^\partial.
    \label{eq:boundaryaction}
\end{align}
The data $\mathfrak{F}^\partial = (\F^\partial, \omega^\partial, \S^\partial, Q^\partial)$ define an exact BFV manifold over the boundary $\partial M$. For completeness, we also define the \emph{pre-boundary action functional} $\check S \coloneqq \pi^*S^\partial$. Pulling back Equation \eqref{eq:boundaryaction} via $\pi^*$ yields
\begin{align}
    \check \S = \iota_{E}\iota_{Q} \check \omega.
    \label{eq:checkS}
\end{align}
Note that by taking Equations \eqref{eq:alphacheck} and \eqref{eq:checkS}, we ensure that the data $(\check \alpha, \check S)$ can always be defined, even if the quotient in Equation \eqref{eq:quotientboundary} does not yield a smooth structure.

\begin{remark}\label{rem:n-ext}
The procedure we just presented can be repeated in case that the manifold $M$ not only has a boundary but also corners (higher strata), as presented in \cite{cattaneo2014classical,mnev2020towards}. If this is possible up to codimension $n$, then we call the theory a \emph{$n$-extended (exact) BV-BFV theory}.
\end{remark}

\begin{remark}
The quantization program introduced in \cite{cattaneo2018perturbative} relies on the BV-BFV structure of a given classical theory. As such, even if two theories are classically equivalent, only one might turn out to have a BV-BFV structure and so be suitable for quantization, as we now explore in the example of the Jacobi theory and 1D GR.
\end{remark}


Remark \ref{rem:mCME} and Section \ref{s:constructingBVBFV} discuss two potential roadblocks for our construction of equivalence in the presence of boundaries and corners (and more generally codimension-$k$ strata). First, to extend the notion of equivalence discussed in Section \ref{sec:EquivalenceBV} to the case with higher strata, we wish to capture the possibility of local forms being $\mathcal{L}_Q$-cocycles up to boundary terms, as is the case with $\omega$ and $\S$. This is the problem of \emph{descent}, where we enrich the differential $\mathcal{L}_Q$ by the de Rham differential on $M$. The second big problem one encounters is that not all BV theories satisfy the regularity requirement necessary to induce compatible BV-BFV data. In order to describe such theories as well we will relax our definitions.

In order to do this, we turn to a ``lax'' version of the BV-BFV formalism \cite{mnev2020towards}. We will work with local forms on $\mathcal{F} \times M$ with inhomogeneous form degree on $M$, namely $\kappa^\bullet \in \Omega^{p,\bullet}_\mathrm{loc}(\mathcal{F} \times M)$, and use the codimension to enumerate them, as it makes the notation less cumbersome and more intuitive, i.e.\ $\kappa^{k}$ denotes the $(\mathrm{top} - k)$-form part of
\[\kappa^\bullet = \sum^{\dim M}_{k=0} \kappa^{k},\] 
where $\kappa^k \in \Omega^{p,(\mathrm{top} - k)}_\mathrm{loc}(\mathcal{F} \times M)$. 

\begin{remark}
{What we call ``lax'' BV-BFV formalism is a rewriting of known approaches to local field theory in the BV/BRST formalism such as \cite{barnich1995local,Brandt01,barnich2000local,GrigorievParent11, Sharapov}. We use the term ``lax'' to contrast it with the ``strict'' version given by the BV-BFV formalism proper.}
\end{remark}

This should be compared to the standard BV-BFV formalism (extended to codimension $k$ \cite{cattaneo2014classical}), which instead looks at $\Omega^{\bullet}_{\mathrm{loc}}(\mathcal{F}^{(k)})$ with $\mathcal{F}^{(k)}$ an appropriate space of codimension-$k$ fields. In other words, we describe the BV-BFV picture presented above in terms of densities instead of integrals (cf.\ Remark \ref{rem:LocalFormsBulksvsBoundary}), and forfeiting the symplectic structure at codimension $k$. This setting allows us to phrase equivalence with higher strata in a cohomological way, and it collects all the relevant data before performing the quotient in Equation \eqref{eq:quotientboundary}, thus temporarily avoiding potential complications.\footnote{A similar idea is contained in the work of Brandt, Barnich and Henneaux \cite{barnich1995local}, but without the structural BV-BFV equations.}

The definitions that we work with rely on the \emph{lax degree}\footnote{In \cite{mnev2020towards} the authors denote the lax degree by total degree.} $\ld(\cdot)$, which describes the interplay between the co-form degree on $M$ and the ghost number. Let $\text{fd}_M(\cdot)$ denote the form degree on $M$. The lax degree is defined as the difference of the ghost number $\text{gh}(\cdot)$ and the co-form degree $\cfd(\cdot) \coloneqq \dim M - \text{fd}_M(\cdot)$
\begin{align*}
    \ld (\cdot) \coloneqq \text{gh}(\cdot) - \cfd(\cdot).
\end{align*}
In particular, if an inhomoegeneous local form has vanishing lax degree, then the codimension of its homogeneous components corresponds to their respective ghost number. Most notably, this will be the case for the Lagrangian density. We will use the total degree for computations, which for elements in $\Omega^{\bullet,\bullet}_{\mathrm{loc}}({\mathcal{E}}\times M)$ is given by $|\cdot| = \text{gh}(\cdot) + \text{fd}_M(\cdot) + \text{fd}_{\F}(\cdot)$, where $\text{fd}_{\F}(\cdot)$ is the form degree on $\F$.

\begin{definition}[Lax BV-BFV theory]\label{def:LaxBVBFV}
A \emph{lax BV-BFV theory} over a manifold $M$ is a quadruple $\mathfrak{F}^\text{lax} = (\F^\mathrm{lax},\theta^\bullet,L^\bullet,Q)$ where
\begin{itemize}
    \item $\F^\mathrm{lax} = \Gamma(M,F)$ for some $\mathbb{Z}$-graded fibre bundle $F\rightarrow M$,
    \item $\theta^\bullet \in \Omega^{1,\bullet}_{\text{loc}}(\F^\mathrm{lax} \times M)$ is a local form with lax degree -1,
    \item $L^\bullet \in \Omega^{0,\bullet}_{\text{loc}}(\F^\mathrm{lax} \times M)$ is a local functional with lax degree 0,
    \item $Q \in \mathfrak{X}_{\text{evo}}(\F^\mathrm{lax})$ is an evolutionary, cohomological vector field on $\F^\mathrm{lax}$ of degree 1, i.e.\ $[\mathcal{L}_Q,\d] = [Q,Q] = 0$,
\end{itemize}
such that
\begin{subequations}
\label{eq:LaxStructure}
\begin{align}
    \label{eq:LaxStructure1}
    &\iota_Q \varpi^\bullet = \delta L^\bullet + \d \theta^\bullet,\\
     \label{eq:LaxStructure2}
    &\iota_Q \iota_Q \varpi^\bullet = 2\d L^\bullet,
\end{align}
\end{subequations}
where $\varpi^\bullet \coloneqq \delta \theta^\bullet$.
\end{definition}

\begin{remark}[{Strictification of lax data}]
\label{remark:BVandBVBFVfromlax}\label{rem:strictification}
Let $M^\circ$ be the interior (bulk) of $M$. 
\begin{enumerate}
    \item If $M = M^\circ$ is a closed manifold, then we can assign a BV theory $\mathfrak{F}$ to $M^\circ$ from a lax BV-BFV theory $\mathfrak{F}^\mathrm{lax}$ by choosing\footnote{We denote by $\F^\mathrm{lax}\vert_{M^\circ}$ (resp. $\F^\mathrm{lax}\vert_{M^\partial}$) the restriction of fields to the interior (resp the boundary stratum, where we also restrict normal jets) of $M$, seen as section of a fibre bundle (resp. the tangent bundle to the induced bundle).}
\begin{align*}
    &\F = \F^\mathrm{lax}\vert_{M^\circ},&
    &\omega = \int_{M^\circ} \varpi^0,&
    &\S = \int_{M^\circ} L^0
\end{align*}
and restricting $Q$ to $\F$. 
    \item Similarly, if $M$ is a compact manifold with boundary, the {pre-}BFV data on $\partial M$ may be induced by setting
\begin{align*}
    &\check \F^\partial = \F^\mathrm{lax}\vert_{M^\partial},&
    &\check \alpha = \int_{M^\partial} \theta^1,&
    &\check \S = \int_{M^\partial} L^1,
\end{align*}
{and restricting $Q$ to $\check{\F}^\partial$. When pre-symplectic reduction w.r.t. $\check \omega = \delta \check \alpha = \int_{M^\partial} \varpi^1$ is possible \cite{cattaneo2014classical}, we can define the space of boundary fields $\mathcal{F}^\partial:=\check \F^\partial/\mathrm{ker}(\check \omega^\sharp)$}. Together with the bulk data presented above, this produces a BV-BFV theory.
\end{enumerate}
The procedure is analogous for higher codimensions: {if $M^{(k)}$ denotes the $k$th-codimension stratum, we can induce a Hamiltonian dg manifold of fields in codimension $k$ by performing pre-symplectic reduction of 
\[
\left(\check{\mathcal{F}}^{(k)}=\mathcal{F}^{\mathrm{lax}}\vert_{M^{(k)}}, \check{\omega}^{(k)}=\int_{M^{(k)}}\delta \theta^k\right) \leadsto \mathcal{F}^{(k)}:=\check{\mathcal{F}}^{(k)}/\mathrm{ker}(\check{\omega}^{(k)\sharp})
\]
Notice that pre-symplectic reduction might fail to be smooth, resulting in an obstruction to \emph{strictification}. When there are no obstructions to the pre-symplectic reduction, this procedure yields an $n$-extended BV-BFV theory (cf.\ Remark \ref{rem:n-ext}), and we have an \emph{$n$-strictification of a lax BV-BFV theory}. (For more details we refer to \cite{mnev2020towards}.) It is crucial to observe that this step can fail \cite{cattaneo2017bv,cattaneo2017time,martinoli2020bvanalysis}.  

Unlike the latter}, a lax BV-BFV theory does not require working with symplectic structures at higher codimensions $\geq 1$. {This means that lax data allow us to extract \emph{some} information about the higher codimension behaviour of the field theory but, as we will see, the fact that a theory is strictifiable at a given codimension yields a refinement of the notion of BV equivalence.}

\end{remark}

\begin{remark}
At codimension $\geq 1$, it is sufficient to know $\theta^\bullet$ in order to compute $L^\bullet$. Applying $\iota_E$ to Equation \eqref{eq:LaxStructure1} yields
\[
  \iota_E \iota_Q \varpi^\bullet = \mathcal{L}_E L^\bullet + \iota_E \d \theta^\bullet \nonumber,
\]
which implies
\begin{equation}
    \mathcal{L}_E L^\bullet = \iota_E \left(\iota_Q \delta - \d \right) \theta^\bullet.
    \label{eq:LkLax}
\end{equation}
We can then compute $L^k$ at codimension $k \geq 1$ by using $\gh(L^k) = \cfd(L^k) = k$:
\begin{align*}
    L^k = \frac{1}{k}\iota_E \left(\iota_Q \delta \theta^k - \d \theta^{k+1}\right).
\end{align*}
\end{remark}

\begin{lemma}[\cite{cattaneo2014classical,mnev2020towards}]
The following equations hold for a lax BV-BFV theory:
\begin{subequations}
\label{eq:Lax}
\begin{align}
    \label{eq:Lax2}
    &\mathcal{L}_Q \varpi^\bullet = \d \varpi^\bullet,\\
    &\mathcal{L}_Q L^\bullet = \d (2L^\bullet - \iota_Q \theta^\bullet)\label{eq:Lax1}.
\end{align}
\end{subequations}
\end{lemma}

\begin{remark}
Equations (\ref{eq:Lax}) are the density versions of Equations (\ref{eq:FailureOmegaSLQcocycles}). Comparing the two versions, we see that boundary terms are now encoded as $\d$-exact terms, instead of objects in the image of $\pi^*$. Note that $\varpi^\bullet$ is a cocycle of the differential $(\mathcal{L}_Q - \d)$ and that $L^\bullet$ is so whenever $L^\bullet - \iota_Q \theta^\bullet = 0$.
\end{remark}
In the lax BV-BFV formalism, the relevant differential will no longer be $\mathcal{L}_Q$, as we want to take the boundary configurations into account. Instead, we want to consider a cochain complex of local forms on $\F^\mathrm{lax} \times M$ with differential $\mathcal{L}_Q - \d$, which describes the interplay between gauge invariance and boundary terms:
\begin{definition}[\cite{barnich1995local,mnev2020towards}]
\label{def:BVBFVcomplex}
The \emph{BV-BFV complex}  of a lax BV-BFV theory $\mathfrak{F}^{\text{lax}}$ is defined as the space of inhomogeneous local forms on $\F^\mathrm{lax} \times M$ endowed with the differential $(\mathcal{L}_Q - \d)$
\begin{align*}
    (\bvbfv)^\bullet
    &\coloneqq \left(\left(\bigoplus_k \Omega^{\bullet,k}_\text{loc}(\F^\mathrm{lax} \times M)\right),(\mathcal{L}_Q - \d) \right),
\end{align*}
where the grading of $(\bvbfv)^\bullet$ is given by the lax degree. We will denote its cohomology by $H^\bullet((\bvbfv)^\bullet)$ and call it the \emph{BV-BFV cohomology}.
\end{definition}

\begin{remark}
The cocycle conditions for an inhomogenoeous local form $\mathcal{O}^\bullet \in \Omega^{p,\bullet}_\mathrm{loc}(\F^\mathrm{lax} \times M)$ 
are often called the \emph{descent equations} \cite{zumino1985cohomology,manes1985algebraic,witten1988topological,mnev2020towards}
\begin{align*}
    (\mathcal{L}_Q - \d) \mathcal{O}^\bullet = 0 
\end{align*}
i.e.\ $\mathcal{L}_Q \mathcal{O}^k = \d \mathcal{O}^{k+1}$ with homogeneous components $\mathcal{O}^k$. Such equations are of interest since their solutions produce classical observables, i.e.\ local functionals (i.e.\ $p=0$) which belong to $H^0(\mathfrak{BV}^\bullet)$. Let $\gamma^k$ denote a $(\dim M - k)$-dimensional closed submanifold of $M$. We can then construct a classical observable by integrating $\mathcal{O}^k$ over $\gamma^k$ since
\begin{align*}
    \mathcal{L}_Q \int_{\gamma^k} \mathcal{O}^k =
    \int_{\gamma^k} \mathcal{L}_Q \mathcal{O}^k =  \int_{\gamma^k} \d \mathcal{O}^{k+1} = 0.
\end{align*}
As such, comparing the BV-BFV cohomologies of two lax theories offers a natural way of comparing their spaces of classical observables.
\end{remark}

\subsection{Equivalence in the lax setting}
\label{sec:EquivalenceLax}
Before adapting our notion of equivalence to the case when a boundary and corners are present, let us define $f$-transformations, which encode the facts that eventually (i) we are interested in the 2-forms $\varpi^\bullet$ (and not in their potentials $\theta^\bullet$) and (ii) Lagrangian densities will be integrated (so total derivatives become irrelevant). The two issues are actually related.

\begin{definition}
Let $f \in \Omega^{0,\bullet}_\mathrm{loc}(\F^{\mathrm{lax}} \times M)$ be a local functional with $\ld(f) = -1$. An $f$-transformation of a lax BV-BFV theory $\mathfrak{F}^\mathrm{lax}$ changes $(\theta^\bullet,L^\bullet)$ as
\begin{align*}
    &\theta^\bullet \mapsto \theta^\bullet + \delta f^\bullet,&
    &L^\bullet \mapsto L^\bullet + \d f^\bullet.
\end{align*}
\end{definition}
\begin{remark}
\label{rem:ftransformation}
Note that an $f$-transformation preserves Equations \eqref{eq:LaxStructure} since $\varpi^\bullet = \delta \theta^\bullet$ and $\d L^\bullet$ are unchanged, as is the term $\delta L^\bullet + \d \theta^\bullet$
\begin{align*}
    \delta L^\bullet + \d \theta^\bullet
    \mapsto \delta L^\bullet + \delta \d f 
    + \d \theta^\bullet +\d \delta f = \delta L^\bullet + \d \theta^\bullet,
\end{align*}
where we used $[\delta,\d]=0$. Hence, we will also allow this kind of freedom in our definition of equivalence.
\end{remark}

In the following, we will denote the vertical differentials on $\F^\mathrm{lax}_i$ by $\delta$ and the horizontal (de Rham) differentials on $M_i$ by $\d$.

\begin{definition}[Lax equivalence] 
\label{def:laxequivalence}
We say that two lax theories $\mathfrak{F}^\mathrm{lax}_1$ and $\mathfrak{F}^\mathrm{lax}_2$ are \emph{lax equivalent} if there are two morphisms of graded manifolds $\phi\colon \F_2 \rightarrow \F_1$ and $\psi\colon \F_1 \rightarrow \F_2$, which induce quasi-isomorphisms $\phi^*\colon \bvbfv^\bullet_1 \rightarrow \bvbfv^\bullet_2$, $\psi^*\colon \bvbfv^\bullet_2 \rightarrow \bvbfv^\bullet_1$ between the BV-BFV complexes, such that $\phi^*$ and $\psi^*$ are quasi-inverse to each other and transform $(\theta^\bullet_i, L^\bullet_i)$ as
\begin{align}
    &\phi^* \theta^\bullet_1  = \theta^\bullet_2  + (\mathcal{L}_{Q_2} - \mathrm{d}) \beta^\bullet_2 + \delta f^\bullet_2,&
    &\psi^* \theta^\bullet_2  = \theta^\bullet_1  + (\mathcal{L}_{Q_1} - \mathrm{d}) \beta^\bullet_1 + \delta f^\bullet_1 \nonumber,\\
    &\phi^* L^\bullet_1  = L^\bullet_2  + (\mathcal{L}_{Q_2} - \mathrm{d}) \zeta^\bullet_2 + \mathrm{d} f^{\bullet}_2,&
    &\psi^* L^\bullet_2  = L^\bullet_1  + (\mathcal{L}_{Q_1 } - \mathrm{d}) \zeta^\bullet_1 + \mathrm{d} f^{\bullet}_1,
    \label{eq:Laxequivalencetransformation}
\end{align}
with 
$\beta^\bullet_i \in \Omega^{1,\bullet}_\mathrm{loc}(\F^\mathrm{lax}_i\times M_i)$, 
$\# (\beta^\bullet_i) = -2$, 
$\zeta^\bullet_i \in \Omega^{0,\bullet}_\mathrm{loc}(\F^\mathrm{lax}_i \times M_i)$, 
$\# (\zeta^\bullet_i) = -1$ and
$f^\bullet_i \in \Omega^{0,\bullet}_\mathrm{loc}(\F^\mathrm{lax}_i\times M_i)$,
$\# (f^\bullet_i) = -1$.
\end{definition}

\begin{remark}
Similarly to the bulk case, in order to show that the composition maps 
\begin{align*}
    &\chi^* = \psi^* \circ \phi^*: \bvbfv^\bullet_1 \rightarrow \bvbfv^\bullet_1,\\
    &\lambda^* = \phi^* \circ \psi^*: \bvbfv^\bullet_2 \rightarrow \bvbfv^\bullet_2,
\end{align*}
are the identity when restricted to the respective cohomologies, one needs to find two maps $h_\chi:\bvbfv^\bullet_{1} \rightarrow \bvbfv^\bullet_{1}$, $h_\lambda:\bvbfv^\bullet_{2} \rightarrow \bvbfv^\bullet_{2}$ of lax degree $-1$ such that
\begin{align*}
    &\chi^* - \mathrm{id}_{1} = (\mathcal{L}_{Q_1} - \d) h_\chi + h_\chi (\mathcal{L}_{Q_1} - \d),\\
    &\lambda^* - \mathrm{id}_2 = (\mathcal{L}_{Q_2} - \d) h_\lambda + h_\lambda (\mathcal{L}_{Q_2} - \d).
\end{align*}
\end{remark}

\begin{proposition}
\label{prop:LaxRedundancyTransformation}
If $\gh(\phi) = \gh(\psi) = 0$,\footnote{We restrict ourselves to the $\gh(\phi) =  \gh(\psi) = 0$ case as this will be the relevant one in our examples.} then the transformation of $L^\bullet_i$ is not independent from the transformation of $\theta^\bullet_i$:
\begin{enumerate}
    \item $\zeta^k_i = \iota_{Q_i} \beta^k_i$ at codimension $k \geq 1$,\\
    \item $\mathcal{L}_{Q_i} (\iota_Q \beta^0_i - \zeta^0_i) = 0$.
\end{enumerate}
\end{proposition}

\begin{proof}
\leavevmode
\begin{enumerate}
    \item As $\gh(\phi) = 0$, $\phi^*$ commutes with the Euler vector fields $\mathcal{L}_{E_i}$: $\mathcal{L}_{E_2} \phi^* = \phi^* \mathcal{L}_{E_1}$. Applying $\phi^*$ to Equation \eqref{eq:LkLax} then yields
\begin{align*}
    &\mathcal{L}_{E_2} \phi^* L^\bullet_1 
    = \iota_{E_2} \left(\iota_{Q_2} \delta - \d \right) \phi^* \theta^\bullet_1\\
    &= \iota_{E_2} \left(\iota_{Q_2} \delta - \d \right) \theta^\bullet_2
    + \iota_{E_2} \left(\iota_{Q_2} \delta - \d \right)(\mathcal{L}_{Q_2} - \mathrm{d}) \beta^\bullet_2
    + \iota_{E_2} \left(\iota_{Q_2} \delta - \d \right) \delta f^\bullet_2.
\end{align*}
The first term is simply $\mathcal{L}_{E_2} L^\bullet_2$. For the second term we compute
\begin{align*}
    &\iota_{E_2} \left(\iota_{Q_2} \delta - \d \right)(\mathcal{L}_{Q} - \mathrm{d}) \beta^\bullet_2 
    = \iota_{E_2} \left(\delta \iota_{Q_2} + \mathcal{L}_{Q_2} - \d \right)(\mathcal{L}_{Q_2} - \mathrm{d}) \beta^\bullet_2\\
    &= \iota_{E_2} \delta \iota_{Q_2} (\mathcal{L}_{Q_2} - \mathrm{d}) \beta^\bullet_2
    = (\mathcal{L}_{E_2} - \delta \iota_{E_2}) \iota_{Q_2} (\mathcal{L}_{Q_2} - \mathrm{d}) \beta^\bullet_2\\
    &= \mathcal{L}_{E_2} (\mathcal{L}_{Q_2} - \mathrm{d}) \iota_{Q_2} \beta^\bullet_2,
\end{align*}
where we used that $(\mathcal{L}_{Q_2} - \mathrm{d}) \beta^\bullet_2 \in \Omega^{1,\bullet}(\F^\mathrm{lax}\times M)$ implies $\iota_{E_2} \iota_{Q_2} (\mathcal{L}_{Q_2} - \mathrm{d}) \beta^\bullet_2 = 0$.
The third term reads
\begin{align*}
    \iota_{E_2} \left(\iota_{Q_2} \delta - \d \right) \delta f^\bullet_2
    = - \iota_{E_2} \d \delta f^\bullet_2
    = \d \iota_{E_2} \delta f^\bullet_2
    = \d \mathcal{L}_{E_2}  f^\bullet_2
    = \mathcal{L}_{E_2} \d f^\bullet_2,
\end{align*}
hence
\begin{align*}
        \mathcal{L}_{E_2} \phi^* L^\bullet_1 =
        \mathcal{L}_{E_2} (L^\bullet_2
        + (\mathcal{L}_{Q_2} - \d) \iota_{Q_2} \beta^\bullet_2
        + \d f^\bullet_2).
\end{align*}
By counting degrees we see that both sides have ghost number $k$ at codimension $k$, therefore for $k\geq 1$ one can use this equation to determine $\phi^*L^k_1$, in particular we have $\zeta^k_i = \iota_{Q_i} \beta^k_i$.

\item Applying $\phi^*$ to $\iota_{Q_1} \varpi^\bullet_1 = \delta L^\bullet_1 + \d \theta^\bullet_1$ yields
\begin{align}
    &\iota_{Q_2} \delta (\mathcal{L}_{Q_2} - \mathrm{d}) \beta^\bullet_2 = \delta (\mathcal{L}_{Q_2} - \mathrm{d}) \zeta^\bullet_2
    + \d (\mathcal{L}_{Q_2} - \mathrm{d}) \beta^\bullet_2\nonumber\\
    \Rightarrow \quad&(\mathcal{L}_{Q_2} - \mathrm{d}) \iota_{Q_2} \delta \beta^\bullet_2 = (\mathcal{L}_{Q_2} - \mathrm{d}) \delta \zeta^\bullet_2
    + (\mathcal{L}_{Q_2} - \mathrm{d})\d \beta^\bullet_2\nonumber\\
    \Rightarrow \quad&(\mathcal{L}_{Q_2} - \d) [(\iota_{Q_2} \delta-\d)\beta_2^\bullet - \delta \zeta^\bullet_2] = 0,
    \label{eq:middlecondition}
\end{align}
where we used $\iota_{Q_2} \varpi^\bullet_2 = \delta L^\bullet_2 + \d \theta^\bullet_2$, $\delta^2 = 0$ and the fact that $f$-transformations preserve Equation \eqref{eq:LaxStructure1}. Note that this condition holds automatically for condimention higher than zero due to $\zeta^k_2 = \iota_{Q_2} \beta^k_2$. To see what Equation \eqref{eq:middlecondition} implies at codimention zero, first note that
\begin{align*}
    \iota_{Q_2} \delta \beta_2^1 
    - \delta \zeta^1_2 
    = \iota_{Q_2} \delta \beta_2^1 - \delta \iota_{Q_2} \beta_2^1
    = \mathcal{L}_{Q_2} \beta_2^1.
\end{align*}
Keeping in mind that $[\mathcal{L}_Q,\d] = 0$, Equation \eqref{eq:middlecondition} gives
\begin{align*}
    &&
    &\mathcal{L}_{Q_2}[
    \iota_{Q_2} \delta \beta_2^0 
    - \d \beta_2^1 
    - \delta \zeta^0_2
    ] 
    - \d
    [
    \iota_{Q_2} \delta \beta_2^1 
    - \delta \zeta^1_2
    ] = 0
    &&\\
    &\Rightarrow&
    &\mathcal{L}_{Q_2} \delta (\iota_{Q_i} \beta^0_2 - \zeta_2^0)
    = 0.
    &&
\end{align*}
We now apply $\iota_{E_2}$, but first note that
\begin{align*}
    [\mathcal{L}_Q,\iota_{E}] 
    = \iota_{[E,Q]}
    = \iota_{\mathcal{L}_E Q} = \iota_Q,
\end{align*}
therefore 
\begin{align*}
    &&
    \iota_{E_2}\mathcal{L}_{Q_2} \delta (\iota_{Q_i} \beta^0_2 - \zeta_2^0)
    &= [\iota_Q - \mathcal{L}_Q \iota_E] \delta
    (\iota_{Q_i} \beta^0_2 - \zeta_2^0)
    &&\\
    &&
    &= 2 \mathcal{L}_{Q_2} (\iota_{Q_i} \beta^0_2 - \zeta_2^0)
    &&\\
    &\Rightarrow&
    \mathcal{L}_{Q_2} (\iota_{Q_i} \beta^0_2 - \zeta_2^0) &= 0
    &&
\end{align*}
where we used $\iota_Q \delta (\iota_{Q_i} \beta^0_2 - \zeta_2^0) = \mathcal{L}_Q (\iota_{Q_i} \beta^0_2 - \zeta_2^0)$ and $\gh(\iota_{Q_i} \beta^0_2 - \zeta_2^0) = \#(\iota_{Q_i} \beta^0_2 - \zeta_2^0) = -1$. The computations are analogous for $i = 1$.

\end{enumerate}

\end{proof}

\begin{remark}
\label{rem:WhatToCheck}
The previous lemma means that there is a redundancy in our definition of lax equivalence, as the transformation of $L^\bullet_i$ at codimension $\geq 1$ can be determined through the transformation of $\theta^\bullet_i$. In particular, when computing explicit examples one only needs to check if we have the right transformation for $\theta^\bullet_i$ and $L^0_i$. Observe that if we have $H^{-1}(\mathcal{L}_Q) = 0$ we can conclude that $\iota_{Q_i} \beta^0_2 - \zeta_2^0 = \mathcal{L}_Q(\dots)$.
\end{remark}

\begin{remark}
\label{rem:specialcaseLQdeltaexactLax}
Our definition of lax equivalence directly implies that the local 2-forms $\varpi^\bullet_i$ are interchanged up to $(\mathcal{L}_{Q_i} - \d)$-exact and $\delta$-exact terms
\begin{align*}
    &\phi^* \varpi^\bullet_1
    = \delta \phi^* \theta^\bullet_1
    = \varpi^\bullet_2 - (\mathcal{L}_{Q_2} - \d) \delta \beta^\bullet_2,\\
    &\psi^* \varpi^\bullet_2
    = \delta \psi^* \theta^\bullet_2
    = \varpi^\bullet_1 - (\mathcal{L}_{Q_1} - \d) \delta \beta^\bullet_1.
\end{align*}
Similar to the bulk case (cf.\ Proposition \ref{eq:simplifiedhamiltoniancondition}), choosing $\zeta^0_i = \iota_{Q_i} \beta_i^0$ ensures that the second condition from Proposition \ref{prop:LaxRedundancyTransformation} is satisfied.
\end{remark}

\begin{proposition}
The theories $\mathfrak{F}^\mathrm{lax}_1$, $\chi^*\mathfrak{F}^\mathrm{lax}_1 \coloneqq (\F_1^\mathrm{lax}, \chi^*\theta^\bullet_1, \chi^*L^\bullet_1, Q_1)$ and $\mathfrak{F}^\mathrm{lax}_2$, $\lambda^*\mathfrak{F}^\mathrm{lax}_2 \coloneqq (\F_2^\mathrm{lax}, \lambda^*\theta^\bullet_2, \lambda^*L^\bullet_2, Q_2)$ are pairwise lax-equivalent.

\end{proposition}

\begin{proof}
The proof is analogous to the proof of Proposition \ref{rem:Pairwiselaxequivalent}.
\end{proof}

\begin{theorem}
\label{theorem:laxEqimpliesBVEq}
Let $\mathfrak{F}^\mathrm{lax}_i$, $i \in \{1,2\}$, be lax equivalent. Then the respective BV theories $\mathfrak{F}_i$ (cf.\ Remark \ref{remark:BVandBVBFVfromlax})  are BV equivalent.
\end{theorem}

\begin{proof}
Let $\kappa^\bullet \in \Omega^{p,\bullet}_\mathrm{loc}({\mathcal{E}} \times M)$ and $K \coloneqq \int_{M} \kappa^0 \in \Omega^{p}_{{\int}}({\mathcal{E}})$. We need to check if:
\begin{enumerate}
    \item $\phi^*,\psi^*$ are chain maps w.r.t. the BV complexes,
    \item the cohomological classes of $\omega_i$ and $\S_i$ are mapped into one another,
    \item $\chi^*, \lambda^*$ are the identity on $H^\bullet(\mathfrak{BV}^\bullet_i)$.
\end{enumerate}
To prove these, we simply need to integrate the various conditions over the bulk $M$. For the chain map condition we have
\begin{align*}
    && 
    \phi^* \circ (\mathcal{L}_{Q_1}-\d) \kappa^\bullet &= (\mathcal{L}_{Q_2}-\d) \circ \phi^* \kappa^\bullet
    &&\\
    &\Rightarrow&  
    \int_{M} \phi^* (\mathcal{L}_{Q_1} \kappa^0 -\d\kappa^1) &=  \int_{M} (\mathcal{L}_{Q_2} \phi^* \kappa^ 0 -\d \phi^* \kappa^1)
    &&\\
    &\Rightarrow& 
    \phi^* \mathcal{L}_{Q_1} K &= \mathcal{L}_{Q_2} \phi^* K.
    &&
\end{align*}
We can compute the transformations of $\omega_1$ and $\S_1$ in a similar way
\begin{align*}
    &&
    &\phi^* \varpi^\bullet_1  = \omega^\bullet_2  - (\mathcal{L}_{Q_2} - \mathrm{d}) \delta \beta^\bullet_2,&
    &\phi^* L^\bullet_1  = L^\bullet_2  + (\mathcal{L}_{Q_2} - \mathrm{d}) \zeta^\bullet_2 + \mathrm{d} f^{\bullet}_2,&
    &&\\
    &\Rightarrow& 
    &\phi^* \omega_1 = \omega_2 - \mathcal{L}_{Q_2} \int_{M} \delta \beta^0_2,&
    &\phi^* \S_1  = \S_2  + \mathcal{L}_{Q_2} \int_{M} \zeta^0_2.&
    &&
\end{align*}
In the same manner
\begin{align*}
    &&&(\chi^* - \mathrm{id}_{1}) \kappa^\bullet = (\mathcal{L}_{Q_1} - \d) h_\chi \kappa^\bullet + h_\chi (\mathcal{L}_{Q_1} - \d) \kappa^\bullet,\\
    &\Rightarrow&&(\chi^* - \mathrm{id}_{1}) K = \mathcal{L}_{Q_1} h_\chi K + h_\chi \mathcal{L}_{Q_1} K,
\end{align*}
implying that $\chi^*$, and analogously $\lambda^*$, are also homotopic to the identity in $\mathfrak{BV}^\bullet_i$ and as such the identity when restricted to $H^\bullet(\mathfrak{BV}^\bullet_i)$, meaning that the BV complexes are quasi-isomorphic.
\end{proof}

\section{Examples}
\label{sec:Examples}

This section is dedicated to the explicit computation of lax BV-BFV equivalence in three different examples. We start by presenting the general strategy in Section \ref{sec:Strategy}. We then look at the examples of classical mechanics on a curved background and (non-abelian) Yang--Mills theory in Sections \ref{sec:CM} and \ref{sec:YM} respectively. Subsequently we turn our attention to the classically-equivalent Jacobi theory and one-dimensional gravity coupled to matter (1D GR) in Section \ref{sec:1Dreparametrisationinvarianttheories}, and show that they are lax BV-BFV equivalent, despite their different boundary behaviours w.r.t. the BV-BFV procedure. Furthermore, we show that the chain maps used to prove lax BV-BFV equivalence spoil the compatibility with the regularity condition for the BV-BFV procedure in the case of 1D GR.

\subsection{Strategy}
\label{sec:Strategy}

We shortly demonstrate our strategy to show explicitly that two theories $\mathfrak{F}^\mathrm{lax}_i$ are lax BV-BFV equivalent. In practice, we need two maps $\phi^*$, $\psi^*$ between the BV-BFV complexes $\bvbfv^\bullet_i$
\begin{equation*}
\begin{tikzcd}
        &\bvbfv_1^\bullet \arrow[r, rightarrow, shift left=0.5ex, "\phi^*"]
        &\bvbfv_2^\bullet \arrow[l, rightarrow, shift left=0.5ex, "\psi^*"]
\end{tikzcd}
\end{equation*}
which (cf.\ Definition \ref{def:laxequivalence}):
\begin{enumerate}
    \item are chain maps, 
    \item transform $\theta^\bullet_i, L^\bullet_i$ in the desired way (cf.\ Equations \eqref{eq:Laxequivalencetransformation}),
    \item are quasi-inverse to one another. 
\end{enumerate}

In order to check the first property we note that, as pullback maps, $\phi^*$ and $\psi^*$ automatically commute with the de Rham differentials $\delta$ and $\d$. Thus it suffices to show that $\phi^*, \psi^*$ are chain maps w.r.t. $Q_i$ on the fields $\varphi^j_i \in \F^\mathrm{lax}_i$
\begin{align*}
    &\phi^* \Qf \varphi^j_1 = Q_2 \phi^* \varphi^j_1,&
    &\psi^* Q_2 \varphi^j_2 = \Qf \psi^* \varphi^j_2,
\end{align*}
as this together with the fact that they commute with $\delta$ then implies that they are chain maps w.r.t. $(\mathcal{L}_{Q_i}-\d)$ on $\bvbfv^\bullet$, resulting in the following Lemma:

\begin{lemma}
\label{prop:phipsichainmapfromQ}
If the pullback maps $\phi^*, \psi^*$ are chain maps w.r.t. $Q_i$ on $\F_i$, then they are also chain maps w.r.t. $(\mathcal{L}_{Q_i}-\d)$ on $\bvbfv^\bullet_i$.
\end{lemma}

Showing the second property is a matter of computation. For the third property, we shortly present our strategy to show that $\chi^* = \psi^* \circ \phi^*$ is the identity in cohomology. The same procedure can then be applied to $\lambda^* = \phi^* \circ \psi^*$. Recall that we need a map $h_\chi:\bvbfv^\bullet_{1} \rightarrow \bvbfv^\bullet_{1}$ of lax degree $-1$ such that
\begin{align}
    \label{eq:chihomotopicidentity}
    &\chi^* - \mathrm{id}_{1} = (\mathcal{L}_{\Qf} - \d) h_\chi + h_\chi (\mathcal{L}_{\Qf} - \d).
\end{align}
We start by constructing a homotopy between $\chi^*$ and $\mathrm{id}_1$, by finding an evolutionary vector field $\Rf \in\mathfrak{X}_{\text{evo}}(\F_1)$ with $\gh(\Rf) = \#(\Rf) = -1$ and defining a one-parameter family of morphisms of the form 
\begin{align*}
    \chi^*_s \coloneqq e^{s[(\mathcal{L}_{\Qf} - \d), \mathcal{L}_{\Rf}]},
\end{align*}
such that $\chi^*_{s=0} = \mathrm{id}_1$ and $\lim_{s\rightarrow \infty}\chi^*_{s} = \chi^*$. Note that choosing $s = -\ln\tau$ gives the usual definition of homotopy with $\tau \in [0,1]$, i.e.\ a continuous map $F(\tau) \coloneqq \chi^*_{(-\ln\tau)}$ satisfying $F(0) = \chi^*$ and $F(1) = \mathrm{id}_1$. We will nonetheless work with the parameter $s$ to keep the calculations cleaner, changing when necessary. Furthermore, we can simplify the term in the exponent by setting $\Df \coloneqq [\Qf,\Rf]$, since 
\begin{align*}
    [(\mathcal{L}_{\Qf} - \d), \mathcal{L}_{\Rf}]
    = [\mathcal{L}_{\Qf}, \mathcal{L}_{\Rf}]
    = \mathcal{L}_{[\Qf,\Rf]}
    = \mathcal{L}_{\Df},
\end{align*}
which results in
\begin{align*}
    \chi^*_s = e^{s\mathcal{L}_{\Df}}.
\end{align*}

\begin{lemma}
\label{prop:Dcommuteswithdifferential}
The Lie derivative $\mathcal{L}_{\Df}$ commutes with the differential $(\mathcal{L}_{\Qf} - \d)$.
\end{lemma}

\begin{proof}
Since $\Rf$ and $\Qf$ are evolutionary vector fields, we directly have $[\mathcal{L}_{\Df},\d] = 0$. To see that $\mathcal{L}_{\Df}$ commutes with $\mathcal{L}_{\Qf}$ we compute
\begin{align*}
    &&&[\Df,\Qf] = [[\Qf,\Rf], \Qf] = - [[\Rf,\Qf], \Qf] -[[\Qf ,\Qf], \Rf]
    = - [\Df,\Qf]&\\
    &\Rightarrow& &[\Df, \Qf] = 0.
\end{align*}
where we have used the graded Jacobi identity and $[\Qf,\Qf] = 0$. Thus
\begin{align*}
    [\mathcal{L}_{\Df}, \mathcal{L}_{\Qf}]
    &= \mathcal{L}_{[\Df, \Qf]} = 0,
\end{align*}
proving the statement.
\end{proof}
We can then determine the map $h_\chi$ by rewriting the RHS of Equation (\ref{eq:chihomotopicidentity}) as
\begin{align*}
    \chi^* -\mathrm{id}_1 
    &= \int^\infty_0 \frac{\mathrm{d}}{ds}e^{s\mathcal{L}_{\Df}}
    = \int^\infty_0 e^{s\mathcal{L}_{\Df}} \mathcal{L}_{\Df} \d s\\
    &= \int^\infty_0 e^{s\mathcal{L}_{\Df}} [(\mathcal{L}_{\Qf}-\d) \mathcal{L}_{\Rf} + \mathcal{L}_{\Rf} (\mathcal{L}_{\Qf}-\d)] \d s\\
    &= (\mathcal{L}_{\Qf}-\d) \left(\int^\infty_0 e^{s\mathcal{L}_{\Df}} \mathcal{L}_{\Rf} \d s \right)
    + \left(\int^\infty_0 e^{s\mathcal{L}_{\Df}} \mathcal{L}_{\Rf} \d s \right) (\mathcal{L}_{\Qf}-\d),
\end{align*}
resulting in the following Lemma:
\begin{lemma}
\label{lem:hchiexplicitform}
The map $h_\chi:\bvbfv^\bullet_{1} \rightarrow \bvbfv^\bullet_{1}$ defined through
\begin{align*}
    h_\chi \kappa = \int^\infty_0 e^{s\mathcal{L}_{\Df}} \mathcal{L}_{\Rf} \kappa \, \d s
\end{align*}
satisfies Equation (\ref{eq:chihomotopicidentity}).
\end{lemma}

If $h_\chi$ converges on $\bvbfv^\bullet_1$, then $\chi^*$ will be the identity in the BV-BFV cohomology $H^\bullet(\bvbfv^\bullet_1)$ as desired. For this last step, the next Lemma will be useful:

\begin{lemma}
\label{prop:hchiconvergence}
If $h_\chi$ converges on $\F_1$, then it converges on the whole BV-BFV complex $\bvbfv^\bullet_1$.
\end{lemma}

\begin{proof}
Let $\kappa \in \bvbfv^\bullet_1$. Start by redefining $s = - \ln \tau$ with $\tau\in [0,1]$, such that we integrate over a compact interval instead of over $\mathbb{R}_{\geq 0}$. Performing this transformation results in
\begin{align*}
    h_\chi \kappa = \int^\infty_0 e^{s\mathcal{L}_{\Df}} \mathcal{L}_{\Rf} \kappa \,\mathrm{d}s
    = \int^0_1 e^{-\ln (\tau)\mathcal{L}_{\Df}} \mathcal{L}_{\Rf} \kappa\, \mathrm{d}(-\ln \tau)
    = \int^1_0 \frac{e^{-\ln (\tau)\mathcal{L}_{\Df}}}{\tau} \mathcal{L}_{\Rf} \kappa \,\mathrm{d}\tau.
\end{align*}
Assuming that 
\begin{align*}
    h_\chi \varphi^j_1 
    = \int^\infty_0 e^{s \mathcal{L}_{\Df}} \Rf \varphi^j_1 \, \d s
    = \int^1_0 \frac{e^{-\ln (\tau) \mathcal{L}_{\Df}}}{\tau} \Rf \varphi^j_1 \, \mathrm{d}\tau
    < \infty 
    \quad \forall \varphi^j_1 \in \F_1,
\end{align*}
we can show that $h_\chi \kappa$ converges on $\bvbfv^\bullet_1$. 
Let first $\kappa = f$ be a local functional. Writing $\Rf = \Rf\varphi^j_1 \frac{\delta }{\delta \varphi^j_1}$ then gives
\begin{align*}
    h_\chi f &= \int^1_0 \frac{e^{-\ln (\tau)\mathcal{L}_{\Df}}}{\tau} \Rf f \,\mathrm{d}\tau
    = \int^1_0 \frac{e^{-\ln (\tau)\mathcal{L}_{\Df}}}{\tau}  \left(\Rf\varphi^j_1\frac{\delta f}{\delta \varphi^j_1}\right) \mathrm{d}\tau\\
    &= \int^1_0 \left(\frac{e^{-\ln (\tau)\mathcal{L}_{\Df}}}{\tau}  \Rf\varphi^j_1\right) \left(e^{-\ln (\tau)\mathcal{L}_{\Df}}  \frac{\delta f}{\delta \varphi^j_1}\right)  \mathrm{d}\tau,
\end{align*}
where we used that $\chi^*_s = e^{s\mathcal{L}_{\Df}} = e^{-\ln (\tau) \mathcal{L}_{\Df}}$ is a morphism in the last equality. The integral over the first integrand is finite by assumption and the second integrand $e^{-\ln (\tau) \mathcal{L}_{\Df}} \frac{\delta f}{\delta \varphi^j_1} = e^{s \mathcal{L}_{\Df}} \frac{\delta f}{\delta \varphi^j_1}$ is nowhere divergent $\forall \tau \in [0,1]$, since we assume $\chi^*_s = e^{s\mathcal{L}_{\Df}}$ to be well-defined on $\bvbfv^\bullet_1$.

Consider now the local form $\kappa = f \delta \varphi_J \otimes \nu \in \bvbfv^\bullet_1$, where $J$ is a multiindex raging over the fields and their jets, and $\nu \in \Omega^\bullet(M)$ is a form on $M$. We then have
\begin{align*}
    h_\chi \kappa
    =& \int^1_0 \frac{e^{-\ln(\tau)\mathcal{L}_{\Df}}}{\tau} \mathcal{L}_{\Rf} [f \delta \varphi_J \otimes \nu] \d \tau\\
    =& \int^1_0 \frac{e^{-\ln(\tau)\mathcal{L}_{\Df}}}{\tau} [\mathcal{L}_{\Rf} f \delta \varphi_J \pm f \mathcal{L}_{\Rf}\delta \varphi_J] \otimes \nu \d \tau\\
    =&\int^1_0 \bigg[\left\{\frac{e^{-\ln(\tau)\mathcal{L}_{\Df}}}{\tau} \Rf f \right\}\delta \left(e^{-\ln(\tau)\mathcal{L}_{\Df}} \varphi_J\right) \\
    &\mp e^{-\ln(\tau)\mathcal{L}_{\Df}} f \delta \left\{\frac{e^{-\ln(\tau)\mathcal{L}_{\Df}}}{\tau} \Rf \varphi_J\right\}\bigg] \otimes \nu \d \tau.
\end{align*}
The terms in the brackets $\{\cdot\}$ are just the integrands of $h_\chi f$ and $h_\chi \varphi_J$, which converge. Since the other terms $e^{-\ln(\tau)\mathcal{L}_{\Df}} \varphi_J = e^{s\mathcal{L}_{\Df}} \varphi_J$, $e^{-\ln(\tau)\mathcal{L}_{\Df}} f= e^{s\mathcal{L}_{\Df}} f$ are well-defined $\forall \tau\in[0,1]$ and we are integrating over a compact interval, the integral converges and $h_\chi$ is well-defined on $\bvbfv^\bullet_1$.
\end{proof}

\subsection{Contractible pairs}
\label{sec:Cp}
The simplest example we can discuss is when the cohomologically trivial fields are nicely decoupled from the rest. This follows the procedure of \cite{Henneaux90Auxiliary}{, which we explicitly embed in our framework by constructing suitable chain maps to fit Definition \ref{def:BVequivalence} (see also \cite{barnich1995local} and \cite{BarnichGrigorievSemikhatovTipunin,BarnichGrigoriev11,GrigorievParent11})}.
Namely, we have an action of the form
\[
S_1[\tilde a,v,\tilde a^\dagger,v^\dagger]= S_2 [\tilde a,\tilde a^\dagger] + \frac12 (v,v),
\]
where $(\ ,\ )$ is some constant nondegenerate bilinear form on the space $V$ of the $v$ fields and $S_2$ is a solution of the master equation (w.r.t.\ the fields $\Tilde a,\Tilde a^\dagger$). We want to compare this theory with the one defined by $S_2 [a,a^\dagger]$, where we remove the tilde for clarity of the notation.

The cohomological vector field $Q_1$ of $S_1$ acts on the fields $\Tilde a,\Tilde a^\dagger$ as $Q_1$. In addition we have
\[
Q_1 v^\dagger = v,\qquad Q_1 v = 0.
\]
where we hve identified $V^*$ with $V$ using the bilinear fom. The fields $(v,v^\dagger)$ are called a \emph{contractible pair}.

Define maps $\phi$, $\psi$.
\begin{subequations}\label{e:simpleConPairMaps}\begin{align}
\phi^*a &= \Tilde a, &  \phi^*a^\dagger &=\Tilde a^\dagger, &
\phi^*v&=0, & \phi^*v^\dagger &=0,\\
& & \psi^*\Tilde a &= a,& \psi^*\Tilde a^\dagger & = a^\dagger.
\end{align}\end{subequations}

\begin{lemma}
	The composition map $\lambda^*=\phi^*\circ\psi^*$ is the identity, while
	the composition map $\chi^*=\psi^*\circ\phi^*$ acts as
\[
\chi^*a=a,\quad\chi^*a^\dagger=a^\dagger,\quad
\chi^*v=0,\quad\chi^*v^\dagger=0,
\]
and is homotopic to the identity.
\end{lemma}
\begin{proof}
	One can directly check that $\lambda^*$ is the identity.
	In order to show that $\chi^*$ is the identity in cohomology we define a family of maps $\chi^*_s = e^{s\mathcal{L}_{\Df}}$, where $\Df = [\Qf,\Rf]$, and show $\lim_{s\rightarrow \infty} \chi^*_s = \chi^*$. 
	
	We choose $\Rf$ to act as 
	\[
	\Rf a=0,\quad\Rf a^\dagger =0,\quad
	\Rf v=-v^\dagger,\quad \Rf v^\dagger = 0.
	\]
	We can then compute 
	\[
	D_1 v = (Q_1 R_1 + R_1 Q_1) v = Q_1 R_1 v =  - v \qquad  \Longleftrightarrow \qquad D_1^{k} v = (-1)^k v, \quad \forall k\geq 1,
	\]
	and
	\[
	D_1 v^\dag = R_1 v = -v^\dag, \qquad \Longleftrightarrow \qquad D_1^{k} v^\dag = -v^\dag,
	\]
	which yield
	\[
	\chi_s^* v = e^{s\mathcal{L}_{D_1}} v = v+ \sum_{k=1}^\infty \frac{s^k}{k!} D_1^{k} v= e^{-s}v \stackrel{s\to \infty}{\longrightarrow} \chi_\infty ^* v = 0,
	\]
	and similarly $\chi_{\infty}^*v^\dag = 0$.
	
	On the other hand,
	\[
	D_1 a = D_1 a^\dag = 0 \Longrightarrow e^{s\mathcal{L}_{D_1}}a = a, \quad e^{s\mathcal{L}_{D_1}}a^\dag = a^\dag, \quad \forall s, 
	\]
	so that $\lim_{s\to\infty}\chi_s^* \equiv \chi^*$.
	
	Furthermore the map $h_\chi$ converges on all the fields, as $h_\chi a = h_\chi a^\dagger = h_\chi v^\dagger$ trivially and
	\[
	h_\chi v = \int^\infty_0 e^{s\mathcal{L}_{\Df}} \mathcal{L}_{\Rf} v \, \d s
	= - \int^\infty_0 e^{s\mathcal{L}_{\Df}} v^\dagger \, \d s
	= - \int^\infty_0 e^{-s} v^\dagger \, \d s
	= - v^\dagger.
	\]
\end{proof}

A direct consequence of this Lemma, together with the facts that $\phi^*S_1=S_2$ and $\phi^* \omega_1= \omega_2$ for the canonical BV forms, is
\begin{theorem}
	The BV theories defined by $S_1$ and $S_2$ are BV equivalent.
\end{theorem}

\subsubsection{More general contractible pairs}
As in \cite{Henneaux90Auxiliary} we may consider a situation where $S_1$ depends on fields $\Tilde a,w,\Tilde a^\dagger,\Tilde w^\dagger$, where $(w,w^\dagger)$ are not a contractible pair on the nose but satisfy the condition that
\[
(Q_1 w^\dagger)|_{w^\dagger=0} = 0
\]
has a unique solution $w=w(\Tilde a,\Tilde a^\dagger)$. We can then get closer to the previously discussed case by defining\footnote{We may also think of this construction as the semiclassical approximation of a BV pushforward \cite[Section 2.2.2]{cattaneo2018perturbative} that gets rid of the $(w,w^\dagger)$ variables. Indeed, the above condition may be read as the statement that setting $w^\dagger$ to zero is a good gauge fixing.}
\[
v = Q_1 w^\dagger,\quad v^\dagger = w^\dagger.
\]
We have indeed that $Q_1 v^\dagger = v$ and $Q_1 v = 0$. Moreover, the above condition implies that the change of variables $(w,w^\dagger)\mapsto (v,v^\dagger)$ is invertible (near $w^\dagger=0$, or everywhere if $w^\dagger$ is odd) and that the submanifold defined by the constraints $v=0$ and $v^\dagger =0$ is symplectic.

The above strategy then works, in the absence of boundary, with some modifications. Namely, the fact that now $v$ and $v^\dagger$ are not Darboux coordinates requires modifying the map of Equations \eqref{e:simpleConPairMaps}. In turn, the transformation $\Rf$ will get a nontrivial action on the fields $(a,a^\dagger)$.

All the examples we discuss below belong to this class of contractible pairs. As we will see, the modifications required in the case of classical mechanics and Yang--Mills theory are minimal, whereas those required in the case of 1D parametrisation invariant theories are more consistent---due to the fact that pairing of the $v$ fields will depend on the $a$ fields. In addition, in all the examples we show how to extend this construction to lax theories in order to encompass the presence of boundaries.

\subsection{Classical mechanics on a curved background}
\label{sec:CM}
We start by discussing the example of classical mechanics on a curved background as a warm-up exercise. We take the source manifold to be a time interval $I = [a,b] \subset \mathbb{R}$ and some smooth Riemannian manifold $(M,\gf)$ as target. We will denote time derivatives with a dot and use tildes to distinguish fields between the different formulations of the theory.

We can formulate the theory by considering a \say{matter} field $\qs \in F_2 \coloneqq C^\infty(I,M)$, and the metric tensor on the target will depend on the map $\qs$. We introduce the shorthand notation $\gs \coloneqq \gs(\qs)$. The classical action functional is given by
\begin{align*}
    S_2[\qs] = \int_I \frac{1}{2} \gs (\dqs,\dqs) \dt.
\end{align*}
This is usually called the \emph{second-order formulation}.

On the other hand, we can phrase the theory in its \emph{first-order formulation}, by { considering again a map $\qf\colon I \to M$, together with an \say{auxiliary} field\footnote{Obviously the pair $(q,p)$ is a map from $I$ to $T^*M$.} $\pf \in C^\infty(I,\qf^*T^*M)$}  and the classical action functional
\begin{align*}
    S_1[\qf, \pf] = \int_I \left(\langle\pf, \dqf \rangle - \frac{1}{2} \hf(\pf,\pf)\right) \dt,
\end{align*}
where $\hf \coloneqq \gf^{-1}$ denotes the inverse of the target metric.

For ease of notation, we will introduce the \emph{musical isomorphisms} 
\begin{align*}
&\gf^\flat\colon TM \to T^*M, & & \hf^\sharp\colon T^*M \to TM \\
&\gf^\flat(v)(\cdot) = \gf(v,\cdot) & & \hf^\sharp(\alpha)(\cdot) = g^{-1}(\alpha,\cdot),
\end{align*}
and clearly $\gf^\flat\circ \hf^\sharp = \hf^\sharp\circ \gf^\flat = \mathrm{id}$.

We recall the rather obvious and well-known
\begin{proposition}
The first-order and second-order formulations of classical mechanics with a background metric are classically equivalent. 
\end{proposition}

\begin{proof}
Following Definition \ref{def:classicalequivalence}, we start by solving the EL equation of the first-order theory corresponding to the auxiliary field $\pf$, we have
\begin{align*}
    \delta_{\pf} \clSf[\qf,\pf] = \int_I \langle\left(\dqf - \hf^\sharp(\pf)\right),\delta \pf\rangle \, \dt,
\end{align*}
which results in
\begin{align*}
    \pf = \gf^\flat(\dqf).
\end{align*}

Let $C_1$ be the set of such solutions and define the map $\phi_{cl}: C_1 \rightarrow F_2$ through $\phi^*_{cl} \qs = \qf$. Then the restriction of $\clSf[\qf,\pf]$ to $C_1$ coincides with the pullback of $\clSs$ via $\phi^*_{cl}$
\begin{align*}
    \clSf[\qf,\pf]\big \vert_{C_1}
    &= \clSf[\qf,\pf =  \gf^\flat(\dqf)] = \int_I \frac{1}{2} \langle\gf^\flat(\dqf), \dqf\rangle \dt 
    = \phi^*_{cl}\int_I \frac{1}{2} \gs(\dqs,\dqs) \dt
    = \phi^*_{cl} \clSs[\qs],
\end{align*}
hence the two formulations are classical equivalent.
\end{proof}

Both of these theories can be extended to the lax BV-BFV formalism. Note that, as there are no gauge symmetries in these models, there is no need to introduce ghost fields. We start with the second-order formulation: 

\begin{proposition/definition}
The data $$\mathfrak{F}^\mathrm{lax}_{2CM} = (\FsCM,\ts^\bullet,\Ls^\bullet,\Qs)$$ 
where
\begin{align*}
    \FsCM = T^*[-1]C^\infty(I,M).
\end{align*}
together with $\ts^\bullet \in \Omega^{1,\bullet}_\mathrm{loc}(\FsCM)$ and $\Ls^\bullet \in \Omega^{0,\bullet}_\mathrm{loc}(\FsCM)$, which are given by
\begin{align*}
    &\ts^\bullet
    = \langle\qs^\dagger, \delta \qs\rangle \dt + \langle\gs^\flat(\dqs), \delta \qs\rangle,\\
    &\Ls^\bullet = \frac{1}{2} \gs(\dqs,\dqs) \dt,
\end{align*}
and the cohomological vector field $\Qs \in \mathfrak{X}_{\text{evo}}(\FsCM)$
\begin{align*}
    &\Qs \qs = 0,&
    &\Qs \qs^\dagger = +\frac{1}{2}\partial\gs(\dqs,\dqs) - \frac{d}{dt}(\gs^\flat)(\dqs) - \gs^\flat(\ddqs),
\end{align*}
where\footnote{Observe that, in a local chart, we have $\partial\gf(\dqf,\dqf)=\partial_\rho \gf_{\mu\nu}\dqf^\mu \dqf^\nu$, while $(\frac{d}{dt}\gf^\flat)(\dqf)=\dqf^\rho \partial_\rho\gf_{\mu\nu}\dqf^\mu\dqf^\nu$.} $\partial \gs \coloneqq \frac{\delta \gs}{\delta \qs}$,
defines a lax BV-BFV theory.
\end{proposition/definition}

\begin{proof}
We need to check Equations \eqref{eq:LaxStructure} at codimension 0 and 1, namely
\begin{align*}
    &\iota_\Qs \vps^0 = \delta \Ls^0 + \d \ts^1,&
    &\iota_\Qs \iota_\Qs \vps^0 = 2\d \Ls^1,\\
    &\iota_\Qs \vps^1 = \delta \Ls^1,&
    &\iota_\Qs \iota_\Qs \vps^1 = 0.
\end{align*}
Note that the Lagrangian only has a top-form term $L^\bullet_1 = L^0_1$. The only non-trivial equation is $\iota_Q \vps^0 = \delta \Ls^0 + \d \ts^1$, since $\Ls^1 = \iota_\Qs \iota_\Qs \vps^0 = \iota_\Qs \vps^1 = 0$. We compute 
\begin{align*}
    \iota_\Qs \vps^0 
    &= \langle \Qs \qs^\dagger, \delta \qs \rangle \dt
    = \langle\left(\frac{1}{2}\partial\gs(\dqs,\dqs) - (\frac{d}{dt} \gs^\flat)(\dqs) - \gs^\flat(\ddqs)\right),\delta \qs \rangle \dt \\
    \delta \Ls^0 &= \frac{1}{2} \delta \gs(\dqs,\dqs) \dt + \gs(\delta \dqs,\dqs) \dt\\
    &= \frac{1}{2} \langle \partial \gs(\dqs,\dqs) , \delta \qs \rangle \dt 
    + \frac{\d}{\dt}(\gs(\dqs,\dqs)) \dt
    -\frac{\d}{\dt}(\gs)(\delta\qs,\dqs) \dt
    - \gs(\delta\qs,\ddqs)\dt\\
    &= \frac{1}{2} \langle \partial \gs(\dqs,\dqs) , \delta \qs \rangle \dt 
    - \d \langle \gs^\flat(\dqs),\delta \qs \rangle
    - \langle\frac{d}{dt} \gs^\flat(\dqs),\delta \qs\rangle \dt
    - \langle \gs^\flat(\ddqs)),\delta \qs \rangle \dt \\
    &=- \d \langle \gs^\flat(\dqs),\delta \qs \rangle
    + \left\langle \frac{1}{2} \partial \gs(\dqs,\dqs) - \frac{d}{dt} \gs^\flat(\dqs) -\gs^\flat(\ddqs)) , \delta \qs \right\rangle \dt \\
    &= - \d \ts^1 + \iota_\Qs \vps^0,
\end{align*}
where we used $|\dt| = -1$ and $|\delta \qs| = -1$.
\end{proof}

In the case of the first-order theory we have the lax BV-BFV theory:

\begin{proposition/definition}
The data $$\mathfrak{F}^\mathrm{lax}_{1CM} = (\FfCM,\tf^\bullet,\Lf^\bullet,\Qf)$$ where
\begin{align*}
    \FfCM = T^*[-1](C^\infty(I,M) \times C^\infty(I,M) ),
\end{align*}
together with $\tf^\bullet \in \Omega^{1,\bullet}_\mathrm{loc}(\FfCM)$ and $\Lf^\bullet \in \Omega^{0,\bullet}_\mathrm{loc}(\FfCM)$, which take the forms
\begin{align*}
    &\tf^\bullet
    = (\langle\qf^\dagger, \delta \qf\rangle + \langle\pf^\dagger, \delta \pf\rangle) \dt + \langle\pf, \delta \qf\rangle,\\
    &\Lf^\bullet = \left(\langle\pf, \dqf \rangle - \frac{1}{2} \hf (\pf,\pf)\right) \dt,
\end{align*}
and the cohomological vector field $\Qf \in \mathfrak{X}_{\text{evo}}(\FfCM)$
\begin{align*}
    &\Qf \qf = \Qf \pf = 0,&
    &\Qf \qf^\dagger = - \dpf - \frac{1}{2} \partial\hf(\pf,\pf),&
    &\Qf \pf^\dagger = \dqf- \hf^\sharp(\pf),
\end{align*}
with {$\partial \hf \coloneqq \delta \hf / \delta \qf$}, defines a lax BV-BFV theory.
\end{proposition/definition}

\begin{proof}
Again we need to check Equations \eqref{eq:LaxStructure} at codimension 0 and 1. Explicitly we have
\begin{align*}
    &\iota_\Qf \vpf^0 = \delta \Lf^0 + \d \tf^1,&
    &\iota_\Qf \iota_\Qf \vpf^0 = 2\d \Lf^1,\\
    &\iota_\Qf \vpf^1 = \delta \Lf^1,&
    &\iota_\Qf \iota_\Qf \vpf^1 = 0.
\end{align*}
As in the second-order theory, the Lagrangian only has a top-form component. The only non-trivial equation is $\iota_Q \vpf^0 = \delta \Lf^0 + \d \tf^1$, since $\Lf^1 = \iota_\Qf \iota_\Qf \vpf^0 = \iota_\Qf \vpf^1 = 0$. We compute
\begin{align*}
    \iota_\Qf \vpf^0 
    &= (\langle \Qf \qf^\dagger , \delta \qf \rangle+ \langle\Qf \pf^\dagger, \delta \pf\rangle) \dt
    = -\langle \dpf  + \frac{1}{2} \partial\hf (\pf,\pf), \delta \qf \rangle\dt ,
    + \langle (\dqf- \hf^\sharp(\pf)), \delta \pf \rangle \dt,\\
    \delta \Lf^0 &= 
    \left(\langle\pf ,\delta \dqf\rangle- \frac{1}{2} \langle \partial \hf (\pf,\pf), \delta q\rangle\right) \dt
    + \langle \dqf- \hf^\sharp(\pf), \delta \pf \rangle\dt\\
    &= \frac{\d}{\dt}\langle\pf, \delta \qf \rangle \dt
    -\langle\dpf  + \frac{1}{2} \partial \hf( \pf,\pf),\delta \qf \rangle  \dt 
    + \langle \dqf- \hf^\sharp(\pf), \delta \pf \rangle\dt\\\
    &= - \d \tf^1 + \iota_\Qf \vpf^0.
\end{align*}
\end{proof}

We now present the main theorem of this section, together with an outline of its proof.\footnote{{For a similar conclusion to the one presented here, see \cite[Section 3.2]{GrigorievParent11}. Our method allows to additionally construct explicit chain maps that implement the quasi-isomorphism.}} The computational details and the various lemmata are presented afterwards.

\begin{theorem}
\label{theorem:laxequivalenceCM}
The lax BV-BFV theories $\mathfrak{F}^\mathrm{lax}_{1CM}$ and $\mathfrak{F}^\mathrm{lax}_{1CM}$ of the first-order and second-order formulations of classical mechanics with a background metric are lax BV-BFV equivalent.
\end{theorem}

\begin{proof}
We need to check all the conditions from Definition \ref{def:laxequivalence}. The existence of two maps $\phi$, $\psi$ with the desired properties is presented in Lemmata \ref{lem:philaxCM} and \ref{lem:psilaxCM} respectively, where we also show that the pullback maps $\phi^*$, $\psi^*$ are chain maps w.r.t. the BV-BFV complexes $\bvbfv^\bullet_i$, and that they map $(\theta^\bullet_i,L^\bullet_i)$ in the desired way. 

Furthermore, we need to show that the respective BV-BFV complexes are quasi-isormophic. The composition map $\lambda^* = \phi^* \circ \psi^*$ is shown to be the identity in Lemma \ref{lem:lambdaCM}. In Lemma \ref{lem:chiCM}, we prove that the composition map $\chi^* = \psi^* \circ \phi^*$,
is homotopic to the identity by following the strategy presented in Section \ref{sec:Strategy}.

In Lemma \ref{lem:chiIdCohomologyCM} we demonstrate that $\chi^*$ is the identity in cohomology by showing that the map
$$h_\chi \varphi^j_1 = \int_0^\infty e^{s\mathcal{L}_{\Df}} \mathcal{L}_{\Rf} \varphi^j_1 ds$$ 
satisfying $\chi^* - \mathrm{id}_{1} = (\mathcal{L}_{\Qf} - \d) h_\chi + h_\chi (\mathcal{L}_{\Qf} - \d)$ (cf.\ Lemma \ref{lem:hchiexplicitform}) converges, therefore proving that the two lax BV-BFV theories in question have isomorphic BV-BFV cohomologies
\begin{align*}
    H^\bullet(\bvbfv_{1CM}) \simeq H^\bullet(\bvbfv_{2CM})
\end{align*}
and thus that they are lax BV-BFV equivalent.
\end{proof}

Let us now look at the computations in detail. We start with the chain maps:

\begin{lemma}
\label{lem:philaxCM}
Let $\phi: \FsCM \rightarrow \FfCM$ be the map defined through
\begin{align*}
    &\phi^* \qf = \qs,&
    &\phi^* \pf = {\phi^* \pf = \gs^\flat(\dqs),}\\
    &\phi^* \qf^\dagger = \qs^\dagger,&
    &\phi^* \pf^\dagger = 0.
\end{align*}
Its pullback map $\phi^*$ is a chain map w.r.t. $(\mathcal{L}_{Q_i}-\d)$ and maps the lax BV-BFV data of the first-order theory as
\begin{align*}
    &\phi^* \tf^\bullet = \ts^\bullet,&
    &\phi^* \Lf^\bullet = \Ls^\bullet.
\end{align*}
\end{lemma}

\begin{proof}
To check that $\phi$ is a chain map we compute
\begin{align*}
    \phi^* \Qf \qf &= \phi^*(0) = 0 = \Qs \qs = \Qs \phi^* \qf,\\
    \phi^* \Qf \pf &= \phi^*(0) = 0 = \Qs (\gs^\flat(\dqs)) = \Qs \phi^* \pf ,\\
    \phi^* \Qf \qf^\dagger
    &= \phi^*\left(- \dpf - \frac{1}{2} \partial\hf(\pf,\pf)\right)
    = - \frac{\d}{\dt}(\gs^\flat(\dqs)) - \frac{1}{2} \partial\hs (\gs^\flat(\dqs),\gs^\flat(\dqs))\\
    &= - \frac{d}{dt}(\gs^\flat)(\dqs) - \gs^\flat(\ddqs) + \frac{1}{2} \partial\gs (\dqs,\dqs)
    = \Qs \qs^\dagger = \Qs \phi^* \qf^\dagger,\\
    \phi^* \Qf \pf^\dagger &= \phi^*(\dqf - \hf^\sharp(\pf)) = 0
    = \Qs \phi^* \pf^\dagger
\end{align*}
Together with Proposition \ref{prop:phipsichainmapfromQ}, this shows that $\phi^*$ is a chain map w.r.t. $(\mathcal{L}_{Q_i} - \d)$. In turn, $\phi^*$ acts on $\tf^\bullet,$ $\Lf^\bullet$ as
\begin{align*}
    \phi^* \tf^\bullet 
    &= \phi^*\left((\langle\qf^\dagger, \delta \qf\rangle + \langle\pf^\dagger, \delta \pf\rangle) \dt + \langle\pf, \delta \qf\rangle\right)
    = \langle\qs^\dagger, \delta \qs\rangle \dt + \langle\gs^\flat(\dqs), \delta \qs\rangle
    = \ts^\bullet,\\
    \phi^*\Lf^\bullet 
    &= \phi^* \left(\langle\pf, \dqf \rangle - \frac{1}{2} \hf (\pf,\pf)\right) \dt
    = \left(\gs(\dqs,\dqs) -\frac{1}{2} \hs(\gs^\flat(\dqs),\gs^\flat(\dqs))\right) \dt\\
    &= \frac{1}{2} \gs( \dqs,\dqs) \dt = \Ls^\bullet.
\end{align*}
\end{proof}

\begin{lemma}
\label{lem:psilaxCM}
Let $\psi: \FfCM \rightarrow \FsCM$ be the map defined through
\begin{align*}
    &\psi^* \qs = \qf,&
    &\psi^* \qs^\dagger 
    { = \qf^\dagger 
    - \frac{1}{2} \partial\gf(\pf^\dagger, \Qf \pf^\dagger)
    -\frac{d}{dt}(\gf^\flat(p^\dag)) + \partial \gf(\dqf,\pf^\dagger)}
\end{align*}
Its pullback map $\psi^*$ is a chain map w.r.t.\ $(\mathcal{L}_{Q_i}-\d)$ and maps the lax BV-BFV data of the second-order theory as
\begin{align*}
    &\psi^* \ts^\bullet  = \tf^\bullet  + (\mathcal{L}_{\Qf} - \d) \betaf^\bullet + \delta \ff^\bullet,&
    &\psi^* \Ls^\bullet  = \Lf^\bullet  + (\mathcal{L}_{\Qf} - \d) \iota_{\Qf} \betaf^\bullet + \d \ff^\bullet.
\end{align*}
where
\begin{align*}
    \betaf^\bullet &= \frac{1}{2} \gf (\pf^\dagger, \delta \pf^\dagger) \dt
    + \gf(\pf^\dagger, \delta \qf),&
    \ff^\bullet &= - \frac{1}{2} \gf(\pf^\dagger, \Qf \pf^\dagger) \dt.
\end{align*}
\end{lemma}

\begin{proof}
The only non-trivial calculation that is needed to check if $\psi^*$ is a chain map w.r.t. $(\mathcal{L}_{Q_i}-\d)$ is $\Qf \psi^* \qs^\dagger = \psi^* \Qs \qs^\dagger$. We compute 
\begin{align*}
    \Qf \psi^* \qs^\dagger
    &= \Qf \qf^\dagger 
    - \frac{1}{2} \partial\gf' (\Qf \pf^\dagger,\Qf \pf^\dagger) - \gf^\flat(\Qf \dpf^\dagger) - (\dgf^{\flat})(\Qf \pf^\dag) + \partial \gf(\dqf,\Qf,\pf^\dag)\\
    &= -\dpf - \frac12 \partial \hf(\pf,\pf) - \frac12 \gf(\dqf,\dqf) - \frac12 \partial \gf(\hf^\sharp(\pf),\hf^\sharp(\pf)) \\
    & \phantom{=} +\partial \gf(\dqf,\hf^\sharp(\pf)) - \gf^\flat\left(\ddqf - \dot{\hf}^\sharp(\pf) - \hf^\sharp(\dpf)\right) \\
    & \phantom{=} - \dgf^\flat\left(\dqf - \hf^\sharp(\pf)\right) + \partial \gf(\dqf,\dqf) - \partial \gf(\dqf,\hf^\sharp(\pf))\\
    &= -\gf^\flat(\ddqf) - \dgf^\flat(\dqf) + \frac12 \partial \gf(\dqf,\dqf) = \psi^* \Qs \qs^\dagger.
\end{align*}
where we used $\dgf^\flat(\hf^\sharp(p)) = - \gf^\flat(\dot{\hf}^\sharp(p))$, in virtue of the fact that $\gf^\flat\circ\hf^\sharp = \mathrm{id}$.

Let $\Delta \ts^\bullet \coloneqq \psi^* \ts^\bullet - \tf^\bullet$ and  $\Delta \Ls^\bullet \coloneqq \psi^* \Ls^\bullet - \Lf^\bullet$. We need to check whether
\begin{align*}
    &\Delta \ts^0  = \mathcal{L}_{\Qf} \betaf^0 - \d\betaf^1  + \delta \ff^0,\\
    &\Delta \ts^1  = \mathcal{L}_{\Qf} \betaf^1 + \delta \ff^0,\\
    &\Delta \Ls^0 =  \mathcal{L}_{\Qf} \iota_{\Qf} \betaf^0 - \d \iota_{\Qf} \betaf^1 + \d \ff^1.
\end{align*}
Recall that we only need to compute $\Delta \Ls^0$, as $\Delta \Ls^k$ for $k>0$ are determined through $\Delta \ts^\bullet$ (cf.\ Proposition \ref{prop:LaxRedundancyTransformation}, Remark \ref{rem:WhatToCheck}). Furthermore, note that $\ff^1 = 0$.

\textbf{Computation of $\Delta \ts^0$:} For the first equation we compute
\begin{align*}
    \Delta \ts^0 
    &= \psi^*(\langle \qs^\dagger, \delta \qs\rangle \dt) - (\langle\qf^\dagger, \delta \qf\rangle + \langle\pf^\dagger, \delta \pf\rangle ) \dt\\
    &= \left[\langle \qf^\dagger 
    - \frac{1}{2} \partial \gf (\pf^\dagger ,\Qf \pf^\dagger)  
    - \gf^\flat(\dpf^\dagger)  
    -\frac{d}{dt}(\gf^\flat)(p^\dag) 
    + \partial \gf(\dqf,\pf^\dagger)
    - \qf^\dagger , \delta \qf\rangle 
    - \langle \pf^\dagger, \delta \pf \rangle\right] \dt\\
    &= \left[\langle 
    - \frac{1}{2} \partial \gf (\pf^\dagger ,\Qf \pf^\dagger)  
    - \gf^\flat(\dpf^\dagger)  
    -\frac{d}{dt}(\gf^\flat)(p^\dag) 
    + \partial \gf(\dqf,\pf^\dagger), \delta \qf\rangle 
    - \langle \pf^\dagger, \delta \pf \rangle\right] \dt\\
    \mathcal{L}_{\Qf} \betaf^0
    &= \mathcal{L}_{\Qf}\left(\frac{1}{2} \gf( \pf^\dagger, \delta \pf^\dagger) \dt \right)
    =\left[
    \frac{1}{2} \gf (\Qf \pf^\dagger, \delta \pf^\dagger )
    + \frac{1}{2} \gf( \pf^\dagger, \delta \Qf \pf^\dagger )
    \right] \dt, \\
    \d \betaf^1 
    &= \dt \frac{\d }{\dt}\left(\gf (\pf^\dagger ,\delta \qf) \right)
    = \left[
    \dgf( \pf^\dagger, \delta \qf)
    + \gf (\dpf^\dagger ,\delta \qf)
    + \gf (\pf^\dagger, \delta \dqf)
    \right] \dt,\\
    \delta \ff^0 
    &= \delta \left( - \frac{1}{2} \gf (\pf^\dagger, \Qf \pf^\dagger) \dt \right)\\
    &= \left[
    + \frac{1}{2} \langle \partial\gf( \pf^\dagger ,\Qf \pf^\dagger) , \delta \qf \rangle
    - \frac{1}{2} \gf (\delta \pf^\dagger, \Qf \pf^\dagger)
    + \frac{1}{2} \gf (\pf^\dagger ,\delta \Qf \pf^\dagger)
    \right] \dt.
\end{align*}

Using 
\begin{align*}
    &\delta \Qf \qf^\dagger 
    = \delta(\dqf - \hf^\sharp(\pf))
    = \delta\dqf - \langle \partial\hf^\sharp(\pf),\delta\qf\rangle 
    -\hf^\sharp(\delta \pf),\\
    &\gf(\cdot,\langle \partial\hf^\sharp(\cdot),\delta\qf\rangle) = - \langle \partial \gf (\cdot, \hf^\sharp(\cdot)),\delta \qf \rangle,
\end{align*}
the last three terms together yield
\begin{align*}
    &\mathcal{L}_{\Qf} \betaf^0 - \d \betaf^1 + \delta \ff^0 =
    \frac{1}{2} \gf (\Qf \pf^\dagger, \delta \pf^\dagger )
    + \frac{1}{2} \gf( \pf^\dagger, \delta \Qf \pf^\dagger ) \\
    =& \Big[
    \gf(\pf^\dagger, \delta \Qf \pf^\dagger)
    - \dgf( \pf^\dagger, \delta \qf)
    - \gf (\dpf^\dagger ,\delta \qf)
    - \gf (\pf^\dagger, \delta \dqf)
    + \frac{1}{2} \langle \partial\gf( \pf^\dagger ,\Qf \pf^\dagger) , \delta \qf \rangle
    \Big] \dt\\
    =& \Big[
    \cancel{\gf(\pf^\dagger, \delta\dqf)}
    - \gf(\pf^\dagger,\langle \partial\hf^\sharp(\pf),\delta\qf\rangle)
    -\gf(\pf^\dagger,\hf^\sharp(\delta \pf))
    -\dot \gf (\pf^\dagger,\delta \qf)
    - \gf(\dpf^\dagger, \delta \qf)
    -\cancel{\gf(\pf^\dagger, \delta\dqf )}
    + \frac{1}{2} \langle \partial \gf(\pf^\dagger, \Qf \pf^\dagger). \delta \qf \rangle\Big] \dt\\
    =& 
    \langle \partial \gf (\pf^\dagger, \hf^\sharp(\pf)),\delta \qf \rangle
    - \langle \pf^\dagger, \delta \pf \rangle
    -\dot \gf (\pf^\dagger,\delta \qf)
    - \gf(\dpf^\dagger, \delta \qf)
    + \frac{1}{2} \langle \partial \gf(\pf^\dagger, \Qf \pf^\dagger). \delta \qf \rangle\Big] \dt\\
    =& 
    \langle \partial g (\pf^\dagger, \dqf - \Qf \pf^\dagger)
    -\dot \gf^\flat (\pf^\dagger)
    - \gf^\flat(\dpf^\dagger) 
    + \frac{1}{2} \langle \partial \gf(\pf^\dagger, \Qf \pf^\dagger), \delta \qf \rangle
    - \langle \pf^\dagger, \delta \pf \rangle
    \Big] \dt\\
    =& 
    \langle 
    -\frac{1}{2} \langle \partial \gf(\pf^\dagger, \Qf \pf^\dagger)
    + \partial g (\pf^\dagger, \dqf)
    - \dot \gf^\flat (\pf^\dagger)
    - \gf^\flat(\dpf^\dagger) 
    , \delta \qf \rangle
    - \langle \pf^\dagger, \delta \pf \rangle
    \Big] \dt = \Delta \ts^0.
\end{align*}

\textbf{Computation of $\Delta \ts^1$:} In this case we have
\begin{align*}
    \Delta \ts^1 
    &= \psi^*(\langle \gs^\flat(\dqs), \delta \qs\rangle) - \langle\pf ,\delta \qf\rangle
    = \langle g^\flat(\dqf - h^\sharp(p)), \delta \qf\rangle
    = \gf(\Qf \pf^\dagger ,\delta \qf),\\
    \mathcal{L}_{\Qf} \betaf^1 
    &= \mathcal{L}_{\Qf}\left( \gf( \pf^\dagger, \delta \qf)\right) 
    = \gf( \Qf \pf^\dagger, \delta \qf),\\
    \delta \ff^1 & = 0,
\end{align*}
thus showing $\Delta \ts^1 = \mathcal{L}_{\Qf} \betaf^1 + \delta \ff^1$. 

\textbf{Computation of $\Delta \Ls^0$:}
For the third equality, we compute
\begin{align*}
    \Delta \Ls^0 
    &= \psi^*\left(\frac{1}{2} \gs (\dqs,\dqs) \dt \right) 
    - \left(\langle\pf, \dqf \rangle- \frac{1}{2}\hf(\pf,\pf) \right) \dt
    = \left[
    \frac{1}{2}  \gf (\dqf,\dqf)
    - \langle\pf, \dqf \rangle
    + \frac{1}{2}\hf(\pf,\pf)
    \right] \dt,\\
    \mathcal{L}_{\Qf} \iota_{\Qf} \betaf^0
    &= \mathcal{L}_{\Qf} \iota_{\Qf} \left(\frac{1}{2} \gf (\pf^\dagger, \delta \pf^\dagger) \dt \right)
    = \frac{1}{2} \gf (\Qf \pf^\dagger,\Qf \pf^\dagger) \dt
    = \frac{1}{2} \gf (\dqf - \hf^\sharp( \pf),\dqf - \hf^\sharp (\pf)), \dt \\
    &= \left[\frac{1}{2}\gf(\dqf,\dqf) -  \langle \pf,\dqf\rangle + \frac{1}{2} \hf( \pf, \pf) \right] \dt,\\
    \d \iota_{\Qf} \betaf^1 
    &= \d \left(\gf (\pf^\dagger ,\Qf \qf )\right)
    = 0,\\
    \d \ff^1
    &= 0,
\end{align*}
as desired.
\end{proof}

\begin{lemma}
\label{lem:lambdaCM}
The composition map $\lambda^* = \phi^* \circ \psi^*: \bvbfv^\bullet_2 \rightarrow \bvbfv^\bullet_2$ is the identity
\begin{align*}
    &\lambda^* \qs = \qs,&
    &\lambda^*\qs^\dagger = \qs^\dagger
\end{align*}
and as such the identity in cohomology.
\end{lemma}

\begin{proof}
Using $\lambda^* = \phi^* \circ \psi^*$ we have
\begin{align*}
    \lambda^* \qs 
    &= \phi^* \circ \psi^* (\qs)
    = \phi^* \qf
    = \qs,\\
    \lambda^* \qs^\dagger
    &= \phi^* \circ \psi^* (\qs^\dagger)
    =\phi^*\left(\qf^\dagger 
    - \frac{1}{2} \partial\gf(\pf^\dagger, \Qf \pf^\dagger)
    - \gf^\flat(\dpf^\dagger)-\frac{d}{dt}(\gf^\flat)(p^\dag) + \partial \gf(\dqf,\pf^\dagger))\right) = \qs^\dagger,
\end{align*}
as $\phi^* \pf^\dagger = 0$.
\end{proof}

\begin{lemma}
\label{lem:chiCM}
The composition map $\chi^* = \psi^* \circ \phi^*: \bvbfv^\bullet_1 \rightarrow \bvbfv^\bullet_1$ acts as
\begin{align*}
    &\chi^*\qf = \qf,&
    &\chi^*\pf = {\gf^\flat(\dqf)},\\
    &\chi^*\qf^\dagger { = \qf^\dagger 
    - \frac{1}{2} \partial\gf(\pf^\dagger, \Qf \pf^\dagger)
    - \gf^\flat(\dpf^\dagger) -\frac{d}{dt}(\gf^\flat)(p^\dag) + \partial \gf(\dqf,\pf^\dagger)},&
    &\chi^*\pf^\dagger = 0.
\end{align*}
and is homotopic to the identity.
\end{lemma}

\begin{proof}
By definition $\chi^* = \psi^* \circ \phi^*$. Then
\begin{align*}
    &\chi^*\qf 
    = \psi^* \circ \phi^*(\qf)
    = \psi^* \qs,\\
    &\chi^*\pf 
    = \psi^* \circ \phi^*(\pf)
    = \psi^* (\gs \dqs)
    = { = \psi^* (\gs^\flat( \dqs))
    = \gf^\flat( \dqf)},\\
    &\chi^*\qf^\dagger 
    = \psi^* \circ \phi^*(\qf^\dagger)
    = \psi^* \qs^\dagger
    { = \qf^\dagger 
    - \frac{1}{2} \partial\gf(\pf^\dagger, \Qf \pf^\dagger)
    - \gf^\flat(\dpf^\dagger) -\frac{d}{dt}(\gf^\flat)(p^\dag) + \partial \gf(\dqf,\pf^\dagger)},\\
    &\chi^*\pf^\dagger 
    = \psi^* \circ \phi^*(\pf^\dagger)
    = \psi^* (0)
    = 0.
\end{align*}
In order to prove that $\chi^*$ is homotopic to the identity we first compute $\chi^*_s = e^{s\mathcal{L}_{\Df}}$, where $\Df = [\Qf,\Rf]$, and show $\lim_{s\rightarrow \infty} \chi^*_s = \chi^*$. We choose $\Rf$ to act as 
\begin{align*}
    &\Rf\qf = 0,&
    &\Rf\pf = {\gf^\flat(\pf^\dagger)},&
    &\Rf\qf^\dagger =0,&
    &\Rf\pf^\dagger =0.
\end{align*}
For $\qf$ we have
\begin{align*}
    && \Df \qf &= [\Qf,\Rf] \qf = 0, && \\
    &\Rightarrow& 
    e^{s\mathcal{L}_{\Df}} \qf &= \qf, &&\\
    &\Rightarrow& 
    \lim_{s\rightarrow \infty} e^{s\mathcal{L}_{\Df}} \qf &= \qf = \chi^* \qf. &&
\end{align*}
For $\pf$:
\begin{align*}
    && \Df\pf &= \Qf \Rf \pf = \Qf(\gf^\flat(\pf^\dagger)) 
    = \gf^\flat(\dqf)- \pf,&&\\
    &\Rightarrow&
    \Df^2 \pf &= \Df(\gf^\flat(\dqf) - \pf) 
    = - \Df \pf = - (\gf^\flat(\dqf) - \pf),&&\\
    &\Rightarrow&
    \Df^k \pf &= -(-1)^{k} (\gf^\flat(\dqf) - \pf) \qquad \text{for } k \geq 1,&&\\
    &\Rightarrow&
    e^{s\mathcal{L}_{\Df}} \pf 
    &= \pf + \sum^\infty_{k=1} \frac{s^k}{k!} \Df^k \pf
    = \pf - \sum^\infty_{k=1} \frac{(-s)^k}{k!} (\gf^\flat(\dqf) - \pf)\\
    &&&= \pf - (e^{-s} - 1) (\gf^\flat(\dqf) - \pf)&&\\
    &\Rightarrow& 
    \lim_{s\rightarrow \infty} e^{s\mathcal{L}_{\Df}} \pf 
    &= \gf^\flat(\dpf)
    = \chi^* \pf. &&
\end{align*}
For $\pf^\dagger$:
\begin{align*}
    && \Df\pf^\dagger 
    &= \Rf\Qf \pf^\dagger 
    {= \Rf (\dqf - \hf^\sharp(\pf))
    = -\hf^\sharp(\gf^\flat (\pf^\dagger))}
    = -\pf^\dagger, &&\\
    &\Rightarrow& \Df^k \pf^\dagger &= (-1)^{k} \pf^\dagger
    \qquad \text{for } k \geq 0,&&\\
    &\Rightarrow&
    e^{s\mathcal{L}_{\Df}} \pf^\dagger
    &= \pf^\dagger + \sum^\infty_{k=1} \frac{s^k}{k!} \Df^k \pf^\dagger
    = \pf + \sum^\infty_{k=1} \frac{(-s)^k}{k!} \pf^\dagger
    = e^{-s} \pf^\dagger&&\\
    &\Rightarrow& 
    \lim_{s\rightarrow \infty} e^{s\mathcal{L}_{\Df}} \pf^\dagger 
    &= 0
    = \chi^* \pf^\dagger.&&
\end{align*}
For $\qf^\dagger$:
\begin{align*}
    \Df \qf^\dagger &= \Rf\Qf\qf^\dagger
    = -\Rf\left(\dpf + \frac{1}{2} \partial\hf(\pf,\pf)\right)
    = -\frac{\d}{\dt}(g^\flat(\pf^\dagger)) - \partial\hf(\gf^\flat(\pf^\dagger), \pf)\\
    &= -\dgf^\flat(\pf^\dagger) - \gf^\flat(\dpf^\dagger) - \partial \hf(\pf,\gf^\flat(\pf^\dagger)) \\
    &= -\dgf^\flat(\pf^\dagger) - \gf^\flat(\dpf^\dagger) - \partial \gf(\Qf\pf^\dagger,\pf^\dagger) + \partial \gf(\dqf,\pf^\dagger),
\end{align*}
where we used $p= \gf^\flat(  \dqf - Q \pf^\dagger)$, and $\partial \hf(\gf^\flat(\cdot),\gf^\flat(\cdot)) = -\partial\gf(\cdot,\cdot)$.

We see that
\begin{align*}
    \Df(\dgf^\flat(\pf^\dagger) + \gf^\flat(\dpf^\dagger) - \partial \gf(\dqf,\pf^\dagger)) 
    &= (\dgf^\flat(\Df\pf^\dagger) + \gf^\flat(\frac{d}{dt}(\Df\pf^\dagger)) - \partial \gf(\dqf,\Df\pf^\dagger)) \\
    &= - (\dgf^\flat(\pf^\dagger) + \gf^\flat(\dpf^\dagger) - \partial \gf(\dqf,\pf^\dagger)),
\end{align*}
so that, for all $k\geq0$
\begin{align*}
    \Df^k(\dgf^\flat(\pf^\dagger) + \gf^\flat(\dpf^\dagger) - \partial \gf(\dqf,\pf^\dagger)) &= (-1)^k (\dgf^\flat(\pf^\dagger) + \gf^\flat(\dpf^\dagger) - \partial \gf(\dqf,\pf^\dagger)) 
\end{align*}
and (recalling that $[\Df,\Qf] = 0$)
\begin{align*}
    \Df(\partial\gf(\pf^\dagger, \Qf \pf^\dagger))
    &= \partial\gf(\Df \pf^\dagger, \Qf \pf^\dagger)
    + \partial\gf(\pf^\dagger, \Qf \Df \pf^\dagger)
    = - 2 \partial\gf(\pf^\dagger, \Qf \pf^\dagger),
\end{align*}
we have
\begin{align*}
    \Df^k(\partial\gf (\pf^\dagger, \Qf \pf^\dagger))
    &= (-2)^k \partial\gf(\pf^\dagger, \Qf \pf^\dagger)
\end{align*}
for $k\geq 0$, which ultimately yields
\begin{align*}
\Df^k \qf^\dagger
    &= - \Df^{k-1}(\dgf^\flat(\pf^\dagger) + \gf^\flat(\dpf^\dagger) - \partial \gf(\dqf,\pf^\dagger))
    - \Df^{k-1}(\partial\gf (\pf^\dagger, \Qf \pf^\dagger))\\
    &= - (-1)^{k-1} (\dgf^\flat(\pf^\dagger) + \gf^\flat(\dpf^\dagger) - \partial \gf(\dqf,\pf^\dagger))
    - (-2)^{k-1} (\partial\gf (\pf^\dagger, \Qf \pf^\dagger))\\
    &= (-1)^k (\dgf^\flat(\pf^\dagger) + \gf^\flat(\dpf^\dagger) - \partial \gf(\dqf,\pf^\dagger))
    + \frac{1}{2} (-2)^{k} (\partial\gf (\pf^\dagger, \Qf \pf^\dagger))
\end{align*}
for all $k\geq 0$, so that 
\begin{align*}
    e^{s\mathcal{L}_{\Df}} \qf^\dagger &= \qf^\dagger 
    + \frac{1}{2} (e^{-2s}-1) (\partial\gf (\pf^\dagger, \Qf \pf^\dagger))
    + (e^{-s}-1)(\dgf^\flat(\pf^\dagger) - \gf^\flat(\dpf^\dagger) + \partial \gf(\dqf,\pf^\dagger))
\end{align*}
\begin{align*}
    \lim_{s\rightarrow \infty} e^{s\mathcal{L}_{\Df}} \qf^\dagger 
    &= \qf^\dagger 
    - \frac{1}{2} (\partial\gf (\pf^\dagger, \Qf \pf^\dagger)) - (\dgf^\flat(\pf^\dagger) + \gf^\flat(\dpf^\dagger) - \partial \gf(\dqf,\pf^\dagger))
    = \chi^* \pf^\dagger.&&
\end{align*}
All in all, the homotopy $\chi^*_s$ takes the form
\begin{align*}
    &\chi^*_s \qf = \qf,\\
    &\chi^*_s \pf 
    = e^{-s}\pf - (e^{-s} - 1) {\gf^\flat(\dqf)},\\
    &\chi^*_s \qf^\dagger = \qf^\dagger 
    + {\frac{1}{2} (e^{-2s}-1) \partial\gf (\pf^\dagger ,\Qf \pf^\dagger)
    + (e^{-s}-1) (\dot \gf^\flat( \pf^\dagger) 
    - \gf^\flat(\dot\pf^\dagger)
    +\partial \gf (\dqf,\pf^\dagger))},\\
    &\chi^*_s \pf^\dagger
    = e^{-s} \pf^\dagger, 
\end{align*}
and clearly satisfies $\lim_{s\rightarrow \infty} \chi^*_s = \chi^*$.
\end{proof}

\begin{lemma}
\label{lem:chiIdCohomologyCM}
The map $\chi^*$ is the identity in cohomology.
\end{lemma}

\begin{proof}
We have to check if the map $h_\chi$ converges on $\FfCM$, namely
\begin{align*}
    h_\chi \varphi^j_1 = \int^\infty_0 e^{s \Df} \Rf \varphi^j_1 \, \d s < \infty 
    \quad \forall \varphi^j_1 \in \FfCM.
\end{align*}
As $\{\qf,\qf^\dagger, \pf^\dagger\} \in \ker \Rf$ we have
\begin{align*}
    &h_\chi \qf = 0,&
    &h_\chi \qf^\dagger = 0,&
    &h_\chi \pf^\dagger  = 0.
\end{align*}
For $h_\chi \pf^\dagger$ we compute
\begin{align*}
    h_\chi \pf
    = \int^\infty_0 e^{s\mathcal{L}_{\Df}} \Rf \pf \,\d s
    = {\int^\infty_0 e^{s\mathcal{L}_{\Df}} (\gf^\flat( \pf^\dagger)) \d s
    = \gf^\flat( \pf^\dagger) \int^\infty_0 e^{-s} \d s
    = \gf^\flat( \pf^\dagger)},
\end{align*}
thus, by Proposition \ref{prop:hchiconvergence}, $h_\chi$ converges on $\bvbfv^\bullet_1$ and $\chi^*$ is the identity in cohomology.
\end{proof}

\subsection{Yang--Mills theory}
\label{sec:YM}
We now look at the example of (non-abelian) Yang--Mills theory. Let $(M,g)$ be a $d$-dimensional (pseudo-)Riemanian manifold and $G$ a connected Lie group with Lie algebra $(\mathfrak{g},[\cdot,\cdot])$, endowed with an ad-invariant inner product, which for ease of notation will be denoted by means of an invariant trace operation\footnote{For a better nonperturbative behavior one usually requires $G$ to be compact, in which case one uses the Killing form as the invariant inner product. This is the motivation for using the trace notation.} $\Tr[\cdot]$. As we consider two formulations of Yang--Mills theory, we will use tildes to distinguish the fields between the two. We point out that an alternative proof of the equivalence of first- and second-order formulations of Yang--Mills theory has been given in \cite{RocekHomotopy18} using homotopy transfer of $A_\infty$-structures. We will give here an argument that is different on the surface, but which is compatible to their results. However, we stress that our analysis also includes a comparison of the boundary data of the first- and second-order formulations.

We can phrase the theory by considering connection 1-forms $\As \in \Omega^1(M,\mathfrak{g})$, with curvature $\FAs$, and the classical action functional 
\begin{align*}
    \clSs[\As] = \int_M \Tr \left[\frac{1}{2} \FAs \star \FAs \right].
\end{align*}
This is often known as the \emph{second-order formulation}.

Alternatively, one can phrase the theory in its \emph{first-order formulation}, by considering an additional ``auxiliary'' field $\Bf \in \Omega^{d-2}(M,\mathfrak{g})$ and the classical action functional
\begin{align*}
    \clSf[\Af,\Bf] 
    = \int_M \Tr \left[
    \Bf\FAf 
    - \frac{\es}{2} \Bf \star \Bf
    \right],
\end{align*}
where $\es=\pm 1$ denotes the signature of $g$.

\begin{proposition}
The first- and second-order formulation of Yang--Mills theory are classically equivalent. 
\end{proposition}

\begin{proof}
Solving the EL equations of the first-order theory w.r.t. the auxiliary field $\Bf$ gives
\begin{align*}
    \delta_{\Bf} \clSf[\Af,\Bf] 
    &= \int_M \Tr \left[
    \delta \Bf\FAf 
    - \frac{\es}{2} \delta \Bf \star \Bf
    - \frac{\es}{2} \Bf \star \delta \Bf
    \right]\\
    &=\int_M \Tr \left[
    \delta \Bf\left(\FAf 
    - \es \star \Bf\right)
    \right].
\end{align*}
Let $C_1$ be the set of solutions of $\FAf = \es \star \Bf$, or equivalently
\begin{align*}
    \star \FAf = \es \star^2 \Bf
    = \es^2 (-1)^{2(d-2)} \Bf = \Bf,
\end{align*}
where we used $\star^2 = \es(-1)^{k(d-k)}$ when acting on $k$-forms, and let $\phi_{cl}: C_1 \rightarrow F_2$ be the map defined through $\phi^*_{cl}\As = \Af$. Then
\begin{align*}
    \clSf[\Af,\Bf] \Big\vert_{C_1} 
    &= \int_M \Tr \left[ \star \FAf \FAf - \frac{\es^2}{2} \star \FAf \FAf \right]
    = \int_M \Tr \left[\frac{1}{2} \FAf \star \FAf \right]\\
    &= \phi^*_{cl}\int_M \Tr \left[\frac{1}{2} \FAs \star \FAs \right]
    = \phi^*_{cl} \clSs[\As],
\end{align*}
showing that the two theories are classically equivalent.
\end{proof}

Both first- and second-order formulations of Yang--Mills theory can be extended to lax BV-BFV theories as follows. As the symmetries of the theory are given by a Lie algebra $\mathfrak{g}$, we can follow the construction of Example \ref{ex:BRSTcase}.

\begin{proposition/definition}[\cite{cattaneo2014classical,mnev2020towards}]
The data 
$$
\mathfrak{F}^\mathrm{lax}_{2YM} = (\FsYM,\ts^\bullet,\Ls^\bullet,\Qs),
$$ 
where 
\begin{align*}
    \FsYM = T^*[-1](\Omega^{1}(M,\mathfrak{g}) \oplus \Omega^{0}(M,\mathfrak{g})[1]),
\end{align*}
together with $\ts^\bullet \in \Omega^{1,\bullet}_\mathrm{loc}(\FsYM\MS{\times M})$ and $\Ls^\bullet \in \Omega^{0,\bullet}_\mathrm{loc}(\FsYM\MS{\times M})$, which are given by
\begin{align*}
    &\ts^\bullet
    = \Tr \left[
    \As^\dagger \delta \As
    + \cs^\dagger \delta \cs
    + \delta \As \star \FAs
    + \As^\dagger \delta \cs 
    + \star \FAs \delta \cs
    \right],\\
    &\Ls^\bullet =\Tr \left[
    \frac{1}{2} \FAs \star \FAs
    + \As^\dagger \dAs \cs
    + \frac{1}{2} \cs^\dagger[\cs,\cs]
    + \star \FAs \dAs \cs
    + \frac{1}{2} \As^\dagger [\cs,\cs]
    + \frac{1}{2} \star \FAs [\cs,\cs]
    \right],
\end{align*}
and the cohomological vector field $\Qs \in \mathfrak{X}_{\text{evo}}(\FsYM)$
\begin{align*}
    &\Qs \As = \dAs \cs,&
    &\Qs \cs = \frac{1}{2} [\cs,\cs],\\
    &\Qs \As^\dagger = \dAs \star \FAs + [\cs, \As^\dagger],&
    &\Qs \cs^\dagger = \dAs \As^\dagger + [\cs,\cs^\dagger],
\end{align*}
defines a lax BV-BFV theory.
\end{proposition/definition}

In the first-order formulation we have:
\begin{proposition/definition}[\cite{cattaneo2014classical,mnev2020towards}]
The data 
$$
\mathfrak{F}^\mathrm{lax}_{1YM} = (\FfYM,\tf^\bullet,\Lf^\bullet,\Qf),
$$
where
\begin{align*}
    \FfYM = T^*[-1](\Omega^{1}(M,\mathfrak{g})\oplus \Omega^{d-2}(M,\mathfrak{g}) \oplus \Omega^{0}(M,\mathfrak{g})[1]).
\end{align*}
together with $\tf^\bullet \in \Omega^{1,\bullet}_\mathrm{loc}(\FfYM)$ and $\Lf^\bullet \in \Omega^{0,\bullet}_\mathrm{loc}(\FfYM)$, which take the form
\begin{align*}
    \tf^\bullet
    = \Tr &\left[
    \Af^\dagger \delta \Af
    + \Bf^\dagger \delta \Bf
    + \cf^\dagger \delta \cf
    + \Bf \delta \Af
    + \Af^\dagger \delta \cf 
    + \Bf \delta \cf
    \right],\\
    \Lf^\bullet = \Tr &\left[ 
    \Bf \FAf
    - \frac{\es}{2} \Bf \star \Bf
    + \Af^\dagger \dAf  \cf
    + \Bf^\dagger [\cf,\Bf]
    + \frac{1}{2} \cf^\dagger [\cf,\cf]\right.\\
    &+ \left.\Bf \dAf \cf
    + \frac{1}{2} \Af^\dagger [\cf,\cf]
    + \frac{1}{2} \Bf [\cf,\cf]
    \right]
\end{align*}
and the cohomological vector field $\Qf \in \mathfrak{X}_{\text{evo}}(\FfYM)$
\begin{align*}
    &\Qf \Af = \dAf \cf,&
    &\Qf \Bf = [\cf,\Bf],&
    &\Qf \cf = \frac{1}{2} [\cf,\cf],\\
    &\Qf \Af^\dagger = \dAf \Bf + [\cf,\Af^\dagger],&
    &\Qf \Bf^\dagger = \FAf - \es \star \Bf + [\cf,\Bf^\dagger],&
    &\Qf \cf^\dagger = \dAf \Af^\dagger + [\cf,\cf^\dagger] + [\Bf^\dagger,\Bf],
\end{align*}
defines a lax BV-BFV theory.
\end{proposition/definition}

We now present the main theorem of this section, together with an outline of the proof. The computational details and the various required Lemmata are presented in the Appendix \ref{sec:LengthyCalcYM}.

\begin{theorem}
\label{theorem:laxequivalenceYM}
The lax BV-BFV theories $\mathfrak{F}^\mathrm{lax}_{1YM}$ and $\mathfrak{F}^\mathrm{lax}_{2YM}$ of the first- and second-order formulations of Non-Abelian Yang--Mills theory are lax BV-BFV equivalent.
\end{theorem}

\begin{proof}
We need to check all the conditions from Definition \ref{def:laxequivalence}. The existence of two maps $\phi$, $\psi$ with the desired properties is presented in Lemmata \ref{lem:philaxYM} and \ref{lem:psilaxYM} respectively, where we also show that the pullback maps $\phi^*$, $\psi^*$ are chain maps w.r.t. the BV-BFV complexes $\bvbfv^\bullet_i$, and that they map $(\theta^\bullet_i,L^\bullet_i)$ in the desired way. Specifically, $\phi: \FsYM \rightarrow \FfYM$ is defined through
\begin{align*}
    &\phi^* \Af = \As,&
    &\phi^* \Bf = \star \FAs,&
    &\phi^* \cf = \cs,\\
    &\phi^* \Af^\dagger = \As^\dagger,&
    &\phi^* \Bf^\dagger = 0,&
    &\phi^* \cf^\dagger = \cs^\dagger.
\end{align*}
and maps the lax BV-BFV data of the first-order theory as
\begin{align*}
    &\phi^* \tf^\bullet = \ts^\bullet,&
    &\phi^* \Lf^\bullet = \Ls^\bullet,
\end{align*}
whereas $\psi: \FfYM \rightarrow \FsYM$ is given by
\begin{align*}
    &\psi^* \As = \Af,&
    &\psi^* \cs = \cf,\\
    &\psi^* \As^\dagger = \Af^\dagger -  \dAf \star \Bf^\dagger,&
    &\psi^* \cs^\dagger = \cf^\dagger - \frac{1}{2} [\Bf^\dagger,\star \Bf^\dagger],
\end{align*}
and maps the lax BV-BFV data of the second-order theory as
\begin{align*}
    &\psi^* \ts^\bullet  = \tf^\bullet  + (\mathcal{L}_{\Qf} - \d) \betaf^\bullet + \delta \ff^\bullet,&
    &\psi^* \Ls^\bullet  = \Lf^\bullet  + (\mathcal{L}_{\Qf} - \d) \iota_{\Qf} \betaf^\bullet + \d \ff^\bullet,
\end{align*}
where
\begin{align*}
    \betaf^\bullet &= \Tr \left[
    \frac{1}{2} \Bf^\dagger \star \delta \Bf^\dagger 
    + \star \Bf^\dagger \delta \Af
    + \star \Bf^\dagger \delta \cf
    \right],&
    \ff^\bullet &= \Tr \left[
    \frac{1}{2} \Bf^\dagger (\Bf -\star \FAf ) \right]
\end{align*}
in accordance with our notion of lax BV-BFV equivalence. Note that $\ff^1=\ff^2=0$.

Furthermore, we need to show that the respective BV-BFV complexes are quasi-isormophic. The composition map $\lambda^* = \phi^* \circ \psi^*$ is shown to be the identity in Lemma \ref{lem:lambdaYM}, which follows directly from $\phi^* B^\dagger = 0$. In Lemma \ref{lem:chiYM}, we prove that the composition map $\chi^* = \psi^* \circ \phi^*$, which has the explicit form
\begin{align*}
    &\chi^* \Af = \Af,&
    &\chi^* \Af^\dagger = \Af^\dagger - \dAf \star \Bf^\dagger,\\
    &\chi^* \Bf = \star \FAf,&
    &\chi^* \Bf^\dagger = 0,\\
    &\chi^* \cf = \cf,&
    &\chi^* \cf^\dagger = \cs^\dagger
    - \frac{1}{2} [\Bf^\dagger,\star \Bf^\dagger],
\end{align*}
is homotopic to the identity by constructing the morphism $\chi^*_s = e^{s\mathcal{L}_{\Df}}$ with $\Df = [\Rf,\Qf]$, where $\Rf$ is chosen to act as
\begin{align*}
    &\Rf \Af = 0,&
    &\Rf \Af^\dagger = 0,&\\
    &\Rf \Bf = \star \Bf^\dagger,&
    &\Rf \Bf^\dagger = 0,&\\
    &\Rf \cf = 0,&
    &\Rf \cf^\dagger = 0.&
\end{align*}
The homotopy is explicitly given by
\begin{align*}
    &\chi^*_s \Af = \Af,&
    &\chi^*_s \Af^\dagger = \Af^\dagger + (e^{-s}-1) \dAf \star \Bf^\dagger,\\
    &\chi^*_s \Bf = e^{-s} \Bf - (e^{-s}-1) \star \FAf,&
    &\chi^*_s \Bf^\dagger = e^{-s} \Bf^\dagger,\\
    &\chi^*_s \cf = \cf,&
    &\chi^*_s \cf^\dagger = \cf^\dagger + \frac{1}{2} (e^{-2s} - 1) [ \Bf^\dagger, \star \Bf^\dagger].
\end{align*}
and fulfils $\lim_{s\rightarrow \infty} \chi^*_s = \chi^*$. In Lemma \ref{lem:chiIdCohomologyYM} we demonstrate that $\chi^*$ is the identity in cohomology by showing that the map
$$h_\chi \varphi^j_1 = \int_0^\infty e^{s\mathcal{L}_{\Df}} \mathcal{L}_{\Rf} \varphi^j_1 \d s \qquad \quad \varphi^j_1 \in \FfYM,$$ 
satisfying $\chi^* - \mathrm{id}_{1} = (\mathcal{L}_{\Qf} - \d) h_\chi + h_\chi (\mathcal{L}_{\Qf} - \d)$ (cf.\ Lemma \ref{lem:hchiexplicitform}) converges to
\begin{align*}
    &h_\chi \Af =
    h_\chi \Af^\dagger =
    h_\chi \Bf^\dagger =
    h_\chi \cf =
    h_\chi \cf^\dagger = 0,\\
    &h_\chi \Bf
    = \star \Bf^\dagger,
\end{align*}
therefore proving that the two lax BV-BFV theories in question have isomorphic BV-BFV cohomologies
\begin{align*}
    H^\bullet(\bvbfv_{1YM}) \simeq H^\bullet(\bvbfv_{2YM})
\end{align*}
and thus that they are lax BV-BFV equivalent.
\end{proof}

\subsection{1D reparametrisation invariant theories}
\label{sec:1Dreparametrisationinvarianttheories}
In this section we compare two one-dimensional reparamentrisation invariant theories, namely Jacobi theory, which one can think of as classical mechanics at constant energies, and one-dimensional gravity coupled to matter (1D GR). For an in-depth discussion of these theories we refer to \cite{cattaneo2017time}. We recall that the motivation to investigate the equivalence of these two theories is that, even though they are classically equivalent, 1D GR is compatible with the BV-BFV procedure while the Jacobi theory is not and yields a singular boundary structure. Firstly, this raises the question whether this boundary discrepancy is reflected at a cohomological level. Secondly, this discrepancy in the boundary behaviour is also present in the classically equivalent Einstein-Hilbert gravity and Palatini-Cartan gravity in (3+1) dimensions, where the latter is incompatible with the BV-BFV procedure. Our hope is that the comparison and analysis of these toy models might shed light in the question of equivalence of the (3+1) dimensional theories.

We take the base manifold to be a closed interval on the real line $M = I = [a,b] \subset \mathbb{R}$ with coordinate $t$ for both theories, which should be interpreted as a finite time interval.

In the case of Jacobi theory, we consider a matter field $\qj \in \Gamma(\mathbb{R}^n \times I) = C^\infty(I,\mathbb{R}^n)$ with mass $m$. The kinetic energy is taken to be $\Tj(\dqj) = \frac{m}{2}\|\dqj\|^2$ where $\|\cdot\|$ is the Euclidean norm on $\mathbb{R}^n$ and $\dqj = \partial_t \qj$ is the time derivative of $\qj$. Let $E$ denote a parameter and $V(\qj)$ a potential term. We do not assume $E = \Tj(\dqj) + V(\qj)$. The Jacobi action functional takes the form
\begin{align*}
    S_J[\qj] = \int_I 2\sqrt{(E-V)\Tj} \,\d t.
\end{align*}
To see that $S_J$ is parameterisation invariant, note that writing 
\begin{align*}
    \d s^2 = 2m(E-V)\,\d \qj^2
\end{align*}
lets us interpret the Jacobi action functional as the length of a path in the target space $\mathbb{R}^n$ with metric $\d s^2$. As such the symmetry group of Jacobi theory is the diffeomorphism group of the interval $\text{Diff}(I)$, i.e.\ the reparameterisations of $I$. The critical locus of $S_J$ is then given by the geodesics of the metric $\d s^2$, which are the trajectories of classical mechanics with an arbitrary parameterisation \cite{cattaneo2017time}. Imposing $E =\Tj(\dqj) + V(\qj)$ allows us to recover the standard parameterisation. We set $V(\qj)=0$ for the rest of the discussion. The EL equations can be shown to have the form
\begin{align*}
    \partial_t\left(\sqrt{\frac{E}{\Tj}}m\dqj\right) = 0,
\end{align*}
which are singular for $\dqj = 0$. As such, the space of fields for Jacobi theory is not $C^\infty(I,\mathbb{R}^n)$ but rather
\begin{align*}
    {\mathcal{F}_J} = \left\{\qj\in C^\infty(I,\mathbb{R}^n)\,\,|\,\, \dqs(t)\not=0\ \forall t \in I\right\}.
\end{align*}
We can then interpret Jacobi theory as classical mechanics at constant energies where the solutions do not have turning points, i.e.\ points in which the first derivative vanishes.

In the case of 1D GR \cite{cattaneo2017time} we also consider a metric field $\mathsf{g} \in \Gamma(S^2_+T^*I)$ as a non-vanishing section of the bundle of symmetric non-degenerate rank-$(0,2)$ tensors over $I$. For simplicity, we write $\mathsf{g} = g \, \d t^2$ and work with the component $g\in C^\infty(I,\mathbb{R}_{>0})$. The space of fields is given by
\begin{align*}
    {\mathcal{F}_{GR}}
    ={\mathcal{F}_J} \oplus C^\infty(I,\mathbb{R}_{>0}).
\end{align*}
The condition $\dot q \neq 0$ in ${\mathcal{F}_J}$ is strictly speaking not necessary in the 1D GR case, but we are ultimately interested in comparing 1D GR with the Jacobi theory and therefore impose it for consistency. In this picture, we can interpret 1D GR as an extension of Jacobi theory. We consider the action functional
\begin{align*}
    S_{GR}[q,g] = \int_I \left( \frac{T}{g} + E\right)\sqrt{g}\,\d t
    = \int_I \left( \frac{T}{\sqrt{g}} + \sqrt{g}E\right)\,\d t.
\end{align*}
Note that the Ricci tensor vanishes in 1D and hence the Einstein-Hilbert term is absent. The first term in $S_{GR}$ is simply the matter Lagrangian for vanishing potential in the presence of a metric field and the second is a cosmological term. As such we interpret the parameter $E$ as a cosmological constant. Since we are integrating over the Riemannian density $\sqrt{g}\,\d t$ of the metric $\d s^2 = g \d t^2$, the symmetry group is again $\text{Diff}(I)$.

\begin{proposition}[\cite{cattaneo2017time}]
Jacobi theory and 1D GR are classically equivalent. 
\end{proposition}

Let us now turn to the lax BV-BFV formulation of Jacobi theory. We first need to introduce the ghost field, which in case of diffeomorphims invariance is chosen to be $\xij \partial_t \in \mathfrak{X}(I)[1]$ \cite{piguet2000ghost}. In this setting, the Chevalley-Eilenberg operator acts on the fields as
the Lie derivative $\gamma_J = \mathcal{L}_{\xij \partial_t}$ and on the ghost as the Lie bracket of vector fields (cf.\ Example \ref{ex:BRSTcase}). We work with the component $\xij \in C^\infty(I,\mathbb{R})[1]$ for simplicity.

\begin{proposition/definition}
\label{propdef:laxJacobi}
The data $$\mathfrak{F}^\mathrm{lax}_{J} = (\Fj,\tj^\bullet,\Lj^\bullet,\Qj)$$
where
\begin{align}
    \label{eq:laxspaceoffieldsJacobi}
    \Fj = T^*[-1]({\mathcal{F}_J} \oplus C^\infty(I,\mathbb{R})[1]).
\end{align}
together with $\tj^\bullet \in \Omega^{1,\bullet}_\mathrm{loc}(\Fj{\times I})$ and $\Lj^\bullet \in \Omega^{0,\bullet}_\mathrm{loc}(\Fj{\times I})$, which are given by 
\begin{align*}
    &\tj^\bullet = \left[\qj^+ \cdot \delta \qj + \xij^+ \delta \xij\right] \d t + \sqrt{\frac{E}{\Tj}} m \dot{ \qj} \cdot \delta \qj + \qj^+ \xij \cdot \delta \qj - \xij^+ \xij \delta \xij,\\
    &\Lj^\bullet = \left[2\sqrt{E \Tj} + \qj^+ \cdot \xij \dqj + \xij^+ \xij \dot{\xij}\right] \d t.
\end{align*}
and the cohomological vector field $\Qj \in \mathfrak{X}_{\text{evo}}(\Fj)$
\begin{align*}
    &\Qj \qj = \xij \dqj,&
    &\Qj \qj^+ = -\partial_t\left(\sqrt{\frac{E}{\Tj}} m \dqj + \qj^+ \xij\right),\\
    &\Qj \xij = \xij \dxij,&
    &\Qj \xij^+ = - \qj^+ \cdot \dqj + \xij \dxij^+ + 2 \dxij \xij^+.
\end{align*}
defines a lax BV-BFV theory.
\end{proposition/definition}

\begin{proof}
It is a matter of a straightforward calculation to check that the formulas above satisfy the axioms of Definition \ref{def:LaxBVBFV}. 
\end{proof}

\begin{remark}
\label{rem:Chevalley-EilenbergJacobi}
Note that we can explicitly decompose the cohomological vector field $\Qj $ into its Chevalley-Eilenberg and Koszul-Tate parts as $\Qj  = \gamma_J + \delta_J$ by using Equation (\ref{eq:KoszulTate}) and setting $\gamma_J = \Qj  - \delta_J$ on $\{\qj^+,\xij^+\}$. We have:
\begin{align*}
    &\gamma_J \qj^+ = \xij \dqj^+ + \dot{\xij} q^+,&
    &\delta_J \qj^+ = -\partial_t\left(\sqrt{\frac{E}{\Tj}} m \dqj \xij\right),\\
    &\gamma_J \xij^+ = \xij \dot{\xij}^+ + 2 \dot{\xij} \xij^+,&
    &\delta_J \xij^+ = - \qj^+ \cdot \dqj.
\end{align*}
As $\qj$ is a function and $\xij$ is the component of a vector field, defining $\qj^+$ and $\xij^+$ through Equation \eqref{eq:laxspaceoffieldsJacobi} lets us interpret them as components of tensor fields in $\Omega^{\text{top}}(I) \otimes C^\infty(I,\mathbb{R}^n)[-1]$ and $\Omega^{\text{top}}(I) \otimes \Omega^{\text{top}}(I)[-2]$ respectively, or rather as components of a rank-$(0,1)$ and a rank-$(0,2)$ tensors over $I$. As such we see that the Chevalley-Eilenberg differential also acts as $\gamma_J = \mathcal{L}_{\xij \partial_t}$ on $\{\qj^+,\xij^+\}$.
\end{remark}

In the case of 1D GR we have
\begin{proposition}
The data $$\mathfrak{F}^\mathrm{lax}_{GR} = (\Fgr,\tgr^\bullet,\Lgr^\bullet,\Qgr)$$ 
where
\begin{align*}
    \Fgr = T^*[1]({\mathcal{F}_J} 
    \oplus C^\infty(\mathbb{R}_{>0}) 
    \oplus C^\infty(I)[1]).
\end{align*}
together with $\tgr^\bullet \in \Omega^{1,\bullet}_\mathrm{loc}(\Fgr{\times I})$ and $\Lgr^\bullet \in \Omega^{0,\bullet}_\mathrm{loc}(\Fgr{\times I})$, which are given by 
\begin{align*}
    & \tgr^\bullet = \left[q^+ \cdot \delta q +  \xi^+ \delta \xi + g^+ \delta g\right] \d t
    + \frac{m\dot q}{\sqrt{g}} \cdot \delta q + q^+\xi \cdot \delta q+ g^+\xi \delta g - (2g^+ g + \xi^+ \xi) \delta \xi,\\
    & \Lgr^\bullet = \left[\frac{T}{\sqrt{g}} + \sqrt{g} E + q^+ \cdot \xi \dot{q} + g^+ (\xi \dot g + 2 g \dot \xi) + \xi^+ \xi \dot{\xi}\right] \d t
    + \left(\frac{T}{g} - E \right) \sqrt{g}\xi.
\end{align*}
and the cohomological vector field $\Qgr \in \mathfrak{X}_{\text{evo}}(\Fgr)$
\begin{align*}
    &\Qgr q = \xi \dot q,&
    &\Qgr q^+ = -\partial_t\left(\frac{m\dot q}{\sqrt{g}} + q^+\xi \right),\\
    &\Qgr g = \xi \dot g + 2 \dot \xi g,&
    &\Qgr g^+ = \text{EL}_g + \xi \dot g^+ - \dot \xi g^+,\\
    &\Qgr \xi = \xi \dot q,&
    &\Qgr \xi^+ = - q^+ \cdot \dot q + g^+ \dot g + 2 \dot g^+ g + \xi \dot \xi^+ + 2 \dot \xi \xi^+.
\end{align*}
with
\begin{align*}
    \text{EL}_g \coloneqq \frac{\delta S_{GR}}{\delta g} = \frac{E}{2\sqrt{g}} - \frac{T}{2g^{3/2}},
\end{align*}
defines a lax BV-BFV theory.
\end{proposition}

\begin{proof}
It is straightforward to show that these formulas satisfy the axioms of Definition \ref{def:LaxBVBFV}.
\end{proof}

\begin{remark}
\label{rem:Chevalley-Eilenberg1DGR}
As in Jacobi theory we can decompose the cohomological vector field as $\Qgr  = \gamma_{GR} + \delta_{GR}$. We have
\begin{align*}
    &\gamma_{GR} q^+ = \xi \dot q^+ + \dot \xi q^+,&
    &\delta_{GR} q^+ = -\partial_t\left(\frac{m \dot q}{\sqrt{g}}\right),\\
    &\gamma_{GR} g^+ = \xi \dot g^+ - \dot \xi g^+&
    &\delta_{GR} g^+ = \text{EL}_g,\\
    &\gamma_{GR} \xi^+ = \xi \dot \xi^+ + 2 \dot \xi \xi^+,&
    &\delta_{GR} \xi^+ = - q^+ \cdot \dot q + g^+ \dot g + 2 \dot g^+ g.
\end{align*}
Similarly to the Jacobi case, $q^+$, $g^+$ and $\xi^+$ are components of tensors in $\Omega^{\text{top}}(I) \otimes C^\infty(I,\mathbb{R}^n)[-1]$, $\Gamma[-1](S^2_+TI) \otimes \Omega^{\text{top}}(I)$ and $\Omega^{\text{top}}(I) \otimes \Omega^{\text{top}}(I)[-2]$ respectively, or rather components of a rank-$(0,1)$, a rank-$(2,1)$ and a rank-$(0,2)$ tensors over $I$. Therefore we again have $\gamma_{GR} = \mathcal{L}_{\xi \partial_t}$ on $\{q^+, g^+, \xi^+\}$.
\end{remark}

Before presenting the main theorem of this section we need to introduce some useful notation. Let $v \in C^\infty(I,\mathbb{R}^n)$ be a $\mathbb{R}^n$-valued field and let $u = \dot q / \|\dot q\|$ denote the normalised velocity of $q$. Note that $u$ is always well-defined because we assume $\dot q \neq 0$. We can then decompose $v = v_\parallel + v_\perp$ into its parallel $v_\parallel$ and perpendicular $v_\perp$ components with respect to $u$ as 
\begin{align}
\label{eq:ParallelPerp}
\begin{split}
    v_\parallel 
    &= (u \cdot v) u
    = (\dot q \cdot v) \frac{\dot q}{\|\dot q\|^2}
    = (\dot q \cdot v) \frac{m\dot q}{2T},\\
    v_\perp &= v - (u \cdot v) u = v - (\dot q \cdot v) \frac{m\dot q}{2T},
\end{split}
\end{align}
where we used that $T = \frac{m}{2}\|\dot q\|^2$.

We now present the main theorem of this section, together with an outline of the proof. The computational details and the various required Lemmata are presented in the Appendix \ref{app:LengthyCalcJac1DGR}.

\begin{theorem}
\label{theorem:laxequivalenceJacobi1DGR}
The lax BV-BFV theories $\mathfrak{F}^\mathrm{lax}_{GR}$ and $\mathfrak{F}^\mathrm{lax}_{J}$ of 1D GR and Jacobi theory are lax BV-BFV equivalent.
\end{theorem}

\begin{proof}
We need to check all the conditions from Definition \ref{def:laxequivalence}. The existence of two maps $\phi$, $\psi$ with the desired properties is presented in Lemmata \ref{lem:philaxJacobi1DGR} and \ref{lem:psilaxJacobi1DGR} respectively, where we also show that the pullback maps $\phi^*$, $\psi^*$ are chain maps w.r.t. the BV-BFV complexes $\bvbfv^\bullet_i$, and that they map $(\theta^\bullet_i,L^\bullet_i)$ in the desired way. Specifically, $\phi: \Fj \rightarrow \Fgr$ is defined through
\begin{align*}
    &\phi^*q = \qj,&
    &\phi^*g = \frac{\Tj}{E},&
    &\phi^*\xi = \xij,\\
    &\phi^*q^+ = \qj^+,&
    &\phi^*g^+ = 0,&
    &\phi^*\xi^+ = \xij^+.
\end{align*}
and maps the lax BV-BFV data of 1D GR as
\begin{align*}
    &\phi^* \tgr^\bullet = \tj^\bullet,&
    &\phi^* \Lgr^\bullet = \Lj^\bullet.
\end{align*}
On the other hand $\psi: \Fgr \rightarrow \Fj$ is given by
\begin{align*}
    &\psi^*\qj = q,\\
    &\psi^*\xij = \xi,\\
    &\psi^* \qj^+_\parallel = \eta^{3/2} \left(q^+_\parallel 
        - \left[g^+\dot g + 2\dot g^+ g \right] \frac{m\dot q}{2T}
        -\frac{g^{3/2}}{E} \left[\dot{\text{EL}_g}g^+ - \text{EL}_g \dot g^+ \right]\frac{m\dot q}{2T}\right),\\
    &\psi^*\qj^+_\perp = \eta^{3/2} \left(q^+_\perp 
    +\frac{2m}{E} g^+ \ddot q_\perp \right),\\
    &\psi^*\xij^+ =  \eta^{3/2}\left(\xi^+ + \frac{g^{3/2}}{E} \dot g^+ g^+\right),
\end{align*}
where $\eta \coloneqq \frac{gE}{T}$, and maps the lax BV-BFV data of the Jacobi theory as
\begin{align*}
    &\psi^* \tj^\bullet  = \tgr^\bullet  + (\mathcal{L}_{\Qgr} - \d) \betagr^\bullet + \delta \fgr^\bullet,\\
    &\psi^* \Lj^\bullet  = \Lgr^\bullet  + (\mathcal{L}_{\Qgr} - \d) \iota_{\Qgr} \betagr^\bullet + \d \fgr^\bullet,
\end{align*}
where (writing $\kappa^\bullet = \kappa^0 \dt + \kappa^1$)
\begin{align*}
    \betagr^0 =& 
    -\frac{4g^{7/2}}{\Omega^2}Tg^+\delta g^+
    +\left(\frac{2g^2}{\Omega} + \eta^{3/2}\frac{2\sqrt{g}}{E}\right)g^+q^+_\perp \cdot \delta q\\
    &+\left(\frac{4g^{7/2}}{\Omega^2}T - \eta^{3/2}\frac{g^{3/2}}{E}\right)\dot g^+ g^+ \frac{m\dot q}{2T} \cdot \delta q 
    -(\eta^{3/2}-1)\xi^+ \frac{m\dot q}{2T} \cdot \delta q,\\
    \betagr^1 =& \,\,\xi \betagr^0 + \frac{2g^{3/2}}{\Omega}g^+m\dot q \cdot \delta q,\\
    \fgr^0 =& \,\,2 g^+\left(g - \frac{2g^{3/2}}{\Omega} T\right),\\
    \fgr^1 =& \,\,\xi \fgr^0,
\end{align*}
with $\Omega = \sqrt{g}T + g \sqrt{TE}$, in accordance with our notion of lax BV-BFV equivalence. 

Furthermore, we need to show that the respective BV-BFV complexes are quasi-isormophic. The composition map $\lambda^* = \phi^* \circ \psi^*$ is shown to be the identity in Lemma \ref{lem:lambdaJacobi1DGR}. In Lemma \ref{lem:chiJacobi1DGR}, we prove that the composition map $\chi^* = \psi^* \circ \phi^*$, which has explicitly form
\begin{align*}
    &\chi^*q = q  ,\\
    &\chi^*\xi = \xi,\\
    &\chi^*g = \frac{T}{E},\\
    &\chi^* q^+_\parallel = \eta^{3/2} \left(q^+_\parallel 
        - \left[g^+\dot g + 2\dot g^+ g \right] \frac{m\dot q}{2T}
        -\frac{g^{3/2}}{E} \left[\dot{\text{EL}_g}g^+ - \text{EL}_g \dot g^+ \right]\frac{m\dot q}{2T}\right), \\
    &\chi^* q^+_\perp = \eta^{3/2} \left(q^+_\perp 
    +\frac{2m}{E} g^+ \ddot q_\perp \right),\\
    &\chi^* \xi^+ = \eta^{3/2}\left(\xi^+ + \frac{g^{3/2}}{E} \dot g^+ g^+\right),\\
    &\chi^* g^+ = 0,
\end{align*}
is homotopic to the identity by constructing the morphism $\chi^*_s = e^{s\mathcal{L}_{\Dgr}}$ with $\Dgr = [\Rgr,\Qgr]$, where $\Rgr$ is chosen to act as
\begin{align*}
    &\Rgr q = 0,&
    &\Rgr\xi = 0,&
    &\Rgr g= \frac{-2g^{3/2}}{E} g^+,&\nonumber\\
    &\Rgr q^+_\parallel = -\frac{3\sqrt{g}}{E} \text{EL}_g \xi^+ \frac{m\dot q}{2T},&
    &\Rgr \xi^+ =0,&
    &\Rgr g^+ = 0,\nonumber\\
    &\Rgr q^+_\perp = \frac{3\sqrt{g}}{E} g^+ \dot q^+_\perp.
\end{align*}
The homotopy is given by
\begin{align*}
    &\chi^*_sq = q,\\
    &\chi^*_s\xi = \xi,\\
    &\chi^*_s g 
    = e^{-s}g + (1-e^{-s}) \frac{T}{E},\\
    &\chi^*_s q^+_\parallel
    = \left(\frac{g}{\chi^*_s g}\right)^{3/2}
    \left(q^+_\parallel
    + (e^{-s} - 1) 2g^{-3/2}  \sigma\left(\frac{T}{E}\right)\frac{m\dot q}{2T}
    + (e^{-2s} - 1) \frac{3}{E} g^{-3/2} \sigma(g^{3/2}EL_g)\frac{m\dot q}{2T},\right)\\
    &\chi^*_s q^+_\perp = 
    \left(\frac{g}{\chi^*_s g}\right)^{3/2}
    \left( q^+_\perp - (e^{-s}-1)  \frac{2m}{E}\ddot q_\perp  g^+\right),\\
    &\chi^*_s \xi^+ = \left(\frac{g}{\chi^*_s g}\right)^{3/2}
    \left( \xi^+ -(e^{-2s} - 1)\frac{g^{3/2}}{E}\dot g^+ g^+\right),\\
    &\chi^*_s g^+ = \left(\frac{g}{\chi^*_s g}\right)^{3/2}e^{-s} g^+.
\end{align*}
where $\sigma(\varphi) = \varphi \partial_t{(g^{3/2}g^+)} - \dot \varphi g^{3/2}g^+$, and fulfils $\lim_{s\rightarrow \infty} \chi^*_s = \chi^*$. There are some steps that are important to highlight in this case. First, in Lemma \ref{lem:[R,gamma]} we show that $\Rgr$ commutes with the Chevalley-Eilenberg differential $\gamma_{GR}$\footnote{Note that this is also the case in the Yang--Mills example.}
\begin{align*}
    [\Rgr,\gamma_{GR}] = 0,
\end{align*}
by using general arguments and $\gamma_{GR} \sim \mathcal{L}_{\xi \partial_t}$, but it can also be checked through straightforward calculations. We do this as it greatly simplifies the computations since $\Dgr = [\Rgr,\Qgr] = [\Rgr,\delta_{GR}]$. 

Furthermore, while the computations for the action of $\chi^*_s = e^{s\mathcal{L}_\Dgr}$ on the fields $\varphi \in\{q,\xi,g\}$ is analogous to the ones presented for the other examples, it turns out that in this case of the antifields $\varphi^+ \in \{q^+,\xi^+,g^+\}$ finding a recursive formula for $\Dgrk\varphi^+$ is quite challenging. Instead, it is easier to take a slight detour: we first compute $\chi^*_s(g^{3/2}\varphi^+)$ through $\Dgrk(g^{3/2}\varphi^+)$ and then use the property that $\chi
^*_s = e^{s\mathcal{L}_{\Dgr}}$ is a morphism in order to recover $\chi^*_s\varphi^+$
\begin{align}
    &\chi^*_s(g^{3/2}\varphi^+) = \left(\chi^*_sg\right)^{3/2}\chi^*_s\varphi^+ \nonumber\\
    \Leftrightarrow \quad&\chi^*_s\varphi^+ = \frac{\chi^*_s(g^{3/2}\varphi^+)}{(\chi^*_sg)^{3/2}} = \frac{\chi^*_s(g^{3/2}\varphi^+)}{\left[g + (e^{-s} - 1) \left(g-\frac{T}{E}\right)\right]^{3/2}}.
    \label{eq:chi^*_svarphi^+}
\end{align}
The limit $s\rightarrow \infty$ then reads
\begin{align}
    \label{eq:limchi^*_sPhi^+}
    \lim_{s\rightarrow\infty} \chi^*_s \varphi^+
    = \left(\frac{E}{T}\right)^{3/2}
    \lim_{s\rightarrow\infty} \chi^*_s (g^{3/2}\varphi^+).
\end{align}
Since $\chi^*_sg$ is nowhere vanishing for any $s\in \mathbb{R}_{\geq 0}$, this expression is well-defined $\forall s\in\mathbb{R}_{\geq 0}$ iff $\chi^*_s(g^{3/2}\varphi^+)$ is well-defined $\forall s\in\mathbb{R}_{\geq 0}$ as well.

We exemplify this procedure with the computation of $\chi^*_s g^+$. In order to see where the aforementioned problem arises we compute
\begin{align*}
    \Dgr  g^+ 
    &= \left(\delta_{GR} \Rgr + \Rgr \delta_{GR}\right) g^+
    = \Rgr (\text{EL}_g)
    = \frac{\delta \text{EL}_g}{\delta g} \Rgr g\\
    &= \left(-\frac{E}{4g^{3/2}} + \frac{3T}{4g^{5/2}}\right) \frac{-2g^{3/2}}{E} g^+
    = \left(\frac{1}{2} - \frac{3}{2}\frac{T}{Eg}\right)g^+.
\end{align*}
One can then proceed with the calculation of $\Dgrk g^+$ for higher $k$'s and notice that the expressions become quite lengthy as $\Dgr g = T/E-g$ (Equation \eqref{eq:Dgrg}). The idea to avoid this complication by considering $\Dgrk (g^{3/2} g^+)$, where a recursive formula becomes apparent. We have
\begin{align*}
    \Dgr (g^{3/2} g^+) 
    &= \frac{3}{2} g^{1/2} \Dgr g g^+ + g^{3/2} \Dgr g^+ \\
    &= \frac{3}{2} g^{1/2}\left(\frac{T}{E}-g\right) g^+ 
    + g^{3/2} \left(\frac{1}{2} - \frac{3}{2}\frac{T}{Eg}\right)g^+ \\
    &= - g^{3/2} g^+.
\end{align*}
It is then straightforward to see that
\begin{align*}
    && \Dgrk (g^{3/2} g^+) &= (-1)^k g^{3/2}g^+  \quad \text{for } k \geq 0,&&\\
    &\Rightarrow& 
    e^{s\mathcal{L}_{\Dgr}}(g^{3/2}g^+)
    &= e^{-s} g^{3/2}g^+.&&
\end{align*}
Using Eq. \eqref{eq:chi^*_svarphi^+} we then have
\begin{align*}
    e^{s\mathcal{L}_{\Dgr}} g^+ 
    &= \left(\frac{g}{\chi^*_s g}\right)^{3/2} e^{-s} g^+,\\
    \Rightarrow 
    \lim_{s\rightarrow \infty} e^{s\mathcal{L}_{\Dgr}} g^+ &= 0 = \chi^* g^+.
\end{align*}
The computations for $q^+,\xi^+$ are lengthier and can be found in the Appendix \ref{app:LemmataforTheoremJacobi1DGR}. 

In Lemma \ref{lem:chiIdCohomologyJacobi1DGR} we demonstrate that $\chi^*$ is the identity in cohomology by showing that the map
$$h_\chi \varphi^j = \int_0^\infty e^{s\mathcal{L}_{\Dgr}} \mathcal{L}_{\Rgr} \varphi^j \d s, \quad \quad \varphi^j \in \Fgr,$$ 
satisfying $\chi^* - \mathrm{id}_{GR} = (\mathcal{L}_{\Qgr} - \d) h_\chi + h_\chi (\mathcal{L}_{\Qgr} - \d)$ (cf.\ Lemma \ref{lem:hchiexplicitform}) converges to 
\begin{align*}
    &h_\chi q = 
    h_\chi \xi = 
    h_\chi \xi^+ = 
    h_\chi g^+ = 0,\\
    &h_\chi g = -\frac{2g^{3/2}}{E} g^+,\\
    &h_\chi q^+_\parallel =  (1-\eta^{3/2}) \left[\xi^+ + \frac{g^{3/2}}{E}\dot g^+ g^+\right] \frac{m\dot q}{2T}
    +  \frac{3 \eta - 2 \sqrt{\eta} - 1}{(\sqrt{\eta}+1)^2} \frac{g^{3/2}}{E}\dot g^+ g^+ \frac{m\dot q}{2T}, \\
    &h_\chi q^+_\perp = \frac{2}{T} \left(\frac{\eta}{\sqrt{\eta}+1} + 1\right) g^{3/2} g^+q^+_\perp.
\end{align*}
therefore proving that the two lax BV-BFV theories in question have isomorphic BV-BFV cohomologies
\begin{align*}
    H^\bullet(\bvbfv_{J}) \simeq H^\bullet(\bvbfv_{GR})
\end{align*}
and thus that they are lax BV-BFV equivalent.
\end{proof}

In this example, we are also interested on how the the composition maps $\lambda^*$, $\chi^*$ affect the boundary structure, namely the strict BV-BFV structure of the Jacobi theory and 1D GR. More specifically, we want to investigate how they change the kernel of the pre-boundary forms $\check \omega$ and as such the quotient $\F^\partial = \check \F^\partial/\ker \check \omega$ (cf.\ Equation \eqref{eq:quotientboundary}). 

In the case of $\lambda^*$ this is trivial since it is the identity. Regarding $\chi^*$, we argue that $\ker \chi^* \check \omega_{GR}$ has a singular behaviour and that we cannot construct a BV-BFV theory from the data $\chi^*\mathfrak{F}^\mathrm{lax}_{GR}\coloneqq (\F^\mathrm{lax}_{GR},\chi^*\theta^\bullet_{GR},\chi^*L^\bullet_{GR},\Qgr )$. Thus, although $\chi^*$ is the identity in the BV-BFV cohomology $H^\bullet(\bvbfv_{GR})$, it spoils the BV-BFV structure of 1D GR.

\begin{theorem}
\label{theorem:chispoilsboundary}
The lax BV-BFV theory $\chi^* \mathfrak{F}^\mathrm{lax}_{GR} \coloneqq (\F^\mathrm{lax}_{GR},\chi^*\theta^\bullet_{GR},\chi^*L^\bullet_{GR},\Qgr)$ does not yield a BV-BFV theory.
\end{theorem}

\begin{proof}
Recall that pulling back $(\theta^\bullet_{GR},L^\bullet_{GR})$ with $\phi^*$ gives $\phi^*\theta^\bullet_{GR} = \theta^\bullet_{J}$ and $\phi^*L^\bullet_{GR} = L^\bullet_{J}$. Applying the map $\chi^* = \psi^*\circ\phi^*$ to $(\theta^\bullet_{GR},L^\bullet_{GR})$ then yields
\begin{align*}
    &\chi^* \theta^\bullet_{GR} 
    = (\psi^* \circ \phi^*) \theta^\bullet_{GR} 
    = \psi^* \theta^\bullet_{J} = \theta^\bullet_{J}[\psi^*\qj,\psi^*\qj^+,\psi^*\xij,\psi^*\xij^+],\\
    &\chi^* L^\bullet_{GR} 
    = (\psi^* \circ \phi^*) L^\bullet_{GR} 
    = \psi^* L^\bullet_{J} = L^\bullet_{J}[\psi^*\qj,\psi^*\qj^+,\psi^*\xij,\psi^*\xij^+].
\end{align*}
Thus, the lax BV-BFV data of $\chi^* \mathfrak{F}^\mathrm{lax}_{GR}$ has the same form as the lax BV-BFV data for Jacobi theory presented in the Proposition/Definition \ref{propdef:laxJacobi} on the submanifold of $\F^\mathrm{lax}_{GR}$ with local coordinates $\{\psi^*\qj,\psi^*\qj^+,\psi^*\xij,\psi^*\xij^+\}$. Furthermore, applying $\chi^*$ to Equations (\ref{eq:LaxStructure}) for the lax BV-BFV formulation of 1D GR yields
\begin{align*}
    &\iota_{\Qgr } \psi^* \varpi^\bullet_J 
    = \delta \psi^*L^\bullet_J + \d \psi^*\theta^\bullet_J,\\
    &\iota_{\Qgr } \iota_{\Qgr } \psi^*\varpi^\bullet_J = 2\,\d \psi^* L^\bullet_J.
\end{align*}
This means that the theory $\chi^* \mathfrak{F}^\mathrm{lax}_{GR}$ is just a version of Jacobi theory which is defined on a submanifold of $\mathcal{F}^\mathrm{lax}_{GR}$ with local coordinates $\{\psi^*\qj,\psi^*\qj^+,\psi^*\xij,\psi^*\xij^+\}$. This theory will have the same behaviour as the original Jacobi theory and as such the kernel of the pre-boundary 2-form
\begin{align*}
    \chi^* \check \omega_{GR}
    &= \int_{\partial I} \delta \chi^* \theta^1_{GR} \d t
    = \int_{\partial I} \delta \psi^* \theta^1_{J} \d t
    = \psi^* \check \omega_{J}\\
    &= \check \omega_{J}[\psi^*\qj,\psi^*\qj^+,\psi^*\xij,\psi^*\xij^+],
\end{align*}
will be singular, just as the kernel of the pre-boundary 2-form $\check \omega_{J}$ of Jacobi theory \cite{cattaneo2017time}. As such, the data $\chi^* \mathfrak{F}^\mathrm{lax}_{GR}$ does not yield a BV-BFV theory.
\end{proof}

\begin{remark}
\label{rem:Quasi-isoDoesNotPreserveBFV}
It is clear that the theories $\mathfrak{F}^\mathrm{lax}_{GR}$ and $\chi^*\mathfrak{F}^\mathrm{lax}_{GR}$ are also lax BV-BFV equivalent (see Remark \ref{rem:Pairwiselaxequivalent}). We have thus presented two pairs of theories, $(\mathfrak{F}^\mathrm{lax}_{J}, \mathfrak{F}^\mathrm{lax}_{GR})$ and $ (\mathfrak{F}^\mathrm{lax}_{GR}, \chi^*\mathfrak{F}^\mathrm{lax}_{GR})$, which have isomorphic BV-BFV cohomologies, but differ in terms of their compatibility with the BV-BFV axioms, a behaviour which is not present in the examples of classical mechanics on a curved background and Yang--Mills theory which we considered in Sections \ref{sec:CM} and \ref{sec:YM}. Indeed, a remarkable feature of the classical equivalence between Jacobi theory and 1d GR is that it can actually be promoted to a quasi-isomorphism of their BV-BFV complexes (lax equivalence), which in particular implies BV equivalence in the sense of Definition \ref{def:BVequivalence}. This is compatible with the process of removal of auxiliary fields outlined in \cite{barnich1995local}. However, the request that two lax-equivalent theories both admit a strictification in the sense of Remark \ref{rem:strictification} is a genuine refinement of the notion of BV (and lax) equivalence of field theories. Since the BV-BFV quantization program requires a strict theory, this obstruction marks a roadblock for non-strictifiable lax BV-BFV theories.
\end{remark}

\appendix
\section{Lengthy calculation for Yang--Mills}
\label{sec:LengthyCalcYM}

\subsection{Lemmata used in Theorem \ref{theorem:laxequivalenceYM}}

This Appendix we present the lemmata used in Theorem \ref{theorem:laxequivalenceYM} and the respective detailed proofs and calculations.

\begin{lemma}
\label{lem:philaxYM}
Let $\phi: \FsYM \rightarrow \FfYM$ be defined through
\begin{align*}
    &\phi^* \Af = \As,&
    &\phi^* \Bf = \star \FAs,&
    &\phi^* \cf = \cs,\\
    &\phi^* \Af^\dagger = \As^\dagger,&
    &\phi^* \Bf^\dagger = 0,&
    &\phi^* \cf^\dagger = \cs^\dagger.
\end{align*}
Its pullback map $\phi^*$ is a chain map w.r.t. $(\mathcal{L}_{Q_i}-\d)$ and maps the lax BV-BFV data of the first-order theory as
\begin{align*}
    &\phi^* \tf^\bullet = \ts^\bullet,&
    &\phi^* \Lf^\bullet = \Ls^\bullet.
\end{align*}
\end{lemma}

\begin{proof}
The computations are in the same line as the ones presented in the example of Classical Mechanics on a curved background (cf.\ Lemma \ref{lem:philaxCM}). One should keep in mind that $\Qs \star \FAs= [\cs,\star \FAs]$.
\end{proof}

\begin{lemma}
\label{lem:psilaxYM}
Let $\psi: \FfYM \rightarrow \FsYM$ be the map defined through
\begin{align*}
    &\psi^* \As = \Af,&
    &\psi^* \cs = \cf,\\
    &\psi^* \As^\dagger = \Af^\dagger -  \dAf \star \Bf^\dagger,&
    &\psi^* \cs^\dagger = \cf^\dagger - \frac{1}{2} [\Bf^\dagger,\star \Bf^\dagger].
\end{align*}
Its pullback map $\psi^*$ is a chain map w.r.t. $(\mathcal{L}_{Q_i}-\d)$ and maps the lax BV-BFV data of the second-order theory as
\begin{align*}
    &\psi^* \ts^\bullet  = \tf^\bullet  + (\mathcal{L}_{\Qf} - \d) \betaf^\bullet + \delta \ff^\bullet,&
    &\psi^* \Ls^\bullet  = \Lf^\bullet  + (\mathcal{L}_{\Qf} - \d) \iota_{\Qf} \betaf^\bullet + \d \ff^\bullet,
\end{align*}
where
\begin{align*}
    \betaf^\bullet &= \Tr \left[
    \frac{1}{2} \Bf^\dagger \star \delta \Bf^\dagger 
    + \star \Bf^\dagger \delta \Af
    + \star \Bf^\dagger \delta \cf
    \right],&
    \ff^\bullet &= \Tr \left[
    \frac{1}{2} \Bf^\dagger (\Bf -\star \FAf ) \right].
\end{align*}
\end{lemma}

\begin{proof}
The chain map conditions in the case of $\As$ and $\cs$ are straightforward to check. In the case of $\As^\dagger$, we first note that
\begin{align*}
    \Qf \dAf \star \Bf^\dagger
    &= -\dAf \Qf \star \Bf^\dagger + [\dAf \cf,\star \Bf^\dagger]
    = -\dAf (\star \FAf - \Bf) - \dAf [\cf,\star \Bf^\dagger]+ [\dAf \cf,\star \Bf^\dagger]\\
    &= -\dAf (\star \FAf - \Bf) + [\cf,\dAf \star \Bf^\dagger],
\end{align*}
as such we have
\begin{align*}
    \Qf \psi^* \As^\dagger 
    &= \Qf(\Af^\dagger -  \dAf \star \Bf^\dagger)
    = \dAf \Bf + [\cf,\Af^\dagger] - \Qf \dAf \star \Bf^\dagger\\
    &= \dAf \star \FAf + [\cf,\Af^\dagger] - [\cf,\dAf \star \Bf^\dagger]
    = \psi^*(\dAs \star \FAs + [\cs,\As^\dagger]) = \psi^*\Qs \As.
\end{align*}
Before addressing the case of $\cf^\dagger$, we compute
\begin{align*}
    \Qf [\Bf^\dagger,\star\Bf^\dagger] 
    = &[\FAf,\star\Bf^\dagger] - \es [\star \Bf,\star\Bf^\dagger] + [[\cf,\Bf^\dagger],\star\Bf^\dagger]\\
    & \quad- [\Bf^\dagger,\FAf] - [\Bf^\dagger, \Bf] + [\Bf^\dagger,[\cf,\star \Bf^\dagger]].
\end{align*}
Using $[\alpha,\star \beta] = - [\beta,\star \alpha]$ for $\alpha,\beta \in \Omega^\bullet(M,\mathfrak{g})$, the graded Jacobi identity for $\cf,\Bf$, $\star \Bf$ and $[\FAf,\star\Bf^\dagger] = \dAf^2 \star\Bf^\dagger$, this yields
\begin{align*}
    \Qf [\Bf^\dagger,\star\Bf^\dagger] 
    = 2 \dAf^2 \star\Bf^\dagger
    + 2 [\Bf^\dagger,\Bf]
    + [\cf,[\Bf^\dagger,\star\Bf^\dagger]].
\end{align*}
With this in hand we have
\begin{align*}
    \Qf \psi^* \cs^\dagger &= \dAf \Af^\dagger + [\cf,\cf^\dagger] + [\Bf^\dagger,\Bf] - \frac{1}{2}\Qf[\Bf^\dagger,\star\Bf^\dagger]\\
    &= \dAf (\Af^\dagger - \dAf \star \Bf^\dagger) + \left[\cf,\cf^\dagger-\frac{1}{2}[\Bf^\dagger,\star\Bf^\dagger]\right]
    = \psi^*(\dAs \As^\dagger + [\cs,\cs^\dagger])
    = \psi^* \Qs \cs^\dagger.
\end{align*}
Let now $\Delta \ts^\bullet \coloneqq \psi^* \ts^\bullet - \tf^\bullet$ and  $\Delta \Ls^\bullet \coloneqq \psi^* \Ls^\bullet - \Lf^\bullet$. We need to check whether
\begin{align*}
    &\Delta \ts^0  = \mathcal{L}_{\Qf} \betaf^0 - \d\betaf^1  + \delta \ff^0,\\
    &\Delta \ts^1 = \mathcal{L}_{\Qf} \betaf^1 - \d\betaf^2  + \delta \ff^2,\\
    &\Delta \ts^2  = \mathcal{L}_{\Qf} \betaf^2  + \delta \ff^2,\\
    &\Delta \Ls^0 =  \mathcal{L}_{\Qf} \iota_{\Qf} \betaf^0 - \d \iota_{\Qf} \betaf^1 + \d \ff^1.
\end{align*}
Note that $\ff^1 = \ff^2 = 0$. Recall that we only need to compute $\Delta \Ls^0$, as $\Delta \Ls^k$ for $k>0$ is determined by $\Delta \ts^\bullet$, as shown in Proposition \ref{prop:LaxRedundancyTransformation}.

\textbf{Computation of $\Delta \ts^0$}:
Explicitly computing $\Delta \ts^0 = \psi^* \ts^0 - \tf^0$ yields
\begin{align*}
    \Delta \ts^0 = \Tr\left[-  \dAf \star \Bf^\dagger \delta A - \frac{1}{2} [\Bf^\dagger,\star \Bf^\dagger] \delta c - \Bf^\dagger \delta \Bf \right].
\end{align*}
Before tackling $\mathcal{L}_{\Qf} \betaf^0 = \Tr[ \frac{1}{2} \mathcal{L}_{\Qf}(\Bf^\dagger \star \delta \Bf^\dagger)]$ we note that
\begin{align*}
    &\Bf^\dagger \delta [\cf, \star \Bf^\dagger]
    = \Bf^\dagger[\delta \cf, \star \Bf^\dagger] - [\cf,\delta \star \Bf^\dagger]
    = -[\Bf^\dagger, \star \Bf^\dagger] \delta \cf - [\cf,\delta \star \Bf^\dagger]\\
    \Rightarrow & \Bf^\dagger \delta [\cf, \star \Bf^\dagger] + [\cf,\delta \star \Bf^\dagger] = -[\Bf^\dagger, \star \Bf^\dagger] \delta \cf,
\end{align*}
where we ignored the term $[\Bf^\dagger\delta \cf, \star \Bf^\dagger]$ since $\Tr\big[[\alpha, \beta \gamma]\big] = 0$ for $\alpha,\beta,\gamma \in \Omega^\bullet(M,\mathfrak{g})$. As such
\begin{align*}
    \mathcal{L}_{\Qf}(\Bf^\dagger \star \delta \Bf^\dagger)
    &= (\FAf - \es \star \Bf + [\cf, \Bf^\dagger]) \star \delta \Bf^\dagger
    + \Bf^\dagger \star \delta (\FAf - \es \star \Bf + [\cf, \Bf^\dagger])\\
    &= \delta \Bf^\dagger (\star \FAf - \Bf) + [\cf, \Bf^\dagger] \star \delta \Bf^\dagger
    + \Bf^\dagger \delta (\star \FAf - \Bf) + \Bf^\dagger \delta [\cf, \Bf^\dagger]\\
    &= \delta \Bf^\dagger (\star \FAf - \Bf)
    + \Bf^\dagger \delta (\star \FAf - \Bf) -[\Bf^\dagger, \star \Bf^\dagger] \delta \cf
\end{align*}
Therefore, keeping in mind that $\delta \FAf = -\dAf \delta \Af$,
\begin{align*}
    &\mathcal{L}_{\Qf} \betaf^0 - \d\betaf^1  + \delta \ff^0
    = \Tr\left[\mathcal{L}_{\Qf} \left(\frac{1}{2} \Bf^\dagger \star \delta \Bf^\dagger \right) - \d(\star \Bf^\dagger \delta \Af)  + \delta \left(\frac{1}{2}\Bf^\dagger (\Bf - \star \FAf)\right)\right]\\
    &= \Tr\left[\frac{1}{2}\delta \Bf^\dagger (\star \FAf - \Bf)
    + \frac{1}{2} \Bf^\dagger \delta (\star \FAf - \Bf) -\frac{1}{2} [\Bf^\dagger, \star \Bf^\dagger] \delta \cf - \dAf + \frac{1}{2} \delta(\Bf^\dagger(\Bf - \FAf))\right]\\
    &= \Tr\left[\Tr[\frac{1}{2} \delta \Bf^\dagger \star \FAf 
    - \frac{1}{2} \delta \Bf^\dagger \Bf 
    + \frac{1}{2} \Bf^\dagger \delta \star \FAf 
    - \frac{1}{2} \Bf^\dagger \delta \Bf 
    - \frac{1}{2} [\Bf^\dagger, \star \Bf^\dagger] \delta \cf\right.\\
    &\quad\left.- \dAf \star \Bf^\dagger \delta A 
    - \Bf^\dagger \star \delta \FAf 
    + \frac{1}{2} \delta \Bf^\dagger \Bf 
    - \frac{1}{2} \delta \Bf^\dagger \star \FAf 
    - \frac{1}{2} \Bf^\dagger  \delta \Bf 
    + \frac{1}{2} \Bf^\dagger \star \delta \FAf\right]\\
    &= \Tr\left[-  \dAf \star \Bf^\dagger \delta A - \frac{1}{2} [\Bf^\dagger,\star \Bf^\dagger] \delta c - \Bf^\dagger \delta \Bf \right]= \Delta \ts^0.
\end{align*}

\textbf{Computation of $\Delta \ts^1$}:
We have $\Delta \ts^1 = \psi^* \ts^1 - \tf^1$ and as such
\begin{align*}
    \Delta \ts^1 = 
    \Tr\left[\delta \Af \star \FAf 
    - \dAf \star \Bf^\dagger \delta \cf 
    - \Bf \delta \Af \right].
\end{align*}
Noting the identities 
\begin{align*}
    \delta \dAf \cf &= - \dAf \delta \cf + [\delta A,\cf],\\
    [\cf,\star \Bf^\dagger \delta \Af] &= - \star \Bf^\dagger [\cf,\delta \Af] 
    + [\cf,\star \Bf^\dagger]  \delta \Af
    = \star \Bf^\dagger [\delta \Af,\cf] 
    + [\cf,\star \Bf^\dagger]  \delta \Af,
\end{align*}
we see that
\begin{align*}
    &\mathcal{L}_\Qf \betaf^1 -\d \betaf^2 + \delta \ff^1
    = \Tr\left[\mathcal{L}_\Qf(\star \Bf^\dagger \delta \Af) - \d(\star \Bf^\dagger \delta \cf)\right]\\
    &= \Tr\left[(\star \FAf - \Bf + [\cf, \star \Bf^\dagger])\delta \Af
    + \star \Bf^\dagger \delta \dAf \cf 
    - \dAf \star \Bf^\dagger \delta \cf 
    + \star \Bf^\dagger \dAf \delta \cf\right]\\
    &= \Tr\left[(\star \FAf - \Bf) \delta \Af  
    - \dAf \star \Bf^\dagger \delta \cf 
    + [\cf,\star \Bf^\dagger] \delta \Af + \star \Bf^\dagger [\delta \Af,\cf]\right]
    = \Delta \ts^1
\end{align*}

\textbf{Computation of $\Delta \ts^2$:}
$\Delta \ts^2 = \psi^* \ts^2 - \tf^2$ takes the form
\begin{align*}
    \Delta \ts^2 = (\star \FAf - \Bf) \delta \cf.
\end{align*}
Furthermore since
$
    [c, \star \Bf^\dagger \delta c]
    =    [c, \star \Bf^\dagger] \delta c -  \star \Bf^\dagger [c,\delta c], 
$
we have
\begin{align*}
    \mathcal{L}_\Qf \betaf^2 + \delta \ff^2 
    &= \mathcal{L}_\Qf(\star \Bf^\dagger \delta c)
    = (\star \FAf - \Bf + [c,\star \Bf^\dagger])\delta c
    + \frac{1}{2} \star \Bf^\dagger \delta [c,c]\\
    &= \Delta \ts^2 
    + [c,\star \Bf^\dagger] \delta c
    - \star \Bf^\dagger [c,\delta c]
    = \Delta \ts^2 
\end{align*}

\textbf{Computation of $\Delta \Ls^0$}:
Explicitly $\Delta \Ls^0 = \psi^* \Ls^0 - \Lf^0$ yields
\begin{align*}
    \Delta \Ls^0 
    = \Tr\left[\frac{1}{2} \FAf \star \FAf 
    - \Bf \FAf 
    + \frac{\es}{2} \Bf \star \Bf
    - \dAf \star \Bf^\dagger \dAf \cf
    -\frac{1}{4} [\Bf^\dagger,\star \Bf^\dagger][\cf,\cf]
    - \Bf^\dagger [\cf,\Bf] \right].
\end{align*}
Before computing $\mathcal{L}_\Qf \iota_{\Qf} \betaf^0 - \d \iota_{\Qf} \betaf^1 + \d \ff^1$ we note that
\begin{align*}
    [\star \Bf^\dagger \FAf, \cf] 
    &= \star \Bf^\dagger  [\FAf, \cf]
    + [\star \Bf^\dagger , \cf]\FAf
    = \star \Bf^\dagger  [\FAf, \cf]
    + \FAf [\cf, \star \Bf^\dagger]\\
    [\cf, \Bf^\dagger \Bf] 
    &= - \Bf^\dagger [\cf, \Bf]
    + [\cf, \Bf^\dagger] \Bf
\end{align*}
and
\begin{align*}
    && &[c,\Bf^\dagger][c,\star \Bf^\dagger] 
    = [\cf,\Bf^\dagger [\cf, \star \Bf^\dagger]]
    + \Bf^\dagger [\cf,[\cf,\star\Bf^\dagger]]\\
    && &= [\cf,\Bf^\dagger [\cf, \star \Bf^\dagger]]
    - \Bf^\dagger[\cf,[\star\Bf^\dagger,\cf]]
    - \Bf^\dagger[\star\Bf^\dagger[\cf,\cf]]\\
    && &= [\cf,\Bf^\dagger [\cf, \star \Bf^\dagger]]
    + [\cf,\Bf^\dagger [\star \Bf^\dagger,\cf]]
    - [c,\Bf^\dagger][c,\star \Bf^\dagger]  + [\star\Bf^\dagger,\Bf^\dagger[\cf,\cf]]
    - [\star \Bf^\dagger,\Bf^\dagger][\cf,\cf]\\
    &\Rightarrow& &\Tr\bigg[[c,\Bf^\dagger][c,\star \Bf^\dagger]\bigg]
    = \Tr\bigg[- \frac{1}{2} [\Bf^\dagger,\star \Bf^\dagger][\cf,\cf] \bigg].
\end{align*}
Since $\ff^1 = 0$ we therefore have
\begin{align*}
    &\mathcal{L}_\Qf \iota_{\Qf} \betaf^0 - \d \iota_{\Qf} \betaf^1 + \d \ff^1
    = \Tr\bigg[\frac{1}{2} \Qf\Bf^\dagger \star \Qf\Bf^\dagger
    - \d (\star \Bf^\dagger \Qf \Af) \bigg]\\
    &= \Tr\bigg[ \frac{1}{2} (\FAf - \es \star \Bf + [c,\Bf^\dagger])
    (\star \FAf - \Bf + [c,\star \Bf^\dagger])
    - \dAf (\star \Bf^\dagger \dAf \cf)\bigg]\\
    &=  \Tr\bigg[\frac{1}{2} \Big(
    \FAf \star \FAf 
    - \FAf \Bf  
    + \FAf [c,\star \Bf^\dagger]
    - \es \star \Bf \star \FAf 
    + \es \star \Bf \Bf
    - \es \star \Bf [c,\star \Bf^\dagger]\\
    &\qquad
    + [c,\Bf^\dagger] \star \FAf 
    - [c,\Bf^\dagger] \Bf
    + [c,\Bf^\dagger] [c,\star \Bf^\dagger]\Big)
    -\dAf \star \Bf^\dagger \dAf \cf 
    + \star \Bf^\dagger \dAf^2 \cf\bigg]\\
    &= \Tr\bigg[
    \frac{1}{2} \FAf \star \FAf - \Bf \FAf + \frac{\es}{2} \Bf \star \Bf - \dAf \star \Bf^\dagger\\
    &\qquad 
    + \FAf[\cf,\star \Bf^\dagger]
    + \star \Bf^\dagger [\FAf,\cf]
    - [\cf , \Bf^\dagger] \Bf
    + \frac{1}{2} [c,\Bf^\dagger][c,\star \Bf^\dagger] \bigg]\\
    &= \Tr\bigg[\frac{1}{2} \FAf \star \FAf - \Bf \FAf + \frac{\es}{2} \Bf \star \Bf - \dAf \star \Bf^\dagger - [\cf,\Bf^\dagger] \Bf - \frac{1}{4} [\Bf^\dagger,\star \Bf^\dagger][\cf,\cf]\bigg]\\
    &= \Delta \Ls^0.
\end{align*}
\end{proof}

\begin{lemma}
\label{lem:lambdaYM}
The composition map $\lambda^* = \phi^* \circ \psi^*: \bvbfv^\bullet_2 \rightarrow \bvbfv^\bullet_2$  is the identity
\begin{align*}
    &\lambda^* \As = \As,&
    &\lambda^* \cs = \cs,\\
    &\lambda^* \As^\dagger = \As^\dagger,&
    &\lambda^* \cs^\dagger = \cs^\dagger,
\end{align*}
and as such the identity in cohomology.
\end{lemma}

\begin{proof}
Keeping in mind that $\phi^* B^\dagger=0$, this is a straightforward calculation.
\end{proof}

\begin{lemma}
\label{lem:chiYM}
The composition map $\chi^* = \psi^* \circ \phi^*: \bvbfv^\bullet_1 \rightarrow \bvbfv^\bullet_1$ acts as
\begin{align*}
    &\chi^* \Af = \Af,&
    &\chi^* \Bf = \star \FAf,&
    &\chi^* \cf = \cf,\\
    &\chi^* \Af^\dagger = \Af^\dagger - \dAf \star \Bf^\dagger,&
    &\chi^* \Bf^\dagger = 0,&
    &\chi^* \cf^\dagger = \cs^\dagger
    - \frac{1}{2} [\Bf^\dagger,\star \Bf^\dagger],
\end{align*}
and is homotopic to the identity.
\end{lemma}

\begin{proof}
The explicit computation for $\chi^*$ is again straightforward. To show that it is indeed homotopic to the identity, we choose the vector field $\Rf\in \mathfrak{X}_{\text{evo}}(\FfYM)[-1]$ to act as
\begin{align*}
    &\Rf \Af = 0,&
    &\Rf \Bf = \star \Bf^\dagger,&
    &\Rf \cf = 0,\\
    &\Rf \Af^\dagger = 0,&
    &\Rf \Bf^\dagger = 0,&
    &\Rf \cf^\dagger = 0.
\end{align*}
We now want to compute $\chi^*_s = e^{s\mathcal{L}_{\Df}}$, with $\Df = [\Rf,\Qf]$, and show that $\lim_{s\rightarrow \infty} \chi^*_s = \chi^*$.\\
\textbf{Computation for $\Af$ and $\cf$:}
\begin{align*}
    && \Df \Af &= [\Qf,\Rf] \Af = \Rf \dAf \cf = 0&&\\ 
    &\Rightarrow& e^{s\mathcal{L}_{\Df}} \Af &= \Af&&\\
    &\Rightarrow& \lim_{s\rightarrow \infty} e^{s\mathcal{L}_{\Df}} \Af &= \Af = \chi^* \Af,&&\\
    && \Df \cf &= [\Qf,\Rf] \cf = \frac{1}{2} \Rf [\cf,\cf]= 0&&\\ 
    &\Rightarrow& e^{s\mathcal{L}_{\Df}} \cf &= \cf&&\\
    &\Rightarrow& \lim_{s\rightarrow \infty} e^{s\mathcal{L}_{\Df}} \cf &= \cf = \chi^* \cf.&&
\end{align*}
\textbf{Computation for $B$:}
\begin{align*}
    \Df \Bf 
    = [\Qf,\Rf] \Bf 
    = \Qf \star \Bf^\dagger + \Rf [\cf,\Bf]
    = \star \FAf - \Bf + [\cf,\star \Bf^\dagger] - [\cf,\star \Bf^\dagger]
    = \star \FAf - \Bf.
\end{align*}
Noting that $\Df \Af = 0$ implies $\Df \FAf = 0$, we have
\begin{align*}
    && \Df^2 \Bf &= - \Df \Bf = - (\star \FAf - \Bf) &&\\
    &\Rightarrow& 
    \Df^k \Bf &=  (-1)^k (\Bf - \star \FAf) \qquad \text{ for } k \geq 1&&\\
    &\Rightarrow& 
    e^{s\mathcal{L}_{\Df}} \Bf &= \Bf 
    + \sum^\infty_{k=1} \frac{s^k}{k!} (-1)^k(\Bf - \star \FAf)&&\\
    && &= \Bf + (e^{-s} - 1) (\Bf - \star \FAf)&&\\
    &\Rightarrow& \lim_{s\rightarrow \infty} e^{s\mathcal{L}_{\Df}} \Bf &= \star \FAf = \chi^* \Bf.&&
\end{align*}
\textbf{Computation for $\Bf^\dagger$:}
\begin{align*}
    && \Df \Bf^\dagger &= \Rf(\FAf - \es \star \Bf + [\cf,\Bf^\dagger])
    = - \Bf^\dagger&&\\
    &\Rightarrow&
    \Df^k \Bf^\dagger &= (-1)^k \Bf^\dagger \qquad \text{ for } k \geq 0&&\\
    &\Rightarrow&
    e^{s\mathcal{L}_{\Df}} \Bf^\dagger &= e^{-s} \Bf^\dagger&&\\
    &\Rightarrow& \lim_{s\rightarrow \infty} e^{s\mathcal{L}_{\Df}} \Bf^\dagger &= 0 = \chi^* \Bf^\dagger.&&
\end{align*}
\textbf{Computation for $\Af^\dagger$:}
\begin{align*}
    && \Df \Af^\dagger &= \Rf (\dAf \Bf + [\cf,\Af^\dagger])
    = - \dAf \star \Bf^\dagger&&\\
    &\Rightarrow& 
    \Df^k \Af^\dagger 
    &= - \dAf \star \Df^{k-1} \Bf^\dagger
    = (-1)^k \dAf \star \Bf^\dagger \qquad \text{ for } k \geq 1&&\\
    &\Rightarrow&
    e^{s\mathcal{L}_{\Df}} \Af^\dagger &= \Af^\dagger + (e^{-s}-1) \dAf \star \Bf^\dagger&&\\
    &\Rightarrow& \lim_{s\rightarrow \infty} e^{s\mathcal{L}_{\Df}} \Af^\dagger &= \Af^\dagger - \dAf \star \Bf^\dagger = \chi^* \Af^\dagger.&&
\end{align*}
\textbf{Computation for $\cf^\dagger$:}
\begin{align*}
    && \Df \cf^\dagger &= \Rf (\dAf \Af^\dagger + [\cf,\cf^\dagger] + [\Bf^\dagger,\Bf])
    = -[\Bf^\dagger, \star \Bf^\dagger]&&\\
    &\Rightarrow&
    \Df^2 \cf^\dagger &= -[\Df \Bf^\dagger, \star \Bf^\dagger] - [\Bf^\dagger, \star \Df \Bf^\dagger]
    = 2 [ \Bf^\dagger, \star \Bf^\dagger] 
    = - 2 \Df \cf^\dagger&&\\
    &\Rightarrow&
    \Df^k \cf^\dagger &= -(-2)^{k-1} [ \Bf^\dagger, \star \Bf^\dagger]
    = \frac{1}{2} (-2)^{k} [ \Bf^\dagger, \star \Bf^\dagger] \qquad \text{ for } k \geq 1 &&\\
    &\Rightarrow&
    e^{s\mathcal{L}_{\Df}} \cf^\dagger &= \cf^\dagger + \frac{1}{2} (e^{-2s} - 1) [ \Bf^\dagger, \star \Bf^\dagger]&&\\
    &\Rightarrow& \lim_{s\rightarrow \infty} e^{s\mathcal{L}_{\Df}} \cf^\dagger &= \cf^\dagger - \frac{1}{2} [ \Bf^\dagger, \star \Bf^\dagger] = \chi^* \cf^\dagger.&&
\end{align*}
Thus we have shown that $\chi^*$ is homotopic to the identity.
\end{proof}

\begin{lemma}
\label{lem:chiIdCohomologyYM}
The map $\chi^*$ is the identity in cohomology.
\end{lemma}

\begin{proof}
We have to show that the map $h_\chi$ converges on $\FfYM$, namely
\begin{align*}
    h_\chi \varphi_i = \int^\infty_0 e^{s \Df} \Rf \varphi_i \, \d s < \infty 
    \quad \forall \varphi_i \in \FfYM.
\end{align*}
As $\Rf \Af = \Rf\cf= \Rf\Af^\dagger= \Rf\Bf^\dagger= \Rf\cf^\dagger= 0$ we have
\begin{align*}
    h_\chi \Af =
    h_\chi \Af^\dagger =
    h_\chi \Bf^\dagger =
    h_\chi \cf =
    h_\chi \cf^\dagger = 0.
\end{align*}
In the case of $\Bf$ we compute
\begin{align*}
    h_\chi \Bf = \int^\infty_0 e^{s \Df} \Rf \Bf \d s
    = \int^\infty_0 e^{- s} \star \Bf^\dagger \d s
    = \star \Bf^\dagger.
\end{align*}
As such $h_\chi$ converges and $\chi^*$ is the identity in cohomology.
\end{proof}

\section{Lengthy Calculations for the Jacobi theory/1D GR case}
\label{app:LengthyCalcJac1DGR}
\subsection{Preliminaries for calculations - tensor number}
\label{app:tensornumber}

This appendix has two purposes. It serves as a preliminary for the computations, by presenting a straightforward way to compute the action of the Chevalley-Eilenberg differentials $\gamma_J$, $\gamma_{GR}$, and it provides an explaination for why they act as $\mathcal{L}_{\xi \partial_t}$ on the antifields and antighosts (see Remarks \ref{rem:Chevalley-EilenbergJacobi} and \ref{rem:Chevalley-Eilenberg1DGR}). We will be using the 1D GR theory in this discussion but all considerations hold for the Jacobi theory as well.

Let $M$ be a manifold of arbitrary dimension and $X = X^\sigma \partial_\sigma \in \mathfrak{X}(M)$. Recall that the Lie derivative $\mathcal{L}_X$ acts on the components of a tensor field $\mathbf{A} \in \mathcal{T}^n_m(M)$ of rank-$(n,m)$ as
\begin{align}
\begin{split}
    \label{eq:LieDerivativeTensorM}
    \mathcal{L}_X \mathbf{A}^{\mu_1\dots \mu_n}_{\nu_1\dots \nu_m}
    = X^\sigma \partial_\sigma \mathbf{A}^{\mu_1\dots \mu_n}_{\nu_1\dots \nu_m}
    &- \partial_\sigma X^{\mu_1}\mathbf{A}^{\sigma \dots \mu_n}_{\nu_1\dots \nu_m} 
    - \dots
    - \partial_\sigma X^{\mu_n}\mathbf{A}^{\mu_1 \dots \sigma}_{\nu_1\dots \nu_m}\\ 
    &+ \partial_{\nu_1} X^\sigma \mathbf{A}^{\mu_1\dots \mu_n}_{\sigma\dots \nu_m} 
    + \dots
    + \partial_{\nu_n} X^\sigma \mathbf{A}^{\mu_1\dots \mu_n}_{\nu_1\dots \sigma}.
\end{split}
\end{align}
Let now $M=I \subset \mathbb{R}$ denote an interval and $X = \xi \partial_t \in \mathfrak{X}(I)[I]$ be the ghost field. In this setting Equation (\ref{eq:LieDerivativeTensorM}) is greatly simplified since $\mu_i = \nu_i = t$, where $t$ is the coordinate on $I$. Let $\mathbf{A} \in \mathcal{T}^n_m(I)$ and denote its component by $A$. We define the \emph{tensor number} as $t(A) = (m-n)$. We then have 
\begin{align*}
\begin{split}
    \mathcal{L}_{\xi \partial_t} A
    &= \xi \partial_t A
    - n\partial_t \xi A + m\partial_t \xi A\\
    &=\xi \dot A + t(A) \dot \xi A.
\end{split}
\end{align*}
As an example we list the tensor number for the fields, ghosts, antifields and antighosts of the 1D GR theory
\begin{align}
    \label{eq:tensornumberFieldsAntifields}
    &t(q) = 0-0 = 0,&
    &t(g) = 2-0 = 2,&
    &t(\xi) = 0-1 = -1, \nonumber\\
    &t(q^+) = 1-0 = 1,&
    &t(g^+) = 1-2 = -1,&
    &t(\xi^+) = 2-0 = 2,
\end{align}
which explains why we claimed that the Chevalley-Eilenberg differential $\gamma_{GR}$ acts as $\mathcal{L}_{\xi \partial_t}$ on the antifields and antighosts in Remark \ref{rem:Chevalley-Eilenberg1DGR}. As such we have $\gamma_{GR} = \mathcal{L}_{\xi \partial_t}$ on all the functions on $\{q,g,q^+,g^+,\xi^+\}$ and $\gamma_{GR} = \frac{1}{2}\mathcal{L}_{\xi \partial_t}$ on the ghost. Since the ghost is a special case, we assume that the tensor fields only depend on $\{q,g,q^+,g^+,\xi^+\}$ for the rest of the discussion. When computing $\gamma_{GR}(\cdot)$, we then consider the parts with ghosts and without separately. 

The discussion until now only holds for tensor fields that only depend on the 0th-jets of $\{q,g,q^+,g^+,\xi^+\}$. The action of $\gamma_{GR}$ is then naturally extended to all jets since we assume that $\Qgr$, and as such $\gamma_{GR}$, is evolutionary, i.e.\ $[\mathcal{L}_{\gamma_{GR}},\d] = 0$. For example, if $A$ only depends on 0th-jets then
\begin{align}
\begin{split}  
    \label{eq:gammaderivative}
    \gamma_{GR} \dot A
    &= \partial_t \gamma_{GR} A
    = \partial_t [\xi \dot A + t(A) \dot \xi A]\\
    &= \xi \ddot A 
    + [1+t(A)] \dot \xi \dot A
    + t(A) \ddot \xi A.
\end{split}
\end{align}
For a general tensor field $\mathcal{A}$ which depends on arbitrary jets of the fields we have 
\begin{align}
    \label{eq:gammaA}
    \gamma_{GR} \mathcal{A}
    = \xi \dot{\mathcal{A}}
    + \sum_{n \geq 1} t_n(\mathcal{A}) \partial^n_t \xi a_n,
\end{align}
for some real scalars $t_n(\mathcal{A})$ and some functions $a_n$ that depend on the jets of $\Phi,\Phi^+ \in \F_{GR}$. In order to extend the notion of tensor number to such objects we define

\begin{definition}
Let $\mathcal{A}$ be a tensor field that depends on arbitrary jets of $\Phi,\Phi^+ \in \F_{GR}$. The tensor number $t(\mathcal{A})$ of $\mathcal{A}$ is defined as the scalar $t_1(\mathcal{A})$ in Equation (\ref{eq:gammaA}). 
\end{definition}

Note that in order to compute $\gamma_{GR}$ we only have to find out what the $t_n(\mathcal{A})$ are. For most of the computations we are only going to encounter tensor fields that depend on the $0$-th jets, and they will atmost include 2nd-jets. As such we want to find a pragmatic way of computing $t(\mathcal{A})$. If necessary, we then look at higher $t_n(\mathcal{A})$, for example by following Equation (\ref{eq:gammaderivative}). We list some useful properties of $t(\cdot)$, since they immensely simplify the explicit computations of $\gamma_{GR}(\mathcal{A})$.

\begin{proposition}
Let $\mathcal{A},\mathcal{B}$ be two tensor fields that depend on an arbitrary number of jets of $\Phi,\Phi^+ \in \F_{GR}$. The tensor number has the following properties:
\begin{enumerate}
    \item $t(\mathcal{A}\mathcal{B}) = t(\mathcal{A}) + t(\mathcal{B})$,
    \item $t(\mathcal{A}^n) = n t(\mathcal{A})$,
    \item $t(\dot{\mathcal{A}}) = 1 + t(\mathcal{A})$.
\end{enumerate}
\end{proposition}

\begin{proof}
Note that the only two terms from Equation (\ref{eq:gammaA}) that can contribute to these properties are the first two. Therefore we will only show the computations for two tensor fields $A,B$ that only depend on the 0th-jets, but they extend to the general case in a straightforward way.
\begin{enumerate}
    \item Let $gh(A) = a$. We compute
        \begin{align*}
            &\gamma_{GR} (AB) 
            = (\mathcal{L}_{\xi\partial_t} A) B + (-1)^{a} A(\mathcal{L}_{\xi\partial_t} B)\\
            &= (\xi \dot A + t(A) \dot \xi A) B
            + (-1)^{a}A (\xi \dot B + t(B) \dot \xi B)
            = \xi \partial_t(AB) + [t(A)+t(B)] \dot \xi(AB).
        \end{align*}
    \item Using that $\gamma_{GR}$ is a derivative we see that
    \begin{align*}
    \gamma_{GR} A^n = n A^{n-1} \gamma_{GR} A
    = n A^{n-1} [\xi \dot A + t(A)\dot \xi A]
    = \xi \partial_t A^n + nt(A) \dot \xi A^n.
\end{align*}
    \item This equality follows directly from Equation (\ref{eq:gammaderivative}).
\end{enumerate}
\end{proof}

We finish this section by presenting the action of $\gamma_{GR}$ tensor numbers for some relevant quantities
\begin{align}
\begin{split}
    \label{eq:tensornumberlist}
    &t(\dot q) = 1 + t(q) = 1,\\
    &t(T) = t(\|\dot q\|^2) = 2 t(\dot q) = 2,\\
    &t(u) = t\left(\frac{\dot q}{\|q\|}\right) 
    = t(\dot q) - t(\|\dot q\|) = 0,\\
    &t(\text{EL}_g) = t\left(\frac{E}{2\sqrt{g}} - \frac{T}{2g^{3/2}}\right) 
    = t\left(\frac{1}{\sqrt{g}}\right) = -\frac{1}{2} t(g) = -1.
\end{split}
\end{align}
We exemplify this method of calculating $\gamma_{GR}(\mathcal{A})$ with the computation of $\mathcal{A} = T = \frac{m}{2}\|\dot q\|^2$. Recall that $\gamma_{GR}q = \xi \dot q$ and as such $\gamma_{GR} \dot q = \xi \ddot q + \dot \xi \dot q$. $\gamma_{GR} T$ could potentially have terms proportional to $\ddot \xi$ since it depends on the derivative $\dot q$, but as there are no such terms in $\gamma_{GR} \dot q$ there won't be any in $\gamma_{GR}T$. As such we have
\begin{align*}
    \gamma_{GR} T = \xi \dot T + t(T) \dot \xi T 
    =  \xi \dot T + 2 \dot \xi T.
\end{align*}

\subsection{Lemmata used in Theorem \ref{theorem:laxequivalenceJacobi1DGR}}
\label{app:LemmataforTheoremJacobi1DGR}

In this subsection we explicitly present the lemmata used in Theorem \ref{theorem:laxequivalenceJacobi1DGR} and the detailed calculations.

\begin{lemma}
\label{lem:philaxJacobi1DGR}
Let $\phi: \Fj \rightarrow \Fgr$ be defined through
\begin{align*}
    &\phi^*q = \qj,&
    &\phi^*\xi = \xij,&
    &\phi^*g = \frac{\Tj}{E},\\
    &\phi^*q^+ = \qj^+,&
    &\phi^*\xi^+ = \xij^+,&
    &\phi^*g^+ = 0.
\end{align*}
Its pullback map $\phi^*$ is a chain map w.r.t. $(\mathcal{L}_{Q_i}-\d)$ and maps the lax BV-BFV data of the first-order theory as
\begin{align*}
    &\phi^* \tgr^\bullet = \tj^\bullet,&
    &\phi^* \Lgr^\bullet = \Lj^\bullet.
\end{align*}
\end{lemma}

\begin{proof}
The proof for the chain map condition $\phi^* \circ \Qgr = \Qj  \circ \phi^*$ is a matter of straightforward computations
\begin{align*}
    \phi^*\Qgr q 
    &= \phi^*(\xi \dot q)
    = \xij \dqj 
    = \Qj  \qj 
    = \Qj  \phi^* q,\\
    \phi^*\Qgr \xi 
    &= \phi^*(\xi \dot \xi) 
    = \xij \dot{\xij} 
    = \Qj  \xij 
    = \Qj  \phi^* \xi,\\
    \phi^*\Qgr g 
    &= \phi^*(\xij \dot g + 2 \dot{\xij} g )
    = \xij \frac{\dot{\Tj}}{E} + 2 \dot{\xij} \frac{\Tj}{E}
    = \Qj  \frac{\Tj}{E}
    = \Qj  \phi^* g,\\
    \phi^*\Qgr q^+
    &= \phi^*\left(-\partial_t\left( \frac{m\dot q}{\sqrt{g}}\right) + \xi \dot q^+ + \dot \xi q^+\right)
    = - \partial_t \left(\sqrt{\frac{E}{\Tj}} m \dqj\right) + \xij \dqj^+ + \dot{\xij} \qj^+\\
    &= \Qj  \qj^+ 
    = \Qj  \phi^* q^+ ,\\
    \phi^*\Qgr \xi^+
    &= \phi^*(- q^+ \cdot \dot q + g^+ \dot g  +2 \dot g^+ g  + \xi \dot \xi^+ + 2 \dot \xi \xi^+) 
    = - \qj^+ \cdot \dqj +  \xij\dot{\xij}^+ + 2 \dot{\xij} \xij^+\\ 
    &= \Qj  \xij^+
    = \Qj  \phi^* \xi^+,\\
    \phi^*\Qgr g^+
    &= \phi^*\left(\frac{1}{2\sqrt{g}} \left(E-\frac{T}{g}\right) + \xi \dot g^+  - \dot \xi g^+ \right)
    = 0
    = \Qj  0
    = \Qj  \phi^*g^+.
\end{align*}
$\phi^*$ is then a chain map w.r.t. $\mathcal{L}_{Q_i} - \d$ due to Lemma \ref{prop:phipsichainmapfromQ}. Applying $\phi^*$ to $(\tgr^\bullet,\Lgr^\bullet)$ gives
\begin{align*}
    \phi^*\tgr^0 &= \phi^*\left(q^+ \cdot \delta q +  \xi^+ \delta \xi + g^+ \delta g\right)
    = \qj^+ \cdot \delta \qj +  \xij^+ \delta \xij =  \tj^0,\\
    \phi^*\tgr^1  &= \phi^*\left(\frac{m\dot q}{\sqrt{g}} \cdot \delta q + q^+\xi \delta q+ g^+\xi \delta g - (2g^+ g + \xi^+ \xi) \delta \xi\right)\\
    &= \sqrt{\frac{E}{\Tj}} m \dqj \cdot \delta\qj + \qj^+\xij \delta\qj+ - \xij^+ \xij \delta \xij = \tj^1,\\
    \phi^* \Lgr^0 
    &= \phi^*\left( \frac{T}{\sqrt{g}} + \sqrt{g} E + q^+ \cdot \xi \dot{q} + g^+ (\xi \dot g + 2 g \dot \xi) + \xi^+ \xi \dot{\xi}\right)\\
    &= 2\sqrt{E\Tj} + \qj^+ \cdot \xij \dqj + \xij^+ \xij \dot{\xij} = \Lj^0,\\
    \phi^*\Lgr^1 &= \phi^*\left(\left(\frac{T}{g} - E \right) \sqrt{g}\xi\right) = 0 = \Lj^1.
\end{align*}
\end{proof}

\begin{lemma}
\label{lem:psilaxJacobi1DGR}
Let $\psi: \Fgr \rightarrow \Fj$ be the map defined through
\begin{align*}
    &\psi^*\qj = q,\\
    &\psi^*\xij = \xi,\\
    &\psi^* \qj^+_\parallel = \eta^{3/2} \left(q^+_\parallel 
        - \left[g^+\dot g + 2\dot g^+ g \right] \frac{m\dot q}{2T}
        -\frac{g^{3/2}}{E} \left[\dot{\text{EL}_g}g^+ - \text{EL}_g \dot g^+ \right]\frac{m\dot q}{2T}\right),\\
    &\psi^*\qj^+_\perp = \eta^{3/2} \left(q^+_\perp 
    +\frac{2m}{E} g^+ \ddot q_\perp \right),\\
    &\psi^*\xij^+ =  \eta^{3/2}\left(\xi^+ + \frac{g^{3/2}}{E} \dot g^+ g^+\right),
\end{align*}
where $\eta \coloneqq gE/T$ and the $\qj^+_\parallel,\qj^+_\perp$ notation works as in Equation \ref{eq:ParallelPerp}. Its pullback map $\psi^*$ is a chain map w.r.t. $(\mathcal{L}_{Q_i}-\d)$ and maps the lax BV-BFV data of the first-order theory as
\begin{align*}
    &\psi^* \tj^\bullet  = \tgr^\bullet  + (\mathcal{L}_{\Qgr} - \d) \betagr^\bullet + \delta \fgr^\bullet,\\
    &\psi^* \Lj^\bullet  = \Lgr^\bullet  + (\mathcal{L}_{\Qgr} - \d) \iota_{\Qgr} \betagr^\bullet + \d \fgr^\bullet,
\end{align*}
where
\begin{align*}
    \betagr^0 =& 
    -\frac{4g^{7/2}}{\Omega^2}Tg^+\delta g^+
    +\left(\frac{2g^2}{\Omega} + \eta^{3/2}\frac{2\sqrt{g}}{E}\right)g^+q^+_\perp \cdot \delta q\\
    &+\left(\frac{4g^{7/2}}{\Omega^2}T - \eta^{3/2}\frac{g^{3/2}}{E}\right)\dot g^+ g^+ \frac{m\dot q}{2T} \cdot \delta q 
    -(\eta^{3/2}-1)\xi^+ \frac{m\dot q}{2T} \cdot \delta q,\\
    \betagr^1 =& \,\,\xi \betagr^0 + \frac{2g^{3/2}}{\Omega}g^+m\dot q \cdot \delta q,\\
    \fgr^0 =& \,\,2 g^+\left(g - \frac{2g^{3/2}}{\Omega} T\right),\\
    \fgr^1 =& \,\,\xi \fgr^0,
\end{align*}
with $\Omega = \sqrt{g}T + g \sqrt{TE}$.
\end{lemma}

\begin{proof}
We start with the chain map condition $\psi^*\circ \Qj = \Qgr \circ \psi^*$. In the case of the fields $\{\qj,\xij\}$ we simply compute
\begin{align*}
    &\psi^*Q_J \qj 
    = \psi^*(\xij \dqj)
    = \xi \dot q = \Qgr q = \Qgr \psi^* \qj,\\
    &\psi^*Q_J \xij 
    = \psi^*(\xij \dxij)
    = \xi \dot \xi = \Qgr \xi = \Qgr \psi^* \xij.
\end{align*}

When dealing with the antifields $\{\qj^+,\xij^+\}$ it is useful to first show that $\psi^*$ is a chain map w.r.t. to the Chevalley-Eilenberg differentials and then proceed to show that is also fulfills this condition w.r.t. the Koszul-Tate differentials.

In the case of the Chevalley-Eilenberg differentials $\gamma_J,\gamma_{GR}$ it is sufficient to investigate how $\psi^*$ changes the tensorial properties of the fields, i.e.\ to analyse the tensor number introduced Section \ref{app:tensornumber}. Indeed using $\psi^* \xij = \xi$ we can compute
\begin{align}   
    \label{eq:psichainmapside1}
    \psi^*\gamma_J \tilde \Phi^+ 
    =\psi^*(\xij \dot{\tilde \Phi}^+ + t(\tilde \Phi^+) \dxij \tilde \Phi^+) = \xi \partial_t(\psi^* \tilde \Phi^+) + t(\tilde \Phi^+) \dot \xi \psi^* \tilde \Phi^+.  
\end{align} 
The most general form of the other side of the chain map condition $\gamma_{GR}\psi^*\tilde \Phi^+$ is given by
\begin{align}
    \label{eq:psichainmapside2}
    \gamma_{GR} (\psi^* \tilde \Phi^+)
    = \xi \partial_t(\psi^* \tilde \Phi^+) + t(\psi^*\tilde \Phi^+) \dot \xi \psi^* \tilde \Phi^+
    + \sum_{n \geq 2} \partial^n_t \xi a_n,
\end{align}
where the field dependent coefficients $a_n$ do not vanish trivially since the expressions for $\psi^*\tilde \Phi^+$ depend on derivative terms such as $\dot g, \dot g^+$, and $\dot{\text{EL}_g}$. In order to show that the two sides of the chain map condition given in Equations (\ref{eq:psichainmapside1}) and (\ref{eq:psichainmapside2}) are equal we need prove that $\psi^*$ preserves the tensor number $t(\cdot)$ and that the coefficients $a_n$ vanish
\begin{align*}
    &t(\tilde \Phi^+) =  t(\psi^*\tilde \Phi^+),&
    &a_n = 0.
\end{align*}
Recalling that $t(g) = t(T) = 2$, we see that the rescaling factor $\eta^{3/2}$ has vanishing tensor number
\begin{align*}
    t(\eta) = t\left(\frac{g}{T}\right) = t(g) - t(T) = 0.
\end{align*}
Furthermore, since $t_{n\geq2}(g) = t_{n\geq2}(T) = 0$, we have $\gamma_{GR}\eta = \xi \dot \eta$ and thus it can be ignored, since it neither changes $t(\cdot)$ nor $t_{n\geq2}(\cdot)$.

We start by showing $\psi^* \circ \gamma_J = \gamma_{GR} \circ \psi^*$ on the antifield $\qj^+$. Using Equations (\ref{eq:tensornumberFieldsAntifields}), (\ref{eq:gammaderivative}) and (\ref{eq:tensornumberlist}) it then follows that all the terms in $\psi^*\qj^+_\parallel$ have tensor number 1
\begin{align*}
    &t(q^+_\parallel) = t(u (u\cdot q^+))
    = t(q^+) + 2t(u)= 1,\\
    &t\left(g^+\dot g\frac{m\dot q}{2T}\right) 
    = -1 + (2 + 1) + (0+1-2) = 1,\\
    &t\left(g^+\frac{g^{3/2}}{E}\dot{\text{EL}_g} \frac{m\dot q}{2T}\right)
    = -1 + \frac{3}{2}\cdot 2 + (-1+1) + (0+1-2) = 1,\\
    &t\left(\dot g^+g\frac{m\dot q}{2T}\right) 
    = (-1+1) + 2 + (0+1-2) = 1,\\
    &t\left(\dot g^+\frac{g^{3/2}}{E}\text{EL}_g \frac{m\dot q}{2T}\right)
    = (-1+1) + \frac{3}{2} \cdot 2 - 1 + (0+1-2) = 1,
\end{align*}
showing that $t(\psi^* \qj^+_\parallel) = 1 = t(\qj^+_\parallel)$. We still need to check what happens with the terms in $\gamma_{GR} \psi^*\qj^+_\parallel$ that are proportional to $\ddot \xi$. Using Equation (\ref{eq:gammaderivative}) we can see that they take the form
\begin{align*}
    \ddot \xi \left(-\left[g^+(2g)-2g^+g\right]\frac{m\dot q}{2T}-\frac{g^{3/2}}{E}\left[-\text{EL}_gg^+ + \text{EL}_gg^+\right]\right)= 0,
\end{align*}
and thus $\psi^*  \gamma_{J} \qj^+_\parallel = \gamma_{GR} \psi^*\qj^+_\parallel$. 

In the case of $\qj^+_\perp$ we have $t(q^+_\perp) = t(q^+) + 2t(u) = t(q^+)$ since $t(u) = 0$. The same reasoning applies to $\ddot q_\perp$, but here we need to consider terms which are proportional to higher derivatives of the ghost since
\begin{align*}
    &\gamma_{GR} \ddot q 
    = \partial^2_t (\gamma_{GR} q)
    = \partial^2_t (\xi \dot q)
    = \xi \dddot q + 2 \dot \xi \ddot q + \ddot \xi \dot q.
\end{align*}
The terms proportional to $\ddot \xi$ in $\gamma_{GR}\psi^* \qj_\perp$ come from $\gamma_{GR}\ddot q$ and $u(u\cdot \gamma_{GR}\ddot q)$. As such they take the form
\begin{align*}
    \ddot \xi \left(\dot q - u(\dot q\cdot u)\right)
    = \ddot \xi \left(\dot q - \dot q \right)
    = 0,
\end{align*}
hence showing that there are no terms proportional to $\ddot \xi$ in $\gamma_{GR} \psi^*\qj^+_\perp$. Furthermore, $t(g^+\ddot q_\perp) = -1 + 2 = 1$ and as such $\psi^*\gamma_{J} \qj^+_\perp = \gamma_{GR}\psi^*\qj^+_\perp$. In order to check the chain map condition for $\xij^+$ first note that
\begin{align*}
    \gamma_{GR} \dot g^+ 
    = \partial_t(\xi \dot g^+ - \dot \xi g^+) = \xi \ddot g^+ - \ddot \xi g^+, 
\end{align*}
since $\gamma_{GR} g^+ = \xi \dot g^+ - \dot \xi g^+$. This in turn implies that
\begin{align*}
    \gamma_{GR} (\dot g^+ g^+) 
    &= \gamma_{GR} \dot g^+ g^+ -  \dot g^+ \gamma_{GR} g^+ 
    = \xi \ddot g^+ g^+ - \dot \xi \dot g^+ g^+\\
    &= \xi \partial_t(\dot g^+ g^+)
    - \dot \xi \dot g^+ g^+.
\end{align*}
As such $t(\dot g^+ g^+) = -1$ and $t(g^{3/2}\dot g^+ g^+) = \frac{3}{2} \cdot 2 - 1 = 2$. Furthermore, $t(\xi^+) = 2$ then means that $t(\psi^*\xij^+) = t(\xij^+) = 2$ and since there are no other derivative terms in $\psi^* \xij^+$, we have $a_n = 0$, which completes the proof for
\begin{align*}
    \psi^* \circ \gamma_{J} = \gamma_{GR} \circ \psi^*.
\end{align*}

We now move to the Koszul-Tate differentials. In the case of $\qj^+_\parallel$ we first note that
\begin{align*}
    \psi^*\delta_J \qj^+_\parallel 
    = \psi^*\left(\delta_J(\qj^+ \cdot \dqj) \frac{m\dqj}{2\Tj} \right)
    = -\psi^*\left(\delta^2_J \xij^+ \frac{m\dqj}{2\Tj} \right) = \psi^*(0) = 0,
\end{align*}
since $\delta_J \xij^+ = -\qj^+\cdot \dqj$ and $\delta^2_J = 0$. The term $\delta_{GR} \psi^*\qj^+_\parallel$ vanishes for a similar reason
\begin{align*}
    \delta_{GR}\psi^*(\qj^+_\parallel) 
    &=\delta_{GR}\left\{\eta^{3/2} \left(q^+_\parallel 
        - \left[g^+\dot g + 2\dot g^+ g \right] \frac{m\dot q}{2T}
        -\frac{g^{3/2}}{E} \left[\dot{\text{EL}_g}g^+ - \text{EL}_g \dot g^+ \right]\frac{m\dot q}{2T}\right)\right\}\\
    &= -\eta^{3/2} \delta^2_{GR}\left\{\xi^+ + \frac{g^{3/2}}{E} \dot g^+ g^+\right\} \frac{m\dot q}{2T} = 0,
\end{align*}
where we used that $\delta_{GR} \xi^+ = -q^+\cdot \dot q + g^+\dot g + 2\dot g^+ g$ and $\delta_{GR}(\dot g^+ g^+) = \dot{\text{EL}_g}g^+ - \dot g^+ \text{EL}_g$. 

The computations for the perpendicular part of $\qj^+$ go as follows
\begin{align*}
    &\psi^* \delta_J \qj^+_\perp 
    = \psi^*\left(- \partial_t \left(\sqrt{\frac{E}{\Tj}} m \dqj\right) 
    + \tilde u \partial_t \left(\sqrt{\frac{E}{\Tj}} m \dqj\right) \cdot \tilde u \right) \\
    &=\psi^*\left( \frac{1}{2} \left(\frac{E}{\Tj}\right)^{3/2} \frac{\dot{\Tj}}{E} m \dqj 
    - \sqrt{\frac{E}{\Tj}} m \ddqj 
    + \tilde u \sqrt{2mE} \underbrace{\dot{\tilde u} \cdot \tilde u}_{=0}
    \right)
    = \left(\frac{E}{ T}\right)^{3/2} \left(\frac{\dot{T}}{2E} m \dot{q} - \frac{T}{E}m \ddot{ q} \right),
\end{align*}
the other side of the equation reads
\begin{align*}
    &\delta_{GR} \psi^* \qj^+_\perp
    = \eta^{3/2} \delta_{GR} \left( q^+_\perp + \frac{2m}{E}\ddot q_\perp g^+\right)
    =  \eta^{3/2} \left(
    -\partial_t\left( \frac{m\dot q}{\sqrt{g}}\right) 
    + u \partial_t\left( \frac{m\dot q}{\sqrt{g}}\right) \cdot u
    + \frac{2m}{E}\ddot q_\perp \text{EL}_g\right)\\
    &=\eta^{3/2} \left(-\frac{m\ddot q}{\sqrt{g}} 
    - \cancel{\partial_t\left( \frac{1}{\sqrt{g}}\right) m \dot q}
    + u \frac{m \ddot q \cdot u}{\sqrt{g}}
    + \cancel{\partial_t\left( \frac{1}{\sqrt{g}}\right) m \dot q}
    + \frac{2m}{E}\ddot q_\perp \text{EL}_g\right)\\
    &= \eta^{3/2} \left(-\frac{m\ddot q}{\sqrt{g}} 
    + \frac{m\dot q}{2T} \frac{m \ddot q \cdot \dot q}{\sqrt{g}} 
    + \frac{2m}{E}\ddot q_\perp \text{EL}_g\right)
    =\eta^{3/2} \left(-\frac{m\ddot q}{\sqrt{g}} 
    + \frac{\dot T}{2T\sqrt{g}} m\dot q 
    + \frac{2m}{E}\ddot q_\perp \text{EL}_g\right),
\end{align*}
where we have used that $T = m\|\dot q\|/2$ and $\dot T = m \ddot q \cdot \dot q$. The last term can be expanded to give
\begin{align*}
    \frac{2m}{E}\ddot q_\perp \text{EL}_g
    &= \frac{2m}{E}\ddot q \left(\frac{E}{2\sqrt{g}} - \frac{T}{2g^{3/2}} \right)
    - \frac{2m}{E} \dot q \frac{\ddot q \cdot \dot q}{\|\dot q\|^2}
    \left(\frac{E}{2\sqrt{g}} - \frac{T}{2g^{3/2}} \right)\\
    &= \frac{m\ddot q}{\sqrt{g}} - \frac{mT}{Eg^{3/2}} \ddot q
    - \frac{\dot T}{2T\sqrt{g}} m \dot q 
    + \frac{\dot T}{2Eg^{3/2}} m \dot q.
\end{align*}
Putting everything together results in
\begin{align*}
    \delta_{GR} \psi^* \qj^+_\perp
    = \eta^{3/2}\left(\frac{\dot T}{2Eg^{3/2}} m \dot q - \frac{T}{Eg^{3/2}} m\ddot q\right)
    = \left(\frac{E}{T}\right)^{3/2} \left(\frac{\dot{T}}{2E} m \dot{q} - \frac{T}{E}m \ddot{ q} \right),
\end{align*}
which shows that $\psi^* \delta_{J} \qj^+_\perp = \delta_{GR} \psi^* \qj^+_\perp$. Finally we show that $\psi^*$ acts as a chain map w.r.t. $\delta_{KT}$ on $\xi^+$. We have
\begin{align*}
    \psi^* \delta_{J} \xij^+ 
    &= \psi^*(-\qj^+ \cdot \dqj)
    = \psi^*(-\qj^+) \cdot \dot q\\
    &= \eta^{3/2} \left(-q^+ \cdot \dot q 
        + \left[g^+\dot g + 2\dot g^+ g \right]
        + \frac{g^{3/2}}{E} \left[\dot{\text{EL}_g}g^+ - \text{EL}_g \dot g^+ \right]\right),\\
    \delta_{GR} \psi^* \xij^+ &= \eta^{3/2} 
    \delta_{GR}\left(\xi^+ + \frac{g^{3/2}}{E} \dot g^+ g^+\right)\\
    &= \eta^{3/2} \left(-q^+\cdot \dot q + g^+ \dot g + 2 \dot g^+ g + \frac{g^{3/2}}{E} \dot{\text{EL}_g} g^+ - \frac{g^{3/2}}{E} \dot g^+ \text{EL}_g \right) \\
    &= \eta^{3/2} \left(-q^+ \cdot \dot q 
        + \left[g^+\dot g + 2\dot g^+ g \right]
        + \frac{g^{3/2}}{E} \left[\dot{\text{EL}_g}g^+ - \text{EL}_g \dot g^+ \right]\right),
\end{align*}
finally showing that 
\begin{align*}
    \psi^* \circ \delta_{J} = \delta_{GR} \circ \psi^*,
\end{align*}
which show that $\psi^*$ is indeed a chain map. 

We now present the calculations for $(\psi^*\theta^\bullet_J,\psi^*L^\bullet_J )$. Recall that $\theta^\bullet_J = \theta^0_J \,\d t + \theta^1_J$, $L^\bullet_J = L^0_J \,\d t + L^1_J$ with 
\begin{align*}
    &\theta^0_J = \qj^+ \cdot \delta \qj + \xij^+ \delta \xij,&
    &\theta^1_J = \sqrt{\frac{E}{\Tj}} m \dqj \cdot \delta \qj + \qj^+ \xij \delta \qj - \xij^+ \xij \delta \xij,\\
    &L^0_J = 2\sqrt{E \Tj} + \qj^+ \cdot \xij \dqj + \xij^+ \xij \dxij,&
    &L^1_J = 0.
\end{align*}
as in Proposition/Definition \ref{propdef:laxJacobi}.
Specifically we want compute
\begin{align*}
    &\Delta \theta^0 \mathrm{d}t
    = \mathcal{L}_{\Qgr} \beta^0 \mathrm{d}t- \mathrm{d}\beta^1 + \delta f^0 \mathrm{d}t,\\
    &\Delta \theta^1 
    = \mathcal{L}_{\Qgr} \beta^1 + \delta f^1 \mathrm{d}t,\\
    &\Delta L^0 \mathrm{d}t
    = \mathcal{L}_{\Qgr} \iota_{\Qgr} \beta^0 \mathrm{d}t - \mathrm{d}\iota_{\Qgr} \beta^1 + \mathrm{d}f^1,
\end{align*}
where we skip $\Delta L^1$ as it is determined by $\Delta \theta^\bullet$. We will from now on drop the label $GR$ in order to keep the calculations cleaner, we will denote the Koszul-Tate operator by $\delta_{KT}$ to avoid confusion with the de Rham differential $\delta$. The terms on the left hand sides can be computed explicitly by using the form of $\psi^*$. We get:
\begin{align*}
    \Delta \theta^0
    =&\,\,(\eta^{3/2} - 1) q^+ \cdot \delta q 
    + (\eta^{3/2} - 1) \xi^+ \delta \xi 
    - g^+ \delta g\\
    &- \eta^{3/2} \left[g^+\dot g + 2\dot g^+ g \right] \frac{m\dot q}{2T} \cdot \delta q
    -\eta^{3/2}\frac{g^{3/2}}{E} \left[\dot{\text{EL}_g}g^+ - \text{EL}_g \dot g^+ \right]\frac{m\dot q}{2T}\cdot \delta q\\
    &+ \eta^{3/2} \frac{2m}{E} g^+ \ddot q_\perp \cdot \delta q
    + \eta^{3/2} \frac{g^{3/2}}{E} \dot g^+ g^+ \delta \xi,\\
    \Delta \theta^1 
    = &\,\,\left[\sqrt{\frac{E}{T}}m\dot q - \frac{m\dot q}{\sqrt{g}} \right] \cdot \delta q\\
    &+ (\eta^{3/2}-1) q^+\xi \delta q 
    - (\eta^{3/2}-1) \xi^+\xi \delta \xi 
    - g^+\xi \delta g + 2g^+ g \delta \xi \\
    &- \eta^{3/2} \left[g^+\dot g + 2\dot g^+ g \right] \frac{m\dot q}{2T} \cdot \xi \delta q
    -\eta^{3/2}\frac{g^{3/2}}{E} \left[\dot{\text{EL}_g}g^+ - \text{EL}_g \dot g^+ \right]\frac{m\dot q}{2T}\cdot \xi \delta q\\
    &+\eta^{3/2}\frac{2m}{E} g^+ \xi \ddot q_\perp \cdot \delta q
    + \eta^{3/2} \frac{g^{3/2}}{E} \dot g^+ g^+ \xi \delta \xi\\
    =&\,\, \left[\sqrt{\frac{E}{T}}m\dot q - \frac{m\dot q}{\sqrt{g}} \right] \cdot \delta q + 2g^+g\delta \xi - \xi\Delta \theta^1,\\
    \Delta L^0
    =&\,\,2\sqrt{ET} -\frac{T}{\sqrt{g}} - \sqrt{g}E\\
    &+(\eta^{3/2} - 1) q^+ \cdot \xi \dot q 
    + (\eta^{3/2} - 1) \xi^+ \xi\dot \xi 
    - g^+ (\xi \dot g + 2\dot \xi g)\\
    &- \eta^{3/2} \left[g^+\dot g + 2\dot g^+ g \right] \xi
    -\eta^{3/2}\frac{g^{3/2}}{E} \left[\dot{\text{EL}_g}g^+ - \text{EL}_g \dot g^+ \right]\xi\\
    &+ \eta^{3/2} \frac{g^{3/2}}{E} \dot g^+ g^+ \xi \dot \xi.
\end{align*}
The following identities are going to be used throughout the calculations:
\begin{align}
    \label{eq:usefulidentity1}
    \sqrt{\frac{E}{T}} - \frac{1}{\sqrt{g}} &= \frac{2g^{3/2}}{\Omega} \text{EL}_g,\\
    \label{eq:usefulidentity2}
    2\sqrt{ET} - \frac{T}{\sqrt{g}} - \sqrt{g}E &= - \frac{4g^{7/2}}{\Omega^2}T \text{EL}_g^2,
\end{align}
where $\Omega = \sqrt{g}T + g\sqrt{TE}$ with $t(\Omega) = 3$. To see that the first Equation (\ref{eq:usefulidentity1}) holds we compute
\begin{align*}
    \sqrt{\frac{E}{T}} - \frac{1}{\sqrt{g}}
    &=\frac{1}{\sqrt{gT}} \left(\sqrt{gE}-\sqrt{T}\right)
    =\frac{gE-T}{\sqrt{gT}\left[\sqrt{T} + \sqrt{gE}\right]}\\
    &= \frac{2g^{3/2}}{\Omega} \left(\frac{E}{2\sqrt{g}} - \frac{T}{2g^{3/2}}\right)
    = \frac{2g^{3/2}}{\Omega} \text{EL}_g.
\end{align*}
For the second one (\ref{eq:usefulidentity2}) we have
\begin{align*}
    2\sqrt{ET} - \frac{T}{\sqrt{g}} - \sqrt{g}E= -\left(g^{1/4} \sqrt{E}- \frac{\sqrt{T}}{g^{1/4}}\right)^2.
\end{align*}
The term in the brackets can be changed to
\begin{align*}
    g^{1/4} \sqrt{E} - \frac{\sqrt{T}}{g^{1/4}} 
    &= g^{1/4}\left(\sqrt{E}- \sqrt{\frac{T}{g}}\right)
    = \frac{g^{1/4}}{\sqrt{E} + \sqrt{\frac{T}{g}}} \left(E-\frac{T}{g}\right)\\
    &= \frac{g^{5/4}\sqrt{T}}{\sqrt{g}T + g\sqrt{TE}}(2\sqrt{g} \text{EL}_g)
    = \frac{2g^{7/4}\sqrt{T}}{\Omega}\text{EL}_g,
\end{align*}
and as such
\begin{align*}
     2\sqrt{ET} - \frac{T}{\sqrt{g}} - \sqrt{g}E = - \frac{4g^{7/2}}{\Omega^2}T \,\text{EL}_g^2.
\end{align*}

\textbf{Computation of $\Delta \theta^0$:} We want to show 
\begin{align*}
    \Delta \theta^0 \mathrm{d}t
    = \mathcal{L}_{Q} \beta^0 \mathrm{d}t- \mathrm{d}\beta^1 + \delta f^0 \mathrm{d}t.
\end{align*}
In order to compute $\mathcal{L}_Q\beta^0$ we decompose the cohomological vector field as $Q = \gamma + \delta_{KT}$. Starting with $\mathcal{L}_\gamma$, we write $\beta^0 = A_1 \delta g^+ + A_\perp \cdot \delta q + A_2 \cdot \delta q$ where 
\begin{align}
\begin{split}
    \label{eq:Asforbeta^1}
    &A_1 = -\frac{4g^{7/2}}{\Omega^2}Tg^+,\\
    &A_\perp = \left(\frac{2g^2}{\Omega} + \eta^{3/2}\frac{2\sqrt{g}}{E}\right)g^+q^+_\perp,\\
    &A_2 = \left(\frac{4g^{7/2}}{\Omega^2}T - \eta^{3/2}\frac{g^{3/2}}{E}\right)\dot g^+ g^+ \frac{m\dot q}{2T} 
    -(\eta^{3/2}-1)\xi^+ \frac{m\dot q}{2T},
\end{split}
\end{align}
with the following tensor and ghost numbers:
\begin{align*}
    &t(A_1) = \frac{7}{2} \cdot 2 - 2\cdot3 + 2 -1 = 2,&
    &|A_1| = -1,\\
    &t(A_\perp) = 2\cdot 2 - 3 - 1 + 1 = 1,&
    &|A_\perp| = 1,\\
    &t(A_2) = \frac{7}{2} \cdot 2 - 2\cdot3 + 2 - 1 + 1 - 2 = 1,&
    &|A_2| = -2.
\end{align*}
We then have
\begin{align*}
    \mathcal{L}_\gamma (A_1 \delta g^+) 
    &= \gamma A_1 \delta g^+ + A_1 \delta \gamma g^+
    = \xi \dot A_1 \delta g^+ + 2 \dot \xi A_1 \delta g^+ + A_1 \delta(\xi \dot g^+ - \dot \xi g^+)\\
    &= \xi \dot A_1 \delta g^+ + \cancel{2} \dot \xi A_1 \delta g^+ + A_1 \delta\xi \dot g^+ - A_1 \xi \delta \dot g^+ + \cancel{A_1 \dot \xi \delta g^+}\\
    &= \partial_t \left(\xi A_1\delta g^+\right) + A_1 \dot g^+ \delta \xi,\\
    \mathcal{L}_\gamma (A_\perp \delta q) 
    &= \gamma A_\perp \cdot \delta q - A_\perp \cdot \delta \gamma q
    = \partial_t(\xi A_\perp) \cdot \delta q - A_\perp \cdot \delta (\xi \dot q)\\
    &= \partial_t(\xi A_\perp) \cdot \delta q - \cancel{A_\perp \cdot \dot q} \delta \xi + A_\perp \xi \delta \dot q\\
    &= \partial_t(\xi A_\perp \cdot \delta q)\\
    \mathcal{L}_\gamma (A_2 \cdot \delta q)
    &= \gamma A_2 \cdot \delta q - A_2 \cdot \delta \gamma q
    = \partial_t(\xi A_2) \cdot \delta q - A_2 \cdot \delta (\xi \dot q)\\
    &= \partial_t(\xi A_2) \cdot \delta q - A_2 \cdot \dot q \delta \xi + A_2 \xi \delta \dot q\\
    &= \partial_t(\xi A_2 \cdot \delta q) - A_2 \cdot \dot q \delta \xi,
\end{align*}
which shows
\begin{align*}
    \mathcal{L}_\gamma \beta^0 &= \partial_t(\xi \beta^0) 
    - \frac{4g^{7/2}}{\Omega^2}Tg^+ \dot g^+ 
    - \left(\frac{4g^{7/2}}{\Omega^2}T - \eta^{3/2}\frac{g^{3/2}}{E}\right)\dot g^+ g^+ \delta \xi 
    + (\eta^{3/2}-1)\xi^+ \delta \xi\\
    &=\partial_t(\xi \beta^0) 
    + \eta^{3/2}\frac{g^{3/2}}{E} \dot g^+ g^+ \delta \xi 
    + (\eta^{3/2}-1)\xi^+ \delta \xi.
\end{align*}
Before addressing the computation of $\mathcal{L}_{\delta_{KT}}\beta^0$ we note that
\begin{align*}
    \delta \text{EL}_g &= \left(-\frac{E}{4g^{3/2}} + \frac{3T}{4g^{5/2}}\right) \delta g
    - \frac{\delta T}{2g^{3/2}},\\
    \delta \Omega &= \delta\left(\sqrt{g}T +g\sqrt{ET}\right)
    = \left(\frac{T}{2\sqrt{g}} + \sqrt{ET}\right) \delta g
    + \left(\sqrt{g} + \frac{g\sqrt{E}}{2\sqrt{T}}\right) \delta T\\
    &= \left(\Omega - \frac{\sqrt{g}T}{2} \right)\frac{\delta g}{g}
    + \left(\Omega - \frac{g\sqrt{ET}}{2}\right)\frac{\delta T}{T},
    \\
    \Omega^2 &= gT^2 + 2 g^{3/2} T^{3/2} \sqrt{E} + g^2 ET,\\
    \frac{4g^{7/2}}{\Omega^2}T \text{EL}_g
    &= \frac{2g^2}{\Omega} T \left(\sqrt{\frac{E}{T}} - \frac{1}{\sqrt{g}}\right)
    = \frac{2}{\Omega} \left(g^2\sqrt{ET} - g^{3/2}T\right)\\
    &= \frac{2}{\Omega} \left(g\Omega - 2g^{3/2} T\right)
    = 2g- \frac{4g^{3/2}}{\Omega} T,\\
    \delta\left(\frac{g^{3/2}}{\Omega}\right) 
    &= \frac{3\sqrt{g}}{2\Omega}\delta g 
    - \frac{g^{3/2}}{\Omega^2}\left(\Omega - \frac{\sqrt{g}T}{2} \right)\frac{\delta g}{g}
    - \frac{g^{3/2}}{\Omega^2}\left(\Omega - \frac{g\sqrt{ET}}{2}\right)\frac{\delta T}{T}\\
    &= \frac{1}{2}\left(\frac{\sqrt{g}}{\Omega} + \frac{gT}{\Omega^2}\right) \delta g
    - \left(\frac{g^{3/2}}{T\Omega} - \frac{g^{5/2}\sqrt{E}}{2\sqrt{T}\Omega^2} \right) \delta T.
\end{align*}
The calculation for the first term in $\mathcal{L}_{\delta_{KT}}\beta^0$ goes as follows
\begin{align*}
    \mathcal{L}_{\delta_{KT}}&\left(-\frac{4g^{7/2}}{\Omega^2}Tg^+ \delta g^+\right)
    = -\frac{4g^{7/2}}{\Omega^2}T \text{EL}_g \delta g^+ 
    -\frac{4g^{7/2}}{\Omega^2}Tg^+\delta \text{EL}_g\\
    = &-\delta\left(\frac{4g^{7/2}}{\Omega^2}T \text{EL}_g g^+ \right)
    + \delta\left(\frac{4g^{7/2}}{\Omega^2}T \text{EL}_g\right) g^+ -\frac{4g^{7/2}}{\Omega^2}Tg^+\delta \text{EL}_g \\
    =&\,\, \delta\left(\frac{4g^{3/2}}{\Omega} T g^+ - 2gg^+\right)
    + 2\delta g g^+ 
    - \delta\left(\frac{4g^{3/2}}{\Omega}\right) T g^+ 
    -\frac{4g^{3/2}}{\Omega} \delta T g^+ \\
    &-\frac{4g^{7/2}}{\Omega^2}Tg^+\left(-\frac{E}{4g^{3/2}} + \frac{3T}{4g^{5/2}}\right) \delta g
    +\frac{4g^{7/2}}{\Omega^2}Tg^+\frac{\delta T}{2g^{3/2}}.
\end{align*}
Gathering everything in front of $\delta g$ we have
\begin{align*}
    &g^+\left[-2
    +\left(\frac{2\sqrt{g}}{\Omega} 
    + \frac{2gT}{\Omega^2}\right)T
    +\frac{g^2ET-3gT^2}{\Omega^2} \right] \delta g\\
    &= g^+\left[-2
    +\frac{2\sqrt{g}T}{\Omega} 
    +\frac{g^2ET-gT^2}{\Omega^2} \right] \delta g\\
    &= g^+\left[-2
    +\frac{2\Omega - 2g\sqrt{ET}}{\Omega} 
    +\frac{\Omega^2 - 2 g^{3/2}T^{3/2}\sqrt{E} - 2gT^2}{\Omega^2} \right] \delta g\\
    &= g^+\left[1 
    - \frac{2g\sqrt{ET}\Omega +2 g^{3/2}T^{3/2}\sqrt{E}+2gT^2}{\Omega^2} \right]\delta g\\
    &= g^+\left[1 
    - \frac{2g^2ET +4 g^{3/2}T^{3/2}\sqrt{E}+2gT^2}{\Omega^2} \right]\delta g\\
    &= g^+\left[1 
    - \frac{2\Omega^2}{\Omega^2} \right]\delta g = -g^+ \delta g.
\end{align*}
The terms in front of $\delta T$ simplify to
\begin{align*}
    g^+&\left[- \left(\cancel{\frac{4g^{3/2}}{\Omega} }
    - \frac{2g^{5/2}\sqrt{ET}}{\Omega^2} \right) 
    + \cancel{\frac{4g^{3/2}}{\Omega}}
    + \frac{2g^2T}{\Omega^2}\right]\delta T\\
    &=g^+\frac{2g^{3/2}}{\Omega}\left[\frac{g\sqrt{ET}+\sqrt{g}T}{\Omega}\right]\delta T
    = g^+\frac{2g^{3/2}}{\Omega}\delta T,
\end{align*}
and as such:
\begin{align*}
    \mathcal{L}_{\delta_{KT}}\left(-\frac{4g^{7/2}}{\Omega^2}Tg^+ \delta g^+\right)
    = \delta\left(\frac{4g^{3/2}}{\Omega} T g^+ - 2gg^+\right)
    + g^+ \left(\frac{2g^{3/2}}{\Omega}\delta T - \delta g\right).
\end{align*}
The rest of  $\mathcal{L}_{\delta_{KT}}\beta^0$ yields
\begin{align*}
    \mathcal{L}_{\delta_{KT}}\left(\frac{2g^2}{\Omega} g^+q^+_\perp \cdot \delta q, \right)
    &=\frac{2g^2}{\Omega}
    \left(\text{EL}_gq^+_\perp + g^+\frac{m \ddot q_\perp}{\sqrt{g}}\right) \cdot \delta q\\
    &= \sqrt{g}\left(\sqrt{\frac{E}{T}} - \frac{1}{\sqrt{g}}\right) q^+_\perp \cdot \delta q
    +g^+ \frac{2g^{3/2}}{\Omega} m \ddot q_\perp \cdot \delta q\\
    &= (\eta^{1/2} - 1)q^+_\perp \cdot \delta q 
    +g^+ \frac{2g^{3/2}}{\Omega} m \ddot q_\perp \cdot \delta q,\\
    \mathcal{L}_{\delta_{KT}}\left(\eta^{3/2}\frac{2\sqrt{g}}{E} g^+q^+_\perp \cdot \delta q \right)
    &= \eta^{3/2}\frac{2\sqrt{g}}{E}\left[\left(\frac{E}{2\sqrt{g}}-\frac{T}{2g^{3/2}}\right)q^+_\perp + g^+\frac{m \ddot q_\perp}{\sqrt{g}}\right] \cdot \delta q\\
    &= (\eta^{3/2}-\eta^{1/2}) q^+_\perp \cdot \delta q
    + \eta^{3/2}\frac{2}{E} g^+ m\ddot q_\perp \cdot \delta q,\\
    \mathcal{L}_{\delta_{KT}}\left(\frac{4g^{7/2}}{\Omega^2}T \dot g^+ g^+ \frac{m\dot q}{2T} \cdot \delta q \right)
    &= \frac{4g^{7/2}}{\Omega^2}T \left[\dot{\text{EL}_g} g^+ - \dot g^+ \text{EL}_g\right] \frac{m\dot q}{2T} \cdot \delta q, \\
    \mathcal{L}_{\delta_{KT}}\left(- \eta^{3/2}\frac{g^{3/2}}{E}\dot g^+ g^+ \frac{m\dot q}{2T} \cdot \delta q \right)
    &= - \eta^{3/2}\frac{g^{3/2}}{E}\left[\dot{\text{EL}_g} g^+ - \dot g^+ \text{EL}_g\right] \frac{m\dot q}{2T} \cdot \delta q, \\
    \mathcal{L}_{\delta_{KT}}\left(-(\eta^{3/2}-1)\xi^+ \frac{m\dot q}{2T} \cdot \delta q\right)
    &= (\eta^{3/2}-1)\left[q^+\cdot \dot q - g^+ \dot g - 2\dot g^+ g\right] \frac{m\dot q}{2T} \cdot \delta q\\
    &=(\eta^{3/2}-1)q^+_\parallel \cdot \delta q
    -(\eta^{3/2}-1)\left[g^+ \dot g + 2\dot g^+ g\right] \frac{m\dot q}{2T} \cdot \delta q.
\end{align*}
Gathering everything gives
\begin{align*}
    \mathcal{L}_Q\beta^0 =&\,\, \partial_t(\xi \beta^0)
    + \eta^{3/2}\frac{g^{3/2}}{E}\dot g^+ g^+ \delta \xi 
    + (\eta^{3/2}-1)\xi^+ \delta \xi\\
    &+\delta\left(\frac{4g^{3/2}}{\Omega} T g^+ - 2gg^+\right)
    + g^+ \left(\frac{2g^{3/2}}{\Omega}\delta T - \delta g\right)\\
    &+(\eta^{1/2} - 1)q^+_\perp \cdot \delta q 
    +g^+ \frac{2g^{3/2}}{\Omega} m \ddot q_\perp \cdot \delta q\\
    &+(\eta^{3/2}-\eta^{1/2}) q^+_\perp \cdot \delta q
    + \eta^{3/2}\frac{2}{E} g^+ m\ddot q_\perp \cdot \delta q\\
    &+\left(\frac{4g^{7/2}}{\Omega^2}T - \eta^{3/2}\frac{g^{3/2}}{E}\right)\left[\dot{\text{EL}_g} g^+ - \dot g^+ \text{EL}_g\right] \frac{m\dot q}{2T} \cdot \delta q \\
    &+(\eta^{3/2}-1)q^+_\parallel \cdot \delta q
    -(\eta^{3/2}-1)\left[g^+ \dot g + 2\dot g^+ g\right] \frac{m\dot q}{2T} \cdot \delta q\\
    = &\,\,\Delta \theta^0 + \partial_t(\xi \beta^0)
    -\delta f^0
    +  \frac{2g^{3/2}}{\Omega}g^+ \delta T 
    +g^+ \frac{2g^{3/2}}{\Omega} m \ddot q_\perp \cdot \delta q\\
    &+\frac{4g^{7/2}}{\Omega^2}T \left[\dot{\text{EL}_g} g^+ - \dot g^+ \text{EL}_g\right] \frac{m\dot q}{2T} \cdot \delta q
    +\left[g^+ \dot g + 2\dot g^+ g\right] \frac{m\dot q}{2T} \cdot \delta q,
\end{align*}
which in turn implies
\begin{align}
    \label{eq:IncompleteDeltatheta^1}
    \Delta \theta^0
    = &\,\,\mathcal{L}_Q \beta^0 - \partial_t\left(\xi \beta^0+\frac{2g^{3/2}}{\Omega}g^+m\dot q \cdot \delta q\right)
    + \delta f^0
    +\partial_t\left(\frac{2g^{3/2}}{\Omega}g^+m\dot q \cdot \delta q\right)\nonumber \\
    &-  \frac{2g^{3/2}}{\Omega}g^+ \delta T 
    -g^+ \frac{2g^{3/2}}{\Omega} m \ddot q_\perp \cdot \delta q\nonumber \\
    &-\frac{4g^{7/2}}{\Omega^2}T \left[\dot{\text{EL}_g} g^+ - \dot g^+ \text{EL}_g\right] \frac{m\dot q}{2T} \cdot \delta q 
    -\left[g^+ \dot g + 2\dot g^+ g\right] \frac{m\dot q}{2T} \cdot \delta q\nonumber \\
    =&\,\,\mathcal{L}_Q \beta^0 - \partial_t\beta^1
    + \delta f^0\nonumber \\
    &+g^+\left[
    \partial_t\left(\frac{2g^{3/2}}{\Omega}\right) m\dot q \cdot \delta q
    +\frac{2g^{3/2}}{\Omega} m\ddot q \cdot \delta q 
    + \frac{2g^{3/2}}{\Omega} m\dot q \cdot \delta \dot q 
    - \frac{2g^{3/2}}{\Omega} \delta T \right.\nonumber \\
    &- \frac{2g^{3/2}}{\Omega} m \ddot q_\perp \cdot \delta q 
    - \left.\frac{4g^{7/2}}{\Omega^2}T \dot{\text{EL}_g}\frac{m\dot q}{2T} \cdot \delta q - \dot g \frac{m\dot q}{2T} \cdot \delta q 
    \right]\nonumber\\
    &+\dot g^+\left[
    \left(\frac{4g^{3/2}}{\Omega}T - 2g\right)
    +\frac{4g^{7/2}}{\Omega^2}T \text{EL}_g
    \right]\frac{m\dot q}{2T} \cdot \delta q.
\end{align}
The last term vanishes due to Equation (\ref{eq:usefulidentity2}). In order to show that the term proportional to $g^+$ vanishes as well we note:
\begin{align*}
    \delta T &= m\dot q \cdot \delta \dot q,\\
    \ddot q_\perp  &= \ddot q - \dot q\frac{\dot q \cdot \ddot q}{\|\dot q\|^2}
    = \ddot q - \dot q \frac{\dot T}{2T},\\
    \dot{\text{EL}_g} &= \left(-\frac{E}{4g^{3/2}} + \frac{3T}{4g^{5/2}}\right) \dot g
    - \frac{\dot T}{2g^{3/2}},\\
    \partial_t\left(\frac{2g^{3/2}}{\Omega}\right) 
    &= \left(\frac{\sqrt{g}}{\Omega} + \frac{gT}{\Omega^2}\right) \dot g
    - \left(\frac{2g^{3/2}}{T\Omega} - \frac{g^{5/2}\sqrt{E}}{\sqrt{T}\Omega^2} \right) \dot T.
\end{align*}
Then
\begin{align*}
    &\,\,g^+\left[
    \partial_t\left(\frac{2g^{3/2}}{\Omega}\right) m\dot q \cdot \delta q
    +\frac{2g^{3/2}}{\Omega} m\ddot q \cdot \delta q 
    + \cancel{\frac{2g^{3/2}}{\Omega} m\dot q \cdot \delta \dot q }
    - \cancel{\frac{2g^{3/2}}{\Omega} \delta T }\right.\\
    &- \frac{2g^{3/2}}{\Omega} m \ddot q_\perp \cdot \delta q 
    - \left. \frac{4g^{7/2}}{\Omega^2}T \dot{\text{EL}_g}\frac{m\dot q}{2T} \cdot \delta q - \dot g \frac{m\dot q}{2T} \cdot \delta q 
    \right]\\
    =&\,\,g^+\left[
    \left(\frac{\sqrt{g}}{\Omega} + \frac{gT}{\Omega^2}\right) \dot g m \dot q
    - \left(\frac{2g^{3/2}}{T\Omega} - \frac{g^{5/2}\sqrt{E}}{\sqrt{T}\Omega^2} \right) \dot T m\dot q \right.\\
    &+\cancel{\frac{2g^{3/2}}{\Omega} m\ddot q \cdot \delta q} - \frac{2g^{3/2}}{\Omega} \left(\cancel{m\ddot q} - m\dot q \frac{\dot T}{2T}\right) \\
    &- \left.\frac{4g^{7/2}}{\Omega^2}T \left(\left(-\frac{E}{4g^{3/2}} + \frac{3T}{4g^{5/2}}\right) \dot g
    - \frac{\dot T}{2g^{3/2}}\right)\frac{m\dot q}{2T} - \dot g \frac{m\dot q}{2T}
    \right] \cdot \delta q\\
    =&\,\,g^+ \frac{\dot g }{\Omega^2}
    \left[ 
    2\sqrt{g}T\Omega + 2gT^2
    + g^{2}TE - 3gT^2 - \Omega^2
    \right] \frac{m\dot q}{2T}\cdot \delta q\\
    &+g^+   \frac{\dot T}{\Omega^2}
    \left[ 
    - 4g^{3/2}\Omega +2g^{5/2}\sqrt{ET}
    +2g^{3/2}\Omega
    + 2g^2T
    \right]\frac{m\dot q}{2T}\cdot \delta q\\
    =&\,\,g^+ \frac{\dot g }{\Omega^2}
    \left[ 
    \cancel{2}gT^2 + 2g^{3/2}T^{3/2} \sqrt{E}
    + g^{2}TE - \cancel{gT^2} - \Omega^2
    \right] \frac{m\dot q}{2T}\cdot \delta q\\
    &+g^+   \frac{\dot T}{\Omega^2}
    \left[ 
    - 2g^{3/2}\Omega +
    2g^{3/2}(\sqrt{g}T + g\sqrt{ET})
    \right]\frac{m\dot q}{2T}\cdot \delta q\\
    =&\,\,0.
\end{align*}
Taking this into account and introducing $\d t$ in Equation (\ref{eq:IncompleteDeltatheta^1}) yields
\begin{align}
    \label{eq:Deltatheta^1endresult}
    \Delta \theta^0 \mathrm{d}t = \mathcal{L}_Q \beta^0 \mathrm{d}t - \partial_t\beta^1 \mathrm{d}t + \delta f^0 \mathrm{d}t 
    = \mathcal{L}_Q \beta^0 \mathrm{d}t - \mathrm{d}\beta^1 + \delta f^0 \mathrm{d}t,
\end{align}
since $|\beta^1|=0$.

\textbf{Computation of $\Delta \theta^1$:} We want to compute
\begin{align*}
    \Delta \theta^1 = \mathcal{L}_Q \beta^1 + \delta f^1.
\end{align*}
First note that Equation (\ref{eq:Deltatheta^1endresult}) implies
$\mathcal{L}_Q \beta^0 = \Delta \theta^0 + \partial_t \beta^1- \delta f^0$, which we use to compute $\mathcal{L}_Q \beta^1$. Since 
\begin{align*}
   \beta^1 = \xi \beta^0 + \frac{2g^{3/2}}{\Omega}g^+m\dot q \cdot \delta q,
\end{align*}
we have
\begin{align*}
    \mathcal{L}_Q \beta^1 = \xi \dot \xi \beta^0 - \xi \left(\Delta \theta^0 + \partial_t \beta^1 - \delta f^0\right)+ \mathcal{L}_Q\left(\frac{2g^{3/2}}{\Omega}g^+m\dot q \cdot \delta q\right).
\end{align*}
Keeping in mind that \[t\left(\frac{2g^{3/2}}{\Omega}g^+m\dot q\right) = \frac{3}{2}\cdot 2 - 3 -1 +1 = 0,\] the last term reads
\begin{align*}
    \mathcal{L}_\gamma \left(\frac{2g^{3/2}}{\Omega}g^+m\dot q \cdot \delta q\right)
    =&\,\, \xi \partial_t\left(\frac{2g^{3/2}}{\Omega}g^+m\dot q \right)\cdot \delta q
    + \frac{2g^{3/2}}{\Omega}g^+m\dot q \cdot \delta (\xi \dot q)\\
    =&\,\,\xi \partial_t\left(\frac{2g^{3/2}}{\Omega}g^+m\dot q \right)\cdot \delta q
    + \frac{4g^{3/2}}{\Omega}g^+ T \delta \xi\\
    &+ \xi \frac{2g^{3/2}}{\Omega}g^+ m\dot q \cdot \delta \dot q,\\
    \mathcal{L}_{\delta_{KT}} \left(\frac{2g^{3/2}}{\Omega}g^+m\dot q \cdot \delta q\right)
    =&\,\,\frac{2g^{3/2}}{\Omega}\text{EL}_gm\dot q \cdot \delta q
    =\left[\sqrt{\frac{E}{T}} - \frac{1}{\sqrt{g}}\right]m\dot q \cdot \delta q,
\end{align*}
which results in
\begin{align*}
    \mathcal{L}_Q \beta^1 
    =&\,\, \xi \dot \xi \beta^0
    - \xi \Delta \theta^0
    - \xi \partial_t \left(\xi \beta^0 +\cancel{\frac{2g^{3/2}}{\Omega}g^+m\dot q \cdot \delta q}\right)
    + \xi \delta f^0 \\
    &+\cancel{\xi \partial_t\left(\frac{2g^{3/2}}{\Omega}g^+m\dot q\right)\cdot \delta q}
    + \frac{4g^{3/2}}{\Omega}g^+ T \delta \xi\\
    &+ \cancel{\xi \frac{2g^{3/2}}{\Omega}g^+ m\dot q \cdot \delta \dot q} + \left[\sqrt{\frac{E}{T}} - \frac{1}{\sqrt{g}}\right]m\dot q \cdot \delta q\\
    =&\,\,\left[\sqrt{\frac{E}{T}} - \frac{1}{\sqrt{g}}\right]m\dot q \cdot \delta q
    - \xi \Delta \theta^0 
    + \cancel{\xi \dot \xi \beta^0} - \cancel{\xi \dot \xi \beta^0}\\
    &- \delta(\xi f^0) + \delta \xi 2 g^+\left(g - \cancel{\frac{2g^{3/2}}{\Omega} T}\right) + \cancel{\frac{4g^{3/2}}{\Omega} g^+ T\delta \xi} \\
    = &\,\, \Delta \theta^1 - \delta f^1.
\end{align*}
Showing $\Delta \theta^1 = \mathcal{L}_Q \beta^1 + \delta f^1$ as desired.

\textbf{Computation of $\Delta L^0$:} We want to show  \[\Delta L^0 = \mathcal{L}_Q \iota_Q \beta^0 - \mathrm{d}\iota_Q \beta^1 + \mathrm{d}f^1.\] 
Note that $\beta^0$ is of the form $\beta^0 = a_i \delta b_i$, where we sum over $i$. We have
\begin{align*}
    \mathcal{L}_Q \iota_Q(a_i \delta b_i) 
    &= \mathcal{L}_Q (a_i Q b_i)
    = Q a_i Qb_i\\
    &= \gamma a_i \gamma b_i 
    + \delta_{KT} a_i \gamma b_i
    + \gamma a_i \delta_{KT} b_i
    + \delta_{KT} a_i \delta_{KT} b_i.
\end{align*}
Starting with the term $A_1 \delta g^+$ (see Equation (\ref{eq:Asforbeta^1})) we have
\begin{align*}
    \gamma \left(-\frac{4g^{7/2}}{\Omega^2}Tg^+\right) \gamma g^+
    &= \left[\xi\partial_t\left(-\frac{4g^{7/2}}{\Omega^2}Tg^+\right) - 2 \dot \xi \frac{4g^{7/2}}{\Omega^2}Tg^+\right] \cdot \left[\xi \dot g^+ - \dot \xi g^+\right]\\
    &= \xi\frac{4g^{7/2}}{\Omega^2}T \dot g^+ \dot \xi g^+
    -2 \dot \xi  \frac{4g^{7/2}}{\Omega^2}Tg^+ \xi \dot g^+\\
    &= \frac{4g^{7/2}}{\Omega^2}T \dot g^+ g^+ \xi \dot \xi,\\
    \delta_{KT} \left(-\frac{4g^{7/2}}{\Omega^2}Tg^+\right) \gamma g^+ &= -\frac{4g^{7/2}}{\Omega^2}T \text{EL}_g (\xi \dot g^+ - \dot \xi g^+),\\
    \gamma \left(-\frac{4g^{7/2}}{\Omega^2}Tg^+\right) \delta_{KT} g^+ &= \left[\xi\partial_t\left(-\frac{4g^{7/2}}{\Omega^2}Tg^+\right) - 2 \dot \xi \frac{4g^{7/2}}{\Omega^2}Tg^+\right] \text{EL}_g,\\
    \delta_{KT} \left(-\frac{4g^{7/2}}{\Omega^2}Tg^+\right) \delta_{KT} g^+ &= -\frac{4g^{7/2}}{\Omega^2}T EL^0_g 
    = 2\sqrt{ET} - \frac{T}{\sqrt{g}} - \sqrt{g},
\end{align*}
where we used Equation (\ref{eq:usefulidentity2}). As such
\begin{align*}
    &\mathcal{L}_Q \iota_Q \left(-\frac{4g^{7/2}}{\Omega^2}Tg^+ \delta g^+\right)\\ 
    &= \frac{4g^{7/2}}{\Omega^2}T \dot g^+ g^+ \xi \dot \xi 
    - \partial_t \left(\xi \frac{4g^{7/2}}{\Omega^2}T  g^+\right) \text{EL}_g
    - \frac{4g^{7/2}}{\Omega^2}T \text{EL}_g \xi \dot g^+
    +2\sqrt{ET} - \frac{T}{\sqrt{g}} - \sqrt{g}\\
    & =\frac{4g^{7/2}}{\Omega^2}T \dot g^+ g^+ \xi \dot \xi 
    - \partial_t \left(\xi \frac{4g^{7/2}}{\Omega^2}T  g^+ \text{EL}_g\right)
    +\xi \frac{4g^{7/2}}{\Omega^2}T  g^+ \dot{\text{EL}_g}
    - \frac{4g^{7/2}}{\Omega^2}T \text{EL}_g \xi \dot g^+
    \\
    &\quad+2\sqrt{ET} - \frac{T}{\sqrt{g}} - \sqrt{g}.
\end{align*}
For the term $A_\perp \cdot \delta q$ we have
\begin{align*}
    \mathcal{L}_Q \iota_Q (A_\perp \delta q)
    &= QA_\perp Qq = \gamma A_\perp \cdot \xi \dot q + \delta_{KT}(A_\perp) \cdot \xi \dot q
    = \dot \xi A_\perp \cdot \xi \dot q + \delta_{KT}(A_\perp\cdot \xi \dot q)
    = 0.
\end{align*}
For the computation w.r.t. the term \[\mathcal{L}_Q\iota_Q \left[\left(\frac{4g^{7/2}}{\Omega^2}T - \eta^{3/2}\frac{g^{3/2}}{E}\right)\frac{m\dot q}{2T} \dot g^+g^+\cdot \delta q\right]\] 
first define \[B \coloneqq \left(\frac{4g^{7/2}}{\Omega^2}T - \eta^{3/2}\frac{g^{3/2}}{E}\right)\frac{m\dot q}{2T},\] 
with $t(B) = \frac{7}{2}\cdot 2- 2\cdot 3 + 2 + 1 - 2 = 2$ and note 
\begin{align*}
    t(B\dot g^+ g^+) =2 + 1 + t(g^+) + t(g^+) = 1. 
\end{align*}
As such
\begin{align*}
    &\mathcal{L}_Q\iota_Q \left[B\dot g^+ g^+ \cdot \delta q \right]
    = Q\left[B\dot g^+ g^+\right] \xi \dot q\\
    &= \partial_t\left[\xi B\dot g^+ g^+\right] \cdot \xi \dot q + B \left[\dot{\text{EL}_g}g^+ - \text{EL}_g \dot g^+ \right]\cdot \xi \dot q\\
    &= \dot \xi B\dot g^+ g^+ \cdot \xi \dot q
    + B \left[\dot{\text{EL}_g}g^+ - \text{EL}_g \dot g^+ \right]\cdot \xi \dot q\\
    &= \left(\frac{4g^{7/2}}{\Omega^2}T - \eta^{3/2}\frac{g^{3/2}}{E}\right) \dot g^+ g^+ \dot \xi \xi 
    + \left(\frac{4g^{7/2}}{\Omega^2}T - \eta^{3/2}\frac{g^{3/2}}{E}\right)\left[\dot{\text{EL}_g}g^+ - \text{EL}_g \dot g^+ \right] \xi.
\end{align*}
For the last term in $\beta^0$ 
\[-(\eta^{3/2}-1)\xi^+ \frac{m\dot q}{2T} \cdot \delta q,\]
we have
\begin{align*}
    \mathcal{L}_Q \iota_Q \left[-(\eta^{3/2}-1)\xi^+ \frac{m\dot q}{2T} \cdot \delta q\right] 
    = &-\gamma\left[(\eta^{3/2}-1)\xi^+ \frac{m\dot q}{2T}\right] \cdot \xi \dot q\\
    &- (\eta^{3/2}-1)\left[- q^+\cdot \dot q + g^+\dot g + 2\dot g^+ g\right] \frac{m\dot q}{2T} \cdot \xi \dot q\\
    =&(\eta^{3/2}-1)\xi^+ \xi \dot \xi - (\eta^{3/2}-1)\left[- q^+\cdot \dot q + g^+\dot g + 2\dot g^+ g\right] \xi.
\end{align*}
All in all the expression for $\mathcal{L}_Q \iota_Q \beta^0$ is
\begin{align}
    \label{eq:L_Qiota_Qbeta^1=DeltaL^0+totder}
    \mathcal{L}_Q \iota_Q \beta^0
    =&\,\,\cancel{\frac{4g^{7/2}}{\Omega^2}T \dot g^+ g^+ \xi \dot \xi}
    - \partial_t \left(\xi \frac{4g^{7/2}}{\Omega^2}T  g^+ \text{EL}_g\right)
    +\cancel{\xi \frac{4g^{7/2}}{\Omega^2}T  g^+ \dot{\text{EL}_g}}\nonumber\\
    &- \cancel{\frac{4g^{7/2}}{\Omega^2}T \text{EL}_g \xi \dot g^+}
    +2\sqrt{ET} - \frac{T}{\sqrt{g}} - \sqrt{g}\nonumber\\
    &+\left(\cancel{\frac{4g^{7/2}}{\Omega^2}T} - \eta^{3/2}\frac{g^{3/2}}{E}\right) \dot g^+ g^+ \dot \xi \xi 
    + \left(\cancel{\frac{4g^{7/2}}{\Omega^2}T} - \eta^{3/2}\frac{g^{3/2}}{E}\right) \dot{\text{EL}_g}g^+ \xi \nonumber\\
    &- \left(\cancel{\frac{4g^{7/2}}{\Omega^2}T} - \eta^{3/2}\frac{g^{3/2}}{E}\right) \text{EL}_g \dot g^+  \xi \nonumber\\
    &+(\eta^{3/2}-1)\xi^+ \xi \dot \xi - (\eta^{3/2}-1)\left[- q^+\cdot \dot q + g^+\dot g + 2\dot g^+ g\right] \xi \nonumber\\
    = &\,\,2\sqrt{ET} -\frac{T}{\sqrt{g}} - \sqrt{g}E\nonumber\\
    &+(\eta^{3/2} - 1) q^+ \cdot \xi \dot q 
    + (\eta^{3/2} - 1) \xi^+ \xi\dot \xi 
    - g^+ (\xi \dot g + 2\dot \xi g)\nonumber\\
    &- \eta^{3/2} \left[g^+\dot g + 2\dot g^+ g \right] \xi
    -\eta^{3/2}\frac{g^{3/2}}{E} \left[\dot{\text{EL}_g}g^+ - \text{EL}_g \dot g^+ \right]\xi\nonumber\\
    &- \partial_t \left(\xi \frac{4g^{7/2}}{\Omega^2}T  g^+ \text{EL}_g\right) + \partial_t\left(2g^+g\xi\right)\nonumber\\
    =&\,\, \Delta L^0 + \partial_t\left(2g^+g\xi + \frac{4g^{7/2}}{\Omega^2}T \text{EL}_g g^+ \xi \right),
\end{align}
where we used that 
\begin{align*}
    g^+ \dot g \xi + 2\dot g^+ g \xi &=g^+ \dot g \xi+ \partial_t\left(2g^+g\xi\right) - 2 g^+ \dot g \xi - 2 g^+ \dot g \xi \\
    &= \partial_t\left(2g^+g\xi\right) - 2 g^+ \dot g \xi - g^+ \dot g \xi.
\end{align*}
Recall that we are want to show $\Delta L^0 \,\d t = \mathcal{L}_Q \iota_Q \beta^0 \,\d t- \mathrm{d}\iota_Q \beta^1 + \mathrm{d} f^1$. The last two terms read
\begin{align*}
    \mathrm{d}\iota_Q \beta^1 - \mathrm{d} f^1 
    = \partial_t\left[\xi \iota_Q \beta^0 + \frac{2g^{3/2}}{\Omega}g^+m\dot q \cdot \iota_Q \delta q - 2 \xi g^+\left(g - \frac{2g^{3/2}}{\Omega} T\right) \right] \mathrm{d}t.
\end{align*}
which can be simplified by noting that
\begin{align}
    \label{eq:xiiota_Qbeta^1}
    \xi \iota_Q \beta^0 = &-\xi \frac{4g^{7/2}}{\Omega^2}Tg^+ Qg^+
    +\xi\left(\frac{2g^2}{\Omega} + \eta^{3/2}\frac{2\sqrt{g}}{E}\right)g^+q^+_\perp \cdot Q q,  \nonumber\\
    &+\xi \left(\frac{4g^{7/2}}{\Omega^2}T - \eta^{3/2}\frac{g^{3/2}}{E}\right)\dot g^+ g^+ \frac{m\dot q}{2T} \cdot Q q 
    -\xi (\eta^{3/2}-1)\xi^+ \frac{m\dot q}{2T} \cdot Q q\nonumber \\
    =&\,\,\frac{4g^{7/2}}{\Omega^2}T \text{EL}_g g^+ \xi,
\end{align}
and
\begin{align}
    \label{eq:estousemideias}
    \frac{2g^{3/2}}{\Omega}g^+m\dot q \cdot \iota_Q \delta q = \frac{4g^{3/2}}{\Omega}T g^+ \xi.
\end{align}
The resulting expression for $\mathrm{d}\iota_Q \beta^1 - \mathrm{d} f^1$ is then
\begin{align*}
    \mathrm{d}\iota_Q \beta^1 - \mathrm{d} f^1 
    &= \partial_t\left[\frac{4g^{7/2}}{\Omega^2}T \text{EL}_g g^+ \xi + \frac{4g^{3/2}}{\Omega}T g^+ \xi - 2 \xi g^+\left(g - \frac{2g^{3/2}}{\Omega} T\right) \right] \mathrm{d}t\\
    &= \partial_t\left(2g^+g\xi + \frac{4g^{7/2}}{\Omega^2}T \text{EL}_g g^+ \xi \right) \mathrm{d}t,
\end{align*}
which is exactly the total derivative in Equation (\ref{eq:L_Qiota_Qbeta^1=DeltaL^0+totder}). Hence
\begin{align*}
    &\mathcal{L}_Q \iota_Q \beta^0 \mathrm{d}t 
    = \Delta L^0 \mathrm{d}t +\mathrm{d}\iota_Q \beta^1 - \mathrm{d} f^1 \\
    \Rightarrow &\Delta L^0 \mathrm{d}t 
    = \mathcal{L}_Q \iota_Q \beta^0 \mathrm{d}t
    -\mathrm{d}\iota_Q \beta^1 + \mathrm{d} f^1.
\end{align*}
finishing the proof.
\end{proof}

\begin{lemma}
\label{lem:lambdaJacobi1DGR}
The composition map $\lambda^*$  is the identity
\begin{align*}
    &\lambda^*\qj = \qj,&
    &\lambda^*\xij = \xij,\\
    &\lambda^*\qj^+ = \qj^+,&
    &\lambda^*\xij^+ = \xij^+.
\end{align*}
and as such the identity in cohomology.
\end{lemma}

\begin{proof}
On the matter and ghost fields $\{\qj, \xij\}$ this is trivial, since at this level both $\psi^*$ and $\phi^*$ simply interchange $\qj$ with $q$ and $\xij$ with $\xi$
\begin{align*}
    &\lambda^*\qj = (\phi^* \circ \psi^*) \qj
    = \phi^* q = \qj,\\
    &\lambda^*\xij = (\phi^* \circ \psi^*) \xij
    = \phi^* \xi = \xij.
\end{align*}

In order to compute the action of $\lambda^*$ on the antifield and antighost, first recall that $\phi^* g^+ = 0$ and note that $\phi^*g = \Tj / E$ implies $\phi^*(\eta^{3/2}) = (\phi^*(g)E/\Tj)^{3/2} = 1$. We then have
\begin{align*}
    \lambda^* \qj^+_\parallel &= (\phi^* \circ \psi^*) \qj^+_\parallel\\
    &= \phi^*(\eta^{3/2}) \phi^*\left(q^+_\parallel 
    - \left[g^+\dot g + 2\dot g^+ g \right] \frac{m\dot q}{2T}
    -\frac{g^{3/2}}{E} \left[\dot{\text{EL}_g}g^+ - \text{EL}_g \dot g^+
    \right]\frac{m\dot q}{2T}\right)\\
    &= \phi^* q^+_\parallel = \qj^+_\parallel,\\
    \lambda^* \qj^+_\perp &= (\phi^* \circ \psi^*) \qj^+_\perp
    = \phi^*(\eta^{3/2}) \phi^*\left(q^+_\perp 
    +\frac{2m}{E} g^+ \ddot q_\perp \right) 
    = \phi^*q^+_\perp = \qj^+_\perp,\\
    \lambda^* \xij^+ &= (\phi^* \circ \psi^*) \xij^+
    =  \phi^*(\eta^{3/2})\phi^*\left(\xi^+ + \frac{g^{3/2}}{E} \dot g^+ g^+\right) = \phi^* \xi^+ = \xij^+,
\end{align*}
thus showing $\lambda = \mathrm{id}_J$.
\end{proof}


We now prove that $\Rgr$ commutes with the Chevalley-Eilenberg differential $\gamma_{GR}$ if $\Rgr \xi = 0$. Effectively, this means that we can ignore the Chevalley-Eilenberg part of $\Qgr$ in $\Dgr = [\Qgr,\Rgr]$ and only have to regard the Koszul-Tate differential when explicitly computing the action of $\Dgr $. Recall that the Chevalley-Eilenberg differential acts as $\gamma_{GR} = \mathcal{L}_{\xi\partial_t}$ on $\{q,g,q^+,g^+,\xi^+\}$ and as $\gamma_{GR} = \frac{1}{2}\mathcal{L}_{\xi\partial_t}$ on $\{\xi\}$.
\begin{lemma}
\label{lem:[R,gamma]}
Let $\Rgr\in \mathfrak{X}_{\mathrm{evo}}(\F_{GR})$ be an evolutionary vector field on $\F_{GR}$ with the following properties
\begin{itemize}
    \item $\Rgr$ vanishes on $\mathfrak{X}[1](I)$,
    \item $\Rgr$ preserves the tensor rank on $I$,
\end{itemize}
and let $\gamma_{GR}$ be the Chevalley-Eilenberg differential of the 1D GR theory. Then $[\Rgr,\gamma_{GR}] = 0$.
\end{lemma}
\begin{proof}
Recall that all the fields we are considering are components of tensor fields over $I$. In particular, note that the property that $\Rgr$ vanishes on $\mathfrak{X}[1](I)$ implies $\Rgr\xi = 0$, since $\xi \partial_t \in \mathfrak{X}[1](I)$. As $\gamma_{GR} \sim \mathcal{L}_{\xi\partial_t}$, it suffices to show that $[\Rgr ,\mathcal{L}_{\xi\partial_t}] = 0$ on functions, 1-forms and vector fields over $I$. We assume that all objects have an internal grading throughout the proof in order to account for the ghost number.

We will first show $[\Rgr ,\gamma_{GR}]=0$ for functions and 1-forms. Since $\mathcal{L}_{\xi \partial_t} = [\iota_{\xi \partial_t}, \mathrm{d}]$ on $\Omega^\bullet(I)$ we have
\begin{align*}
    [\Rgr ,\mathcal{L}_{\xi \partial_t}] = [\Rgr ,[\iota_{\xi \partial_t}, \mathrm{d}]] = [\mathrm{d}, [\Rgr ,\iota_{\xi \partial_t}]] + [\iota_{\xi \partial_t}, [\mathrm{d},\Rgr ]]
    = [\mathrm{d}, [\Rgr ,\iota_{\xi \partial_t}]],
\end{align*}
where we used that $\Rgr $ is evolutionary. As such it is sufficient to show that $\Omega^\bullet(I) \subset \ker [\Rgr ,\iota_{\xi \partial_t}]$. By definition all functions on $I$ are in the kernel of $\iota_{\xi \partial t}$: $C^\infty(I) = \Omega^0(I) \subset \ker \iota_{\xi \partial_t}$. Let $f\in \Omega^0(I)$. Since we assume that $\Rgr $ preserves the tensor rank we also have $\Rgr  f\in \Omega^0(I)$, then
\begin{align*}
    &\Rgr  \,\iota_{\xi \partial_t} f = 0,&
    &\iota_{\xi \partial_t} \Rgr  f = 0,
\end{align*}
and as such $[\Rgr ,\iota_{\xi \partial_t}]f = 0$. Let now $\varpi = f \mathrm{d}t \in \Omega^1(I)$ be a 1-form. Taking into account that $|\mathrm{d}t| = 1$, we have
\begin{align*}
    & \Rgr  \,\iota_{\xi \partial_t} \varpi 
    = \Rgr  \,\iota_{\xi \partial_t} (f \mathrm{d}t) 
    = \Rgr (f \mathrm{d}t (\xi \partial_t)) 
    = -\Rgr (f\xi) = - \Rgr f \xi,\\
    & \iota_{\xi \partial_t} \Rgr  \varpi = \iota_{\xi \partial_t} \Rgr (f \mathrm{d}t)
    = \iota_{\xi \partial_t} [\Rgr f \mathrm{d}t]
    = \Rgr f \mathrm{d}t (\xi \partial_t)
    = - \Rgr f \xi,\\
    \Rightarrow &[\Rgr,\iota_{\xi \partial_t}] \varpi = \Rgr\iota_{\xi \partial_t}\varpi - \iota_{\xi \partial_t} \Rgr\varpi = -\Rgr f \xi + \Rgr f \xi = 0,
\end{align*}
thus showing that $\Omega^0(I) \times \Omega^1(I) \subset \ker [\Rgr ,\iota_{\xi \partial_t}]$. This implies
\begin{align*}
    &[\mathrm{d},[\Rgr ,\iota_{\xi\partial_t}]] f = 0,&
    &[\mathrm{d} ,[\Rgr ,\iota_{\xi\partial_t}]] \varpi = \mathrm{d} [\Rgr ,\iota_{\xi\partial_t}] \varpi = 0,
\end{align*}
where we used $f \in \Omega^0(I) \subset \ker [\Rgr ,\iota_{\xi \partial_t}]$, $\d f \in \Omega^1(I) \subset \ker [\Rgr ,\iota_{\xi \partial_t}]$ and $\mathrm{d}\varpi = 0$, since $\Omega^1(I)=\Omega^\text{top}(I)$. 

Consider now a vector field $X = f\partial_t \in \mathfrak{X}(I)$ of degree $n$. In this case we have:
\begin{align*}
    &\Rgr  \mathcal{L}_{\xi \partial_t} X 
    =\Rgr [\xi \partial_t,f \partial_t]  
    =\Rgr (\xi \dot f - (-1)^nf\dot \xi) \partial_t
    =-(\xi \Rgr  \dot f + (-1)^n \Rgr f \dot \xi) \partial_t,\\
    &\mathcal{L}_{\xi \partial_t} \Rgr X 
    = [\xi \partial_t,\Rgr f \partial_t]
    = (\xi \Rgr \dot f \partial_t - (-1)^{n-1} \Rgr f \dot \xi) \partial_t,\\
    \Rightarrow &[\Rgr ,\mathcal{L}_{\xi \partial_t}]X = \Rgr \mathcal{L}_{\xi \partial_t}X + \mathcal{L}_{\xi \partial_t}\Rgr X = 0,
\end{align*}
where we used $\partial_t(\Rgr f) = \Rgr  \dot f$. Since $[\Rgr ,\mathcal{L}_{\xi \partial_t}] = 0$ on functions, 1-forms and vector fields it holds for all tensors. As such we have $[\Rgr ,\gamma_{GR}] = 0$ on $\F_{GR}$.
\end{proof}

A direct implication of Lemma \ref{lem:[R,gamma]} is that the vector field $\Dgr $ reduces to $$\Dgr  = [\Qgr ,\Rgr ] = [\delta_{GR},\Rgr ].$$ 

\begin{lemma}
\label{lem:chiJacobi1DGR}
The composition map $\chi^*$ acts as
\begin{subequations}
\label{eq:chiJacobi1DGR}
\begin{align}
\begin{split}
    &\chi^*q = q  ,\\
    &\chi^*\xi = \xi,\\
    &\chi^*g = \frac{T}{E},\\
    &\chi^* q^+_\parallel = \eta^{3/2} \left(q^+_\parallel 
        - \left[g^+\dot g + 2\dot g^+ g \right] \frac{m\dot q}{2T}
        -\frac{g^{3/2}}{E} \left[\dot{\text{EL}_g}g^+ - \text{EL}_g \dot g^+ \right]\frac{m\dot q}{2T}\right), \\
    &\chi^* q^+_\perp = \eta^{3/2} \left(q^+_\perp 
    +\frac{2m}{E} g^+ \ddot q_\perp \right),\\
    &\chi^* \xi^+ = \eta^{3/2}\left(\xi^+ + \frac{g^{3/2}}{E} \dot g^+ g^+\right),\\
    &\chi^* g^+ = 0.
\end{split}
\end{align}\end{subequations}
and is homotopic to the identity.
\end{lemma}

\begin{proof}
Recall that $\chi^* = \psi^* \circ \phi^*$ and let $\varphi_i \in \{q,q^+,\xi,\xi^+\}$. In this case we have $\phi^*\varphi_i = \tilde \varphi_i$ and as such $\chi^* \varphi_i = \psi^* \tilde \varphi_i$, which reproduces the expressions above due to the explicit form of $\psi^*$. For $\{g,g^+\}$ we compute
\begin{align*}
    \chi^* g &= (\psi^* \circ \phi^*) g = \psi^*\left(\frac{\Tj}{E}\right) = \frac{T}{E},\\
    \chi^* g^+ &= (\psi^* \circ \phi^*) g^+
    = \psi^*(0) = 0.
\end{align*}
In order to show that $\chi^*$ is homotopic to the identity we choose the vector field $\Rgr$ to act as 
\begin{align}
\label{eq:Roperator}
    &\Rgr q = 0,&
    &\Rgr \xi = 0,&
    &\Rgr  g= \frac{-2g^{3/2}}{E} g^+,&\nonumber\\
    &\Rgr  q^+_\parallel = -\frac{3\sqrt{g}}{E} \text{EL}_g \xi^+ \frac{m\dot q}{2T},&
    &\Rgr  \xi^+ =0,&
    &\Rgr  g^+ = 0,\nonumber\\
    &\Rgr  q^+_\perp = \frac{3\sqrt{g}}{E} g^+ q^+_\perp,
\end{align}
Since $\Rgr  \xi = 0$ we can use Lemma \ref{lem:[R,gamma]}. We start with the computation for $q$ and $\xi$. Recalling that they are both in the kernel of $\delta_{GR}$ (cf.\ Remark \ref{rem:Chevalley-Eilenberg1DGR}) we have
\begin{align*}
    && \Dgr  q 
    &= (\delta_{GR} \Rgr  + \Rgr  \delta_{GR}) q 
    = R(\xi \dot q) = 0,&&\\
    &\Rightarrow&
    e^{s\mathcal{L}_{\Dgr}} q &= q,&&\\
    &\Rightarrow& 
    \lim_{s\rightarrow \infty} e^{s\mathcal{L}_{\Dgr}} q &= q = \chi^* q, &&\\
    && \Dgr \xi 
    &= (\delta_{GR} \Rgr + \Rgr \delta_{GR}) \xi
    = 0,&&\\
    &\Rightarrow&
    e^{s\mathcal{L}_{\Dgr}} \xi &= \xi,\\
    &\Rightarrow& 
    \lim_{s\rightarrow \infty} e^{s\mathcal{L}_{\Dgr}} \xi &= \xi = \chi^* \xi. &&
\end{align*}

Keeping in mind that $\delta_{GR}g = 0$ and 
\begin{align*}
    \delta_{GR} g^+ = \text{EL}_g 
    = \frac{E}{2\sqrt{g}}-\frac{T}{2g^{3/2}},
\end{align*}
(cf.\ Remark \ref{rem:Chevalley-Eilenberg1DGR}), we compute $\Dgr g$ to be
\begin{align}
    \Dgr  g 
    &= \left(\delta_{GR} \Rgr + \Rgr \delta_{GR}\right) g
    = \delta_{GR} \left(\frac{-2g^{3/2}}{E}g^+\right)\nonumber\\
    &= \frac{-2g^{3/2}}{E} \left(\frac{E}{2\sqrt{g}}-\frac{T}{2g^{3/2}} \right)
    = -\left(g - \frac{T}{E}\right).
    \label{eq:Dgrg}
\end{align}
We can then show 
\begin{align}
    \label{eq:Dkg}
    \Dgrk  g = (-1)^k \left(g-\frac{T}{E}\right) \quad \text{for } k \geq 1,
\end{align}
using induction. As we have computed, this holds for $k=1$. Assuming that Equation \eqref{eq:Dkg} holds for an arbitrary $k$, we then have the following for $k+1$
\begin{align*}
    \Dgr ^{k+1} g 
    = (-1)^k \left(\Dgr g - \Dgr \frac{T}{E}\right)
    = (-1)^{k+1}\left(g-\frac{T}{E}\right),
\end{align*}
where we used $\Dgr  T = m\dot q \Dgr \dot q = m\dot q \partial_t(\Dgr q) = 0$. This results in
\begin{align*}
    && e^{s\mathcal{L}_{\Dgr}} g 
    &= g + \sum_{k \geq 1} \frac{s^k}{k!} (-1)^k \left(g-\frac{T}{E}\right) = g + (e^{-s} - 1) \left(g-\frac{T}{E}\right)&&\\
    &\Rightarrow& 
    \lim_{s\rightarrow \infty} e^{s\mathcal{L}_{\Dgr}} g &= \frac{T}{E} = \chi^* g. &&
\end{align*}

We now compute $\chi^*_s\xi^+$ and its $s\rightarrow \infty$ limit, which follows the same strategy as the analogous computations for $g^+$ presented in the main proof. We start with $\Dgr \xi^+$
\begin{align*}
    \Dgr \xi^+ 
    &= \Rgr \delta_{GR} \xi^+
    = \Rgr \left(-q^+\cdot \dot q +  g^+ \dot g + 2 \dot g^+ g \right)\\
    &= -\Rgr (q^+\cdot \dot q) - g^+\partial_t\left(\frac{-2g^{3/2}}{E}g^+\right)  
    -\dot g^+ 2\frac{(-2)g^{3/2}}{E}g^+\\
    &= -\Rgr (q^+\cdot \dot q) + \frac{2g^{3/2}}{E} g^+ \dot g^+  +\frac{4g^{3/2}}{E} \dot g^+g^+ \\
    &= -\Rgr (q^+\cdot \dot q) + \frac{2g^{3/2}}{E} \dot g^+ g^+,
\end{align*}
where we used that $g^+ g^+ = 0$ in the second line and $g^+ \dot g^+ = -\dot g^+ g^+$ in the third. With this in hand we can proceed with the calculation of $\Dgr (g^{3/2}\xi^+)$
\begin{align*}
    &\Dgr (g^{3/2} \xi^+) 
    = \frac{3}{2} \sqrt{g} \Dgr g \xi^+ + g^{3/2} \Dgr \xi^+\\
    &= - \frac{3g^2}{E} \text{EL}_g \xi^+ -g^{3/2} \Rgr (q^+\cdot \dot q) + \frac{2g^{3}}{E} \dot g^+ g^+
    = \frac{2}{E} \partial_t(g^{3/2}g^+)g^{3/2}g^+.
\end{align*}
It should now be clear why we chose the specific form for $\Rgr q^+\cdot \dot q$: the first two terms in the second line cancel and we are left with a term for which we know how to compute $\Dgrk $. Using induction we can then prove
\begin{align*}
    \Dgrk  (g^{3/2} \xi^+) = -\frac{(-2)^k}{E} \partial_t(g^{3/2}g^+)g^{3/2}g^+ \quad \text{for } k\geq 1.
\end{align*}
We have already shown that it holds for $k=1$. Assuming that it is true for $k$, the expression for $k+1$ reads:
\begin{align*}
    \Dgr ^{k+1} (g^{3/2} \xi^+) 
    &= -\frac{(-2)^k}{E} \partial_t(\Dgr g^{3/2}g^+)g^{3/2}g^+
    -\frac{(-2)^k}{E} \partial_t(g^{3/2}g^+)\Dgr g^{3/2}g^+\\
    &= 2 \frac{(-2)^k}{E} \partial_t(g^{3/2}g^+)g^{3/2}g^+
    = - \frac{(-2)^{k+1}}{E} \partial_t(g^{3/2}g^+)g^{3/2}g^+.
\end{align*}
Since $\partial_t(g^{3/2}g^+)g^{3/2}g^+ = g^3 \dot g^+ g^+$ due to $g^+g^+ = 0$, we then have
\begin{align*}
    &\chi^*_s(g^{3/2} \xi^+) = e^{s\mathcal{L}_{\Dgr}}(g^{3/2} \xi^+) 
    = \sum_{k\geq0} \frac{s^k}{k!} \Dgrk (g^{3/2} \xi^+)\\
    &= g^{3/2} \xi^+ - \left(\sum_{k\geq1} \frac{(-2s)^k}{k!}\right) \frac{g^3}{E} \dot g^+ g^+
    = g^{3/2} \xi^+ -(e^{-2s} - 1)\frac{g^{3}}{E}\dot g^+ g^+,
\end{align*}
and
\begin{align*}
    \lim_{s\rightarrow\infty} \chi^*_s \xi^+
    &= \left(\frac{E}{T}\right)^{3/2}
    \lim_{s\rightarrow\infty} \chi^*_s (g^{3/2}\xi^+)\\
    &= \left(\frac{E}{T}\right)^{3/2}
    \lim_{s\rightarrow\infty} \left(g^{3/2} \xi^+ -(e^{-2s} - 1)\frac{g^{3}}{E}\dot g^+ g^+\right)\\
    &= \eta^{3/2}\left(\xi^+ +\frac{g^{3/2}}{E}\dot g^+ g^+\right)
    = \chi^*\xi^+.\\
\end{align*}

The strategy in the case of $q^+$ is the same as for $g^+$ and $\xi^+$. Due to $$\Dgr\left(\frac{m \dot q}{2T}\right) = 0,$$ we have
\begin{align*}
     \Dgr(g^{3/2} q^+_\parallel) =  \Dgr\left(g^{3/2} q^+\cdot \frac{m \dot q}{2T}\right) = \Dgr(g^{3/2}q^+ \cdot \dot q) \frac{m \dot q}{2T}.
\end{align*}
We will therefore omit the term $m \dot q/(2T)$ from the computations in order to keep them cleaner. We start by calculating
\begin{align*}
    \Dgr(q^+ \cdot \dot q) 
    &= \left(\delta_{GR} \Rgr  + \Rgr  \delta_{GR}\right) q^+ \cdot \dot q\\
    &= - \delta_{GR} \left(\frac{3\sqrt{g}}{E} \text{EL}_g \xi^+\right)
    - \Rgr \left(\partial_t\left(\frac{m\dot q}{\sqrt{g}}\right)\right) \cdot \dot q\\
    &= \frac{3\sqrt{g}}{E} \text{EL}_g (q^+ \cdot \dot q - g^+ \dot g - 2\dot g^+ g) - \partial_t\left(\frac{m\dot q}{2g^{3/2}} \frac{2g^{3/2}}{E}g^+\right) \cdot \dot q\\
    &=\frac{3}{2}\left(1 - \frac{T}{Eg}\right) q^+ \cdot \dot q
    - \frac{3\sqrt{g}}{E} \text{EL}_g(g^+ \dot g + 2 \dot g^+ g) 
    - \frac{m \ddot q \cdot q}{E}g^+
    - \frac{m\|\dot q\|^2}{E} \dot g^+\\
    &=\frac{3}{2}\left(1 - \frac{T}{Eg}\right) q^+ \cdot \dot q
    - \frac{3\sqrt{g}}{E} \text{EL}_g(g^+ \dot g + 2 \dot g^+ g) 
    - \frac{\dot T}{E}g^+
    - \frac{2T}{E} \dot g^+.
\end{align*}
Let $\sigma(\varphi) = \varphi \partial_t{(g^{3/2}g^+)} - \dot \varphi g^{3/2}g^+$. With the result for $\Dgr(q^+ \cdot \dot q)$ we compute the following
\begin{align*}
    \Dgr(g^{3/2} q^+ \cdot q) 
    &= \frac{3\sqrt{g}}{2} \Dgr g q^+ \cdot \dot q 
    + g^{3/2} \Dgr (q^+\cdot \dot q)\\
    &= \frac{3\sqrt{g}}{2} \left(\frac{T}{E}-g\right)q^+ \cdot \dot q
    +g^{3/2} \Dgr (q^+\cdot \dot q)\\
    &= - \frac{3g^2}{E} \text{EL}_g(g^+ \dot g + 2 \dot g^+ g) 
    - g^{3/2}\frac{\dot T}{E}g^+
    - g^{3/2}\frac{2T}{E} \dot g^+\\
    &= \left(\sqrt{g}\frac{3T}{2E} - \frac{3}{2}g^{3/2} \right)(g^+ \dot g + 2 \dot g^+ g) 
    - g^{3/2}\frac{\dot T}{E}g^+
    - g^{3/2}\frac{2T}{E} \dot g^+\\
    &= \frac{T}{E}\partial_t(g^{3/2}) g^+ 
    + \frac{3T}{E} g^{3/2}\dot g^+
    - \frac{3}{2} g^{3/2} \dot g g^+ - 3 g g^{3/2}\dot g^+
    - g^{3/2}\frac{\dot T}{E}g^+
    - g^{3/2}\frac{2T}{E} \dot g^+\\
    &= \frac{T}{E} \partial_t(g^{3/2} g^+) 
    - \frac{\dot T}{E}g^{3/2} g^+ 
    - 3\left[g \partial_t(g^{3/2} g^+) 
    - \dot g (g^{3/2} g^+)\right]\\
    &= \sigma\left(\frac{T}{E}\right) - 3\sigma(g)
    = -2\sigma\left(\frac{T}{E}\right) -3 \sigma\left(g-\frac{T}{E}\right)\\
    &= -2\sigma\left(\frac{T}{E}\right)
    - 2\frac{3}{E}\sigma(g^{3/2}\text{EL}_g).
\end{align*}
Reintroducing $m\dot q/(2T)$ gives
\begin{align}
    \label{eq:Dg32q+parallel}
    \Dgr (g^{3/2} q^+_\parallel) 
    &= -2\sigma\left(\frac{T}{E}\right) \frac{m\dot q}{2T}
    - 2\cdot \frac{3}{E}\sigma(g^{3/2}\text{EL}_g)\frac{m\dot q}{2T}.
\end{align}
Using induction it is then possible to show that
\begin{align}
    \label{eq:Dkq+parallel}
    \Dgrk (g^{3/2} q^+_\parallel) = (-1)^k 2 \sigma\left(\frac{T}{E}\right) \frac{m\dot q}{2T} + (-2)^k \frac{3}{E} \sigma(g^{3/2}\text{EL}_g) \frac{m\dot q}{2T},
\end{align}
for $k\geq 1$. The case $k=1$ is presented in Equation (\ref{eq:Dg32q+parallel}). To see how the case $k+1$ follows from the case $k$, note that
\begin{align*}
    \Dgr \sigma\left(\frac{T}{E}\right) 
    = \Dgr \left(\frac{T}{E} \partial_t(g^{3/2} g^+) 
    - \frac{\dot T}{E}g^{3/2} g^+ \right)
    = - \sigma\left(\frac{T}{E}\right),
\end{align*}
where we used $\Dgr T = 0$ and $\Dgr (g^{3/2}g^+) = - g^{3/2}g^+$. Before computing $\Dgr \sigma(g^{3/2}\text{EL}_g)$ note that
\begin{align*}
    &g^{3/2}\text{EL}_g = \frac{E}{2}\left(g-\frac{T}{E}\right) = -\frac{E}{2}\Dgr g\\
    \Rightarrow \quad &\Dgr (g^{3/2}\text{EL}_g) = - \frac{E}{2} \Dgr^2g = \frac{E}{2} \Dgr g = - g^{3/2}\text{EL}_g.
\end{align*}
The action of $\Dgr $ on $\sigma(g^{3/2}\text{EL}_g)$ is then
\begin{align*}
    &\Dgr (\sigma(g^{3/2}\text{EL}_g)) =  \Dgr \left(g^{3/2}\text{EL}_g \partial_t(g^{3/2}g^+) - \partial_t(g^{3/2}\text{EL}_g) g^{3/2}g^+\right)\\
    &= -2\left(g^{3/2}\text{EL}_g \partial_t(g^{3/2}g^+) - \partial_t(g^{3/2}\text{EL}_g) g^{3/2}g^+\right)
    = -2\sigma(g^{3/2}\text{EL}_g),
\end{align*}
which proves the Equation (\ref{eq:Dkq+parallel}). Having $\Dgrk (g^{3/2} q^+_\parallel)$ we can now write
\begin{align*}
    \chi^*_s q^+_\parallel &= e^{s\mathcal{L}_{\Dgr}} (g^{3/2} q^+_\parallel) 
    = \sum_{k\geq0} \frac{s^k}{k!}\Dgrk (g^{3/2} q^+_\parallel)\\
    &= g^{3/2} q^+_\parallel 
    + \left(\sum_{k\geq1} \frac{(-s)^k}{k!}\right) 2 \sigma\left(\frac{T}{E}\right) \frac{m\dot q}{2T} 
    + \left(\sum_{k\geq1} \frac{(-2s)^k}{k!}\right) \frac{3}{E} \sigma(g^{3/2}\text{EL}_g) \frac{m\dot q}{2T}\\
    &= g^{3/2} q^+_\parallel
    + (e^{-s} - 1) 2 \sigma\left(\frac{T}{E}\right)\frac{m\dot q}{2T}
    + (e^{-s} - 1) \frac{3}{E} \sigma(g^{3/2}\text{EL}_g)\frac{m\dot q}{2T},
\end{align*}
taking the $s\rightarrow\infty$ limit then yields
\begin{align*}
    \lim_{s\rightarrow\infty}\chi^*_s (g^{3/2} q^+_\parallel)
    &= g^{3/2} q^+_\parallel
    - 2\sigma\left(\frac{T}{E}\right)\frac{m\dot q}{2T}
    - \frac{3}{E} \sigma(g^{3/2}\text{EL}_g)\frac{m\dot q}{2T}
\end{align*}
as desired. We can then extract $\lim_{s\rightarrow\infty}\chi^*_sq^+_\parallel$ from this expression using 
\begin{align*}
    \lim_{s\rightarrow\infty}\chi^*_sq^+_\parallel = (E/T)^{3/2}\lim_{s\rightarrow\infty}\chi^*_s(g^{3/2} q^+_\parallel),
\end{align*}
see Equation (\ref{eq:limchi^*_sPhi^+}). We have
\begin{align}
\begin{split}
    \label{eq:chiq^+complicated}
    \lim_{s\rightarrow\infty}\chi^*_sq^+_\parallel 
    &= \eta^{3/2} \left(q^+_\parallel
    - 2g^{-3/2}\sigma\left(\frac{T}{E}\right)\frac{m\dot q}{2T}
    - \frac{3}{E}g^{-3/2}\sigma(g^{3/2}\text{EL}_g)\frac{m\dot q}{2T}\right)\\
    &= \eta^{3/2} \left(q^+\cdot \dot q
    - 2g^{-3/2}\sigma\left(\frac{T}{E}\right)
    - \frac{3}{E}g^{-3/2}\sigma(g^{3/2}\text{EL}_g)\right)\frac{m\dot q}{2T}.
\end{split}
\end{align}
This expression can be further simplified, but first note that
\begin{align*}
    &\dot{\text{EL}_g} = -\frac{E}{4g^{3/2}} \dot g + \frac{3T}{4g^{5/2}} \dot g - \frac{\dot T}{2g^{3/2}}\\
    \Rightarrow &\frac{3T}{Eg} \dot g 
    - \frac{2\dot T}{E} = \frac{4g^{3/2}}{E} \dot{\text{EL}_g} + \dot g,
\end{align*}
and
\begin{align*}
    &\text{EL}_g = \frac{E}{2\sqrt{g}} - \frac{T}{2g^{3/2}}\\
    \Rightarrow &\frac{2T}{E} = 2g - \frac{4g^{3/2}}{E} \text{EL}_g.
\end{align*}
Using these two identities the last two terms in Equation (\ref{eq:chiq^+complicated}) yield
\begin{align*}
    &2g^{-3/2}\sigma\left(\frac{T}{E}\right)
    + \frac{3}{E}g^{-3/2}\sigma(g^{3/2}\text{EL}_g)\\
    &= \frac{2T}{E} g^{-3/2} \partial_t(g^{3/2}g^+)
    - \frac{2 \dot T}{E} g^+ 
    + \frac{3}{E} \text{EL}_g \partial_t(g^{3/2}g^+)
    - \frac{3}{E}\partial_t(g^{3/2}\text{EL}_g) g^+\\
    &= g^+\left[\frac{3T}{Eg} \dot g 
    - \frac{2\dot T}{E} 
    + \frac{3}{E} \text{EL}_g \partial_t(g^{3/2})
    - \frac{3}{E} \partial_t(g^{3/2}\text{EL}_g) \right]
    + \dot g^+\left[\frac{2T}{E} 
    + \frac{3}{E} g^{3/2} \text{EL}_g \right],\\
    &=g^+\left[\dot g + \frac{g^{3/2}}{E} \dot{\text{EL}_g}\right]
    + \dot g^+\left[2g - \frac{g^{3/2}}{E} \text{EL}_g\right]\\
    &=\left[g^+\dot g + 2\dot g^+ g \right]
    +\frac{g^{3/2}}{E} \left[\dot{\text{EL}_g}g^+ - \text{EL}_g \dot g^+ \right],
\end{align*}
and as such
\begin{align*}
    \lim_{s\rightarrow\infty}\chi^*_s q^+_\parallel &= \eta^{3/2} \left(q^+_\parallel 
        - \left[g^+\dot g + 2\dot g^+ g \right] \frac{m\dot q}{2T}
        -\frac{g^{3/2}}{E} \left[\dot{\text{EL}_g}g^+ - \text{EL}_g \dot g^+ \right]\frac{m\dot q}{2T}\right)\\
    &= \chi^*q^+.   
\end{align*}

We now move to the computation of $\chi^*_sq^+_\perp$. As before our strategy will be to compute $\chi^*_sq^+_\perp$ via $\Dgrk (g^{3/2} q^+_\perp)$. First note that we have
\begin{align*}
    \delta_{GR} q^+_\perp 
    &= -\partial_t\left( \frac{m\dot q}{\sqrt{g}}\right) 
    + \frac{m \dot q}{2T} \partial_t\left( \frac{m\dot q}{\sqrt{g}}\right) \cdot \dot q\\
    &=-\frac{m\ddot q}{\sqrt{g}} 
    - \cancel{\partial_t\left( \frac{1}{\sqrt{g}}\right) m \dot q}
    + \frac{m \dot q}{2T} \frac{m \ddot q \cdot \dot q}{\sqrt{g}}
    + \cancel{\partial_t\left( \frac{1}{\sqrt{g}}\right) m \dot q}
    = -\frac{m \ddot q_\perp}{\sqrt{g}}.
\end{align*}
It follows that
\begin{align*}
    \Dgr q^+_\perp &= (\delta_{GR} \Rgr  + \Rgr  \delta_{GR}) q^+_\perp
    = \delta_{GR}\Rgr q^+_\perp - \Rgr \left(\frac{m \ddot q_\perp}{\sqrt{g}}\right)\\
    &= \delta_{GR}\Rgr q^+_\perp + \frac{m \ddot q_\perp}{2g^{3/2}} \frac{(-2)g^{3/2}}{E} g^+
    = \delta_{GR}\Rgr q^+_\perp - \frac{m\ddot q_\perp}{E}g^+,
\end{align*}
and thus
\begin{align*}
    &\Dgr (g^{3/2} q^+_\perp) 
    = \frac{3\sqrt{g}}{2} \Dgr g q^+_\perp
    + g^{3/2} \Dgr q^+_\perp
    = - \frac{3g^2}{E} \text{EL}_g q^+_\perp + g^{3/2} \delta_{GR} \Rgr q^+_\perp - g^{3/2} \frac{m\ddot q_\perp}{E}g^+\\
    &= - \frac{3g^2}{E} \delta_{GR}g^+ q^+_\perp
    +\frac{3g^2}{E} \delta_{GR}(g^+ q^+_\perp) - g^{3/2} \frac{m\ddot q_\perp}{E}g^+
    = - \frac{3g^2}{E} g^+ \delta_{GR}q^+_\perp 
    - g^{3/2} \frac{m\ddot q_\perp}{E}g^+\\
    &= \frac{3g^2}{E} g^+ \frac{m\ddot q_\perp}{\sqrt{g}}
    - g^{3/2} \frac{m\ddot q_\perp}{E}g^+
    = \frac{2m}{E}\ddot q_\perp g^{3/2} g^+.
\end{align*}
Since $\Dgr \ddot q_\perp = 0$, the computation of the higher powers of $\Dgrk (g^{3/2} q^+_\perp)$ becomes quite straightforward. We have
\begin{align*}
    \Dgrk (g^{3/2} q^+_\perp) = - (-1)^k \frac{2m}{E}\ddot q_\perp g^{3/2} g^+ \quad \text{for } k\geq1,
\end{align*}
which results in
\begin{align*}
    \chi^*_s (g^{3/2} q^+_\perp) 
    &= e^{s\mathcal{L}_{\Dgr}} (g^{3/2} q^+_\perp) 
    = \sum_{k\geq0} \frac{s^k}{k!} \Dgrk (g^{3/2} q^+_\perp)\\
    &= g^{3/2} q^+_\perp - \left(\sum_{k\geq1} \frac{(-s)^k}{k!}\right) \frac{2m}{E}\ddot q_\perp g^{3/2} g^+\\
    &= g^{3/2} q^+_\perp - (e^{-s}-1)  \frac{2m}{E}\ddot q_\perp g^{3/2} g^+,
\end{align*}
and as such
\begin{align*}
    \lim_{s\rightarrow\infty}\chi^*_s q^+_\perp = \eta^{3/2} \left(q^+_\perp 
    + \frac{2m}{E} g^+ \ddot  q_\perp ,\right)
    &= \chi^*q^+.   
\end{align*}
as desired. 
\end{proof}

\begin{lemma}
\label{lem:chiIdCohomologyJacobi1DGR}
The map $\chi^*$ is the identity in cohomology.
\end{lemma}

\begin{proof}
We need to show that \[h_\chi f = \int^\infty_0 e^{s\mathcal{L}_{\Dgr}}  \mathcal{L}_{\Rgr}  f \d s \] is well-defined on $\F_{GR}$. Note that $\{q,\xi,\xi^+,g^+\} \in \ker \Rgr $ and as such 
\begin{align*}
    &h_\chi q = 
    h_\chi \xi = 
    h_\chi \xi^+ = 
    h_\chi g^+ = 0.
\end{align*}
In the case of the metric field $g$ we compute
\begin{align*}
    h_\chi g &= \int^\infty_0e^{s\mathcal{L}_{\Dgr}}\Rgr g \,\mathrm{d}s
    = -\frac{2}{E} \int^\infty_0e^{s\mathcal{L}_{\Dgr}}(g^{3/2}g^+) \mathrm{d}s\\
    &= -\frac{2}{E} \int^\infty_0e^{-s} g^{3/2}g^+ \mathrm{d}s
    = -\frac{2g^{3/2}}{E} g^+.
\end{align*}
We now compute the action of the map $h_\chi$ on $q^+_\parallel$ and $q^+_\perp$. For the perpendicular part of $q^+$ we have
\begin{align*}
    h_\chi q^+_\perp 
    &= \int^\infty_0 e^{s\mathcal{L}_{\Dgr}} \Rgr q^+_\perp \mathrm{d}s
    = \int^\infty_0 e^{s\mathcal{L}_{\Dgr}} \left(\frac{3}{E}\sqrt{g}g^+ q^+_\perp\right) \mathrm{d}s\\
    &=\frac{3}{E} \int^\infty_0 (e^{s\mathcal{L}_{\Dgr}}g)^{1/2} e^{s\mathcal{L}_{\Dgr}}(g^{3/2}g^+)e^{s\mathcal{L}_{\Dgr}}(g^{3/2}q^+_\perp) (e^{s\mathcal{L}_{\Dgr}}g)^{-3} \mathrm{d}s\\
    &=\frac{3}{E} \int^\infty_0 (e^{s\mathcal{L}_{\Dgr}}g)^{-5/2} e^{-s}g^{3/2}g^+ \left(g^{3/2}q^+_\perp - (e^{-s}-1)\cancel{\frac{2m}{E} \ddot q_\perp g^{3/2}g^+} \right) \mathrm{d}s\\
    &=\frac{3}{E} g^3g^+q^+_\perp \int^\infty_0 \frac{e^{-s}}{\left[e^{-s}g + (1-e^{-s})\frac{T}{E}\right]^{5/2}} \mathrm{d}s.
\end{align*}
The integral yields
\begin{align*}
    I_\perp 
    &= \int^\infty_0 \frac{e^{-s}}{\left[e^{-s}g + (1-e^{-s})\frac{T}{E}\right]^{5/2}} \mathrm{d}s
    = \left. \frac{2}{3} \left(g-\frac{T}{E} \right)^{-1} \left[e^{-s}g + (1-e^{-s})\frac{T}{E}\right]^{-3/2}\right|^\infty_0\\
    &= \frac{2}{3} \left(g-\frac{T}{E} \right)^{-1} \left[\left(\frac{E}{T}\right)^{3/2} - \frac{1}{g^{3/2}}\right]
    = \frac{2}{3} \frac{E}{Tg^{3/2}} \frac{\eta^{3/2}-1}{\eta-1}\\
    &= \frac{2}{3} \frac{E}{Tg^{3/2}} \left(\frac{\eta}{\sqrt{\eta}+1} + 1\right).
\end{align*}
Which results in
\begin{align*}
    h_\chi q^+_\perp = \frac{2}{T} \left(\frac{\eta}{\sqrt{\eta}+1} + 1\right) g^{3/2} g^+q^+_\perp.
\end{align*}
Similarly we have
\begin{align*}
    h_\chi q^+_\parallel 
    &= \int^\infty_0 e^{s\mathcal{L}_{\Dgr}} \Rgr q^+_\parallel \mathrm{d}s
    = -\frac{3}{E}\int^\infty_0 e^{s\mathcal{L}_{\Dgr}} \left( \sqrt{g}\left(\frac{E}{2\sqrt{g}} - \frac{T}{2g^{3/2}}\right) \xi^+ \frac{m\dot q}{2T}\right) \mathrm{d}s\\
    &= \frac{3}{2} \frac{m\dot q}{2T} \int^\infty_0 e^{s\mathcal{L}_{\Dgr}} \left( \left(\frac{T}{Eg} - 1\right) \xi^+\right)  \mathrm{d}s\\
    &= \frac{3}{2} \frac{m\dot q}{2T} \int^\infty_0  \left(\frac{T}{E} (e^{s\mathcal{L}_{\Dgr}}g)^{-1} - 1\right) (e^{s\mathcal{L}_{\Dgr}}g)^{-3/2} \left[g^{3/2} \xi^+ -(e^{-2s} - 1)\frac{g^{3}}{E}\dot g^+ g^+\right] \mathrm{d}s\\
    &= \frac{3}{2} \frac{m\dot q}{2T} \int^\infty_0  \left(\frac{T}{E} - e^{s\mathcal{L}_{\Dgr}}g \right) (e^{s\mathcal{L}_{\Dgr}}g)^{-5/2} \left[g^{3/2} \xi^+ + \frac{g^{3}}{E}\dot g^+ g^+ - e^{-2s} \frac{g^{3}}{E}\dot g^+ g^+  \right] \mathrm{d}s\\
    &= \frac{3}{2} \frac{m\dot q}{2T} \left[g^{3/2} \xi^+ + \frac{g^{3}}{E}\dot g^+ g^+\right] I_1
    - \frac{3}{2} \frac{m\dot q}{2T} \frac{g^{3}}{E}\dot g^+ g^+ I_2.
\end{align*}
The integrals that we need to consider are
\begin{align*}
    I_1 
    &= \int^\infty_0 \left(\frac{T}{E} - e^{-s}g - (1-e^{-s})\frac{T}{E}\right) \left[e^{-s}g + (1-e^{-s})\frac{T}{E}\right]^{-5/2} \mathrm{d}s\\
    &= \left(\frac{T}{E}-g\right) \int^\infty_0 \frac{e^{-s}}{\left[e^{-s}g + (1-e^{-s})\frac{T}{E}\right]^{5/2}} \mathrm{d}s\\
    &= \frac{T}{E}(1-\eta) I_\perp = - \frac{2}{3} g^{-3/2} (\eta^{3/2}-1),
\end{align*}
and
\begin{align*}
    I_2 
    &= \int^\infty_0 \left(\frac{T}{E} - e^{-s}g - (1-e^{-s})\frac{T}{E}\right) \frac{e^{-2s}}{(e^{s\mathcal{L}_{\Dgr}}g)^{5/2}} \mathrm{d}s\\
    &= \left(\frac{T}{E}-g\right) \int^\infty_0  \frac{e^{-3s}}{(e^{s\mathcal{L}_{\Dgr}}g)^{5/2}} \mathrm{d}s
    = -\frac{2}{3} \int^\infty_0
    e^{-2s} \frac{\d}{\d s}
    \left[(e^{s\mathcal{L}_{\Dgr}}g)^{-3/2} \right]\mathrm{d}s\\
    &= -\frac{2}{3} \left. \frac{e^{-2s}}{(e^{s\mathcal{L}_{\Dgr}}g)^{3/2}} \right|^\infty_0
    -\frac{4}{3} \int^\infty_0
    \frac{e^{-2s}}{(e^{s\mathcal{L}_{\Dgr}}g)^{3/2}} \mathrm{d}s\\
    &= \frac{2}{3} g^{-3/2}
    -\frac{4}{3} \int^\infty_0
    e^{-s}\frac{\d}{\d s}
    \left[\frac{2}{\left(g-\frac{T}{E}\right)\sqrt{e^{s\mathcal{L}_{\Dgr}}g}}\right] \mathrm{d}s\\
    &= \frac{2}{3} g^{-3/2}
    -\frac{8}{3} 
    \left.\frac{e^{-s}}{\left(g-\frac{T}{E}\right)\sqrt{e^{s\mathcal{L}_{\Dgr}}g}}\right|^\infty_0
    - \frac{8}{3} \int^\infty_0 \frac{e^{-s}}{\left(g-\frac{T}{E}\right)\sqrt{e^{s\mathcal{L}_{\Dgr}}g}} \d s\\
    &= \frac{2}{3} g^{-3/2} 
    + \frac{8}{3} \left(g-\frac{T}{E}\right)^{-1} g^{-1/2}  
    + \frac{16}{3} \left(g-\frac{T}{E}\right)^{-2} \left.\sqrt{e^{s\mathcal{L}_{\Dgr}}g}\right|^\infty_0\\
    &= \frac{2}{3} g^{-3/2} 
    + \frac{8}{3} \left(g-\frac{T}{E}\right)^{-1} g^{-1/2}  
    + \frac{16}{3} \left(g-\frac{T}{E}\right)^{-2} \left(\sqrt\frac{T}{E} - \sqrt{g}\right)\\
    &=\left(\frac{E}{T}\right)^{3/2}
    \left[
    \frac{2}{3} \frac{1}{\eta^{3/2}} 
    + \frac{8}{3} \frac{1}{\left(\eta-1\right)\sqrt{\eta}}
    + \frac{16}{3} \frac{1 - \sqrt{\eta}}{\left(\eta-1\right)^{2}} 
    \right]\\
    &= -\frac{2}{3}g^{-3/2}\frac{3 \eta - 2 \sqrt{\eta} - 1}{(\sqrt{\eta} + 1)^2}.
\end{align*}
Where on the third line we used the integral in $I_\perp$ and similarly for the other integrals.
Gathering everything results in
\begin{align*}
    h_\chi q^+_\parallel=  (1-\eta^{3/2}) \left[\xi^+ + \frac{g^{3/2}}{E}\dot g^+ g^+\right] \frac{m\dot q}{2T}
    + \frac{3 \eta - 2 \sqrt{\eta} - 1}{(\sqrt{\eta}+1)^2} \frac{g^{3/2}}{E}\dot g^+ g^+ \frac{m\dot q}{2T}. 
\end{align*}
\end{proof}

\section*{Competing Interests} The authors have no competing interests to declare that are relevant to the content of this article.

\printbibliography

\end{document}